%% file: main.tex
\documentclass[sigconf]{acmart}
\usepackage{multirow}
\usepackage{subfigure}
\usepackage[linesnumbered,ruled, vlined]{algorithm2e}
\usepackage{amsmath, amsthm,amsfonts, bm}
\usepackage{array}
\usepackage{balance}
\usepackage{enumitem}
\usepackage{soul}
\usepackage{tablefootnote}
\usepackage{threeparttable}
\usepackage{footnote}
\usepackage{xcolor}
\makesavenoteenv{tabular}
\makesavenoteenv{table}
\usepackage{url}

\newcolumntype{L}[1]{>{\raggedright\let\newline\\\arraybackslash\hspace{0pt}}m{#1}}
\newcolumntype{C}[1]{>{\centering\let\newline\\\arraybackslash\hspace{0pt}}m{#1}}
\newcolumntype{R}[1]{>{\raggedleft\let\newline\\\arraybackslash\hspace{0pt}}m{#1}}

\newtheorem{definition}{Definition}[section]
\newtheorem{problem}{Problem}[section]

\AtBeginDocument{%
  \providecommand\BibTeX{{%
    \normalfont B\kern-0.5em{\scshape i\kern-0.25em b}\kern-0.8em\TeX}}}

\newcommand{\kaixin}{}
\newcommand{\cheng}{}

\newcommand{\yui}{}
\newcommand{\Yui}{\color{black}}
\newcommand{\chengB}{}
\newcommand{\chengC}{\color{black}}

\newcommand{\Revise}{\color{black}}

\acmConference{SIGMOD}
\acmYear{2024} 

\begin{document}
\sloppy

\renewcommand\footnotetextcopyrightpermission[1]{} 



\title{Efficient $k$-Clique Listing: An Edge-Oriented Branching Strategy}


\author{Kaixin Wang}
\affiliation{%
  \institution{Nanyang Technological University}
  \country{Singapore}
}
\email{kaixin.wang@ntu.edu.sg}

\author{Kaiqiang Yu}
\authornote{Kaiqiang Yu is the corresponding author.}
\affiliation{%
  \institution{Nanyang Technological University}
  \country{Singapore}
}
\email{kaiqiang002@e.ntu.edu.sg}

\author{Cheng Long}
\affiliation{%
  \institution{Nanyang Technological University}
  \country{Singapore}
}
\email{c.long@ntu.edu.sg}

\input{abstract}



\maketitle

\input{intro}
\input{problem}
\input{VBBkC}

\input{EBBkC}
\input{early-stop}

\input{exp}

\input{related}

\balance
\input{conclusion}

\input{appendix}

\balance
\clearpage
\bibliographystyle{ACM-Reference-Format}
\bibliography{main}



\end{document}

%% file: abstract.tex
\begin{abstract} 
    $k$-clique listing is a vital graph mining operator with diverse applications in various networks. The state-of-the-art algorithms all adopt a branch-and-bound (BB) framework with a vertex-oriented branching strategy (called \texttt{VBBkC}), which forms a sub-branch by expanding a partial $k$-clique with a \emph{vertex}. These algorithms have the time complexity of $O(k\cdot m \cdot (\delta/2)^{k-2})$, where $m$ is the number of edges in the graph and $\delta$ is the degeneracy of the graph. In this paper, we propose a BB framework with a new \emph{edge-oriented branching} (called \texttt{EBBkC}), which forms a sub-branch by expanding a partial $k$-clique with two vertices that connect each other (which correspond to an \emph{edge}). We explore various edge orderings for \texttt{EBBkC} such that it achieves a time complexity of $O( m\cdot \delta + k\cdot m \cdot (\tau/2)^{k-2})$, where $\tau$ is an integer related to the maximum truss number of the graph and we have $\tau < \delta$. The time complexity of \texttt{EBBkC} is better than that of \texttt{VBBkC} algorithms for $k>3$ since both $O(m\cdot \delta)$ and $O(k\cdot m \cdot (\tau/2)^{k-2})$ are bounded by $O(k\cdot m \cdot (\delta/2)^{k-2})$. Furthermore, we develop specialized algorithms for sub-branches on dense graphs so that we can early-terminate them and apply the specialized algorithms. We conduct extensive experiments on 19 real graphs, and the results show that our newly developed \texttt{EBBkC} based algorithms with the early termination technique consistently and largely outperform the state-of-the-art (\texttt{VBBkC} based) algorithms. 
\end{abstract}

%% file: intro.tex
\section{Introduction}
\label{sec:intro}

%
%
Given a graph $G$, a $k$-clique is subgraph of $G$ with $k$ vertices such that each pair of vertices {\cheng inside} are connected 
\cite{chiba1985arboricity}. 
%
{\cheng 
$k$-clique listing, {\chengB which is to list all $k$-cliques in a graph,} is a fundamental graph mining operator that plays a crucial role in various data mining applications across different networks, including social networks, mobile networks, {\chengB Web} networks, and biological networks. Some significant applications include the detection of overlapping communities in social networks~\cite{palla2015uncovering}, identifying $k$-clique communities in mobile networks~\cite{gregori2012parallel, palla2005uncovering, hui2008human}, detecting link spams in {\chengB Web} networks~\cite{jayanthi2012clique, saito2007large}, and discovering groups of functionally related proteins (known as modules) in gene association networks~\cite{adamcsek2006cfinder}. Moreover, $k$-clique listing serves as a key component for several other tasks, such as finding large near cliques~\cite{tsourakakis2015k}, uncovering the hierarchical structure of dense subgraphs~\cite{sariyuce2015finding}, exploring $k$-clique densest subgraphs~\cite{tsourakakis2013denser,fang2019efficient}, identifying stories in social media~\cite{angel2014dense}, and detecting latent higher-order organization in real-world networks~\cite{benson2016higher}. For more detailed information on how $k$-clique listing is applied in these contexts, please refer to references~\cite{li2020ordering} and~\cite{yuan2022efficient}.}


{\cheng Quite a few algorithms have been proposed 
for listing $k$-cliques~\cite{chiba1985arboricity, danisch2018listing, li2020ordering, yuan2022efficient}. The majority of these approaches adopt a \emph{branch-and-bound} (BB) framework, which involves recursively dividing the problem of listing all $k$-cliques in graph $G$ into smaller sub-problems of listing smaller cliques in $G$ through \emph{branching} operations~\cite{danisch2018listing, li2020ordering, yuan2022efficient}. This process continues until each sub-problem can be trivially solved. The underlying principle behind these methods is the observation that a $k$-clique can be constructed by merging two smaller cliques: a clique $S$ and an $l$-clique, where $|S| + l = k$. A branch $B$ is represented as a triplet $(S, g, l)$, where $S$ denotes a previously found clique with $|S| < k$, $g$ represents a subgraph induced by vertices that connect each vertex in $S$, and $l$ corresponds to $k - |S|$. Essentially, branch $B$ encompasses all $k$-cliques, each comprising of $S$ and an $l$-clique in $g$. To enumerate all $k$-cliques within branch $B$, a set of sub-branches is created through branching operations. Each sub-branch expands the set $S$ by adding one vertex from $g$, updates the graph $g$ accordingly (by removing vertices that {are disconnected from} the newly added vertex from $S$), and decrements $l$ by 1. This recursive process continues until $l$ for a branch reduces to 2, at which point the $l$-cliques (corresponding to edges) can be trivially listed within $g$. The original $k$-clique listing problem on graph $G$ can be solved by initiating the branch $(S, g, l)$ with $S = \emptyset$, $g = G$, and $l = k$. The branching step employed in these existing methods is referred to as \emph{vertex-oriented branching} since each sub-branch is formed by adding a vertex to the set $S$. We term the \underline{v}ertex-oriented branching \underline{BB} framework for \underline{$k$}-\underline{c}lique listing as \texttt{VBBkC}.

Existing research has focused on leveraging vertex information within the graph $g$ to enhance performance by generating sub-branches with smaller graphs and pruning sub-branches more effectively. Specifically, when a new vertex $v_i$ is added, the resulting sub-branch $B_i$ can reduce the graph $g$ to a smaller size by considering only the neighbors of $v_i$ (e.g., removing vertices that are not adjacent to $v_i$). Consequently, branch $B_i$ can be pruned if $v_i$ has fewer than $(l-1)$ neighbors in $g$. Furthermore, by carefully specifying the ordering of vertices in $g$ during the branching process, it becomes possible to generate a group of sub-branches with compact graphs. The state-of-the-art \texttt{VBBkC} algorithms~\cite{li2020ordering, yuan2022efficient} achieve a time complexity of $O(km(\delta/2)^{k-2})$, where $m$ represents the number of edges in the graph, and $\delta$ denotes the degeneracy\footnote{{It is defined to be the maximum value of $k$ such that there exists a non-empty $k$-core in a graph, where a $k$-core is a graph where each vertex has the degree at least $k$.}} of the graph.

In this paper, we propose to construct a branch by simultaneously including {\chengB \emph{two} vertices that connect to each other (i.e., an \emph{edge})} from $g$ into $S$. Note that these two vertices must be connected, as only then they can form a larger partial $k$-clique together with $S$. This strategy allows for the consideration of additional information (i.e., an edge or two vertices instead of a single vertex) to facilitate the formation of sub-branches with smaller graphs and pruning. When a new edge is added, the resulting branch can reduce the graph to one induced by the common neighbors of the two vertices, which is smaller compared to the graph generated by vertex-oriented branching. To achieve this, we explore an edge ordering technique based on truss decomposition \cite{wang2012truss}, which we refer to as the \emph{truss-based edge ordering}. This ordering aids in creating sub-branches with smaller graphs than those by vertex-oriented branching. Additionally, more branches can be pruned based on the information derived from the added edge. For instance, a branch can be pruned if the vertices within the newly added edge have fewer than $(l-2)$ common neighbors in $g$. We term this branching strategy as \emph{edge-oriented branching}. Consequently, the \underline{e}dge-oriented branching-based \underline{BB} framework for \underline{$k$}-\underline{c}lique listing is denoted as \texttt{EBBkC}. Our \texttt{EBBkC} algorithm, combined with the proposed edge ordering, exhibits a time complexity of {\Revise $O(\delta m + km(\tau/2)^{k-2})$}, where $\tau$ represents the truss number\footnote{The maximum truss number of a graph defined in \cite{wang2012truss}, denoted by $k_{\max}$, has the following relationship with $\tau$: $k_{\max} = \tau+2$.} of the graph. We formally prove that $\tau < \delta$, signifying that our \texttt{EBBkC} algorithm possesses a time complexity that is {\Revise better than that of \texttt{VBBkC} algorithms \cite{li2020ordering, yuan2022efficient} since both $O(\delta m)$ and $O(km(\tau/2)^{k-2})$ are bounded by $O(km(\delta/2)^{k-2})$ when $k>3$.}  

It is important to note that although a single branching step in our edge-oriented branching (i.e., including an edge) can be seen as two branching steps in the existing vertex-oriented branching (i.e., including two vertices of an edge via two steps), there exists a significant distinction between our \texttt{EBBkC} and \texttt{VBBkC} frameworks. In \texttt{EBBkC}, we have the flexibility to explore \emph{arbitrary} edge orderings for each branching step, whereas \texttt{VBBkC} is inherently \emph{constrained} by the chosen vertex ordering. For instance, once a vertex ordering is established for vertex-oriented branching, with vertex $v_i$ appearing before $v_j$, the edges incident to $v_i$ would precede those incident to $v_j$ in the corresponding edge-oriented branching. Consequently, the existing \texttt{VBBkC} framework is encompassed by our \texttt{EBBkC} framework. In other words, for any instance of \texttt{VBBkC} with a given vertex ordering, there exists an edge ordering such that the corresponding \texttt{EBBkC} instance is equivalent to the \texttt{VBBkC} instance, \emph{but not vice versa}. This elucidates why \texttt{EBBkC}, when based on certain edge orderings, achieves a superior time complexity compared to \texttt{VBBkC}.

To further enhance the efficiency of BB frameworks, we develop an \emph{early termination} technique, which is based on two key observations. Firstly, if the graph $g$ within a branch $(S, g, l)$ is either a clique or a $2$-plex\footnote{A $t$-plex is a graph where each vertex inside {has at most $t$ non-neighbors including itself.}}, we can efficiently list $l$-cliques in $g$ using a combinatorial approach. For instance, in the case of a clique, we can directly enumerate all possible sets of $l$ vertices in $g$. Secondly, if the graph $g$ is a $t$-plex (with $t\geq 3$), the branching process based on $g$ can be converted to a procedure on its inverse graph, denoted as $g_{inv}$\footnote{{The inverse graph $g_{inv}$ has the same set of vertices as $g$, with an edge between two vertices in $g_{inv}$ if and only if they are {disconnected from each other} in $g$.}}. Since $g$ is dense, $g_{inv}$ would be sparse, and the converted procedure operates more rapidly. Therefore, during the recursive branching process, we can employ early termination at a branch $(S, g, l)$ if $g$ transforms into a $t$-plex, utilizing efficient algorithms to list $l$-cliques within $g$. We note that the early termination technique is applicable to all BB frameworks, including our \texttt{EBBkC} framework, without impacting the worst-case time complexity of the BB frameworks.}

\smallskip\noindent\textbf{Contributions.}
We summarize our contributions as follows.

\begin{itemize}[leftmargin=*]
    \item 
    We propose a new BB framework for $k$-clique listing problem, namely \texttt{EBBkC}, which is based on an edge-oriented branching strategy. {\cheng We further explore different edge orderings for \texttt{EBBkC} such that it achieves a {\Revise better} time complexity than that of the state-of-the-art \texttt{VBBkC} {\Yui for $k>3$}, i.e., the former is {\Revise $O(\delta m + km(\tau/2)^{k-2})$} and the latter is $O(km(\delta/2)^{k-2})$, where $\tau$ is a number related to the maximum truss number of the graph and $\delta$ is the degeneracy of the graph and we have $\tau < \delta$. (Section~\ref{sec:edge-oriented-bb-framework})}


    
    \item 
    {\cheng We further {\chengB develop} an early termination technique for boosting the efficiency of branch-and-bound frameworks including \texttt{EBBkC}, i.e., for branches of listing $l$-cliques in a dense graph (e.g., a $t$-plex), we develop more efficient algorithms based on combinatorial approaches (for a clique and a $2$-plex) and conduct the branching process on its inverse graph (for a $t$-plex with $t\ge 3$), which would be {\chengB faster}. (Section~\ref{sec:early-termination})}
    
    \item We conduct extensive experiments on 19 real graphs, and the results show that 
    {\cheng our \texttt{EBBkC} based algorithms with the early termination technique consistently and largely outperform the state-of-the-art (\texttt{VBBkC} based) algorithms.}
    (Section~\ref{sec:exp})
    
\end{itemize}

The rest of the paper is organized as follows. Section~\ref{sec:problem} reviews the problem and presents some preliminaries. Section~\ref{sec:vbbkc} summarizes the {\chengB existing} vertex-oriented branching-based BB framework. Section~\ref{sec:related} reviews the related work and Section~\ref{sec:conclusion} concludes the paper.



%% file: problem.tex
\section{Problem and Preliminaries}
\label{sec:problem}

We consider an \emph{unweighted} and \emph{undirected} simple graph $G=(V,E)$, where $V$ is the set of vertices and $E$ is the set of edges. We denote by $n=|V|$ and $m=|E|$ the cardinalities of $V$ and $E$, respectively. Given $u,v\in V$, both $(u,v)$ and $(v,u)$ denote the undirected edge between $u$ and $v$. Given $u\in V$, we denote by $N(u,G)$ the set of neighbors of $u$ in $G$, i.e., $N(u,G)=\{v\in V\mid (u,v)\in E\}$ and define $d(u, G)=|N(u,G)|$.
Given $V_{sub}\subseteq V$, we use $N(V_{sub},G)$ to denote the common neighbors of vertices in $V_{sub}$, i.e., $N(V_{sub},G)=\{v\in V\mid \forall u\in V_{sub},\ (v,u)\in E\}$.

Given $V_{sub}\subseteq V$, we denote by $G[V_{sub}]$ the subgraph of $G$ induced by $V_{sub}$, i.e., $G[V_{sub}]$ includes the set of vertices $V_{sub}$ and the set of edges $\{(u,v)\in E\mid u, v\in V_{sub}\}$.
Given $E_{sub}\subseteq E$, we denote by $G[E_{sub}]$ the subgraph of $G$ induced by $E_{sub}$, i.e., $G[E_{sub}]$ includes the set of edges $E_{sub}$ and the set of vertices $\{ v\in V \mid (v, \cdot) \in E_{sub} \}$.
Let $g$ be a subgraph of $G$ induced by either a vertex subset of $V$ or an edge subset of $E$. We denote by $V(g)$ and $E(g)$ its set of vertices and its set of edges, respectively.

In this paper, we focus on a {\cheng widely-used} cohesive graph structure, namely $k$-clique~\cite{erdos1935combinatorial}, which is defined formally as below.

\begin{definition}[$k$-clique~\cite{erdos1935combinatorial}]
Given a positive integer $k$, a subgraph $g$ is said to be a $k$-clique if and only if it has $k$ vertices and has an edge between every pair of vertices, i.e., $|V(g)|=k$ and $E(g)=\{(u,v)\mid u,v\in V(g), {\cheng u\neq v} \}$.
\end{definition}

We note that 1-clique ($k=1$) and 2-clique ($k=2$) correspond to single vertex and single edge, respectively, {\cheng which are basic elements of a graph.} Therefore, we focus on those $k$-cliques with $k$ at least 3 (note that 3-clique is widely known as triangle and has found many applications \cite{ latapy2008main, ortmann2014triangle}).
We now formulate the problem studied in this paper as follows.

\begin{problem}[$k$-clique listing~\cite{chiba1985arboricity}]
Given a graph $G=(V,E)$ and a positive integer $k\geq 3$, the {$k$-clique listing} problem aims to find all $k$-cliques in $G$.
\end{problem}

\smallskip
\noindent\textbf{Hardness.} The $k$-clique listing problem is a hard problem since the decision problem of determining whether a graph contains a $k$-clique is NP-hard~\cite{karp2010reducibility} {\cheng and this problem} can be solved by listing all $k$-cliques and returning true if any $k$-clique is listed.

\smallskip
\noindent\textbf{Remark.} The problem of listing all 3-cliques {\cheng (i.e., $k = 3$)}, known as \emph{triangle listing problem}, has been widely studied~\cite{latapy2008main,ortmann2014triangle}. There are many efficient algorithms proposed for triangle listing, which run in \emph{polynomial} time. We remark that these algorithms cannot be used to solve the general $k$-clique listing problem.







%% file: VBBkC.tex
\section{The Branch-and-Bound Framework of Existing Algorithms: \texttt{VBBkC}}
\label{sec:vbbkc}

Many algorithms {\cheng have been} proposed for listing $k$-cliques in the literature~\cite{chiba1985arboricity, danisch2018listing, li2020ordering, yuan2022efficient}.
{\cheng Most of them adopt a \emph{branch-and-bound} (BB) framework}, which recursively partitions the problem instance ({\cheng of} listing all $k$-cliques in $G$) into several sub-problem instances {\cheng (of listing smaller cliques in $G$)} via \emph{branching} until each of them can be solved trivially~\cite{danisch2018listing, li2020ordering, yuan2022efficient}. 
{
\cheng 
The rationale behind these methods is that a $k$-clique can be constructed by merging two smaller cliques, namely a clique $S$ and an $l$-clique with $|S| + l = k$. Specifically, a branch $B$ can be represented as a triplet $(S, g, l)$, where 
\begin{itemize}[leftmargin=*]
\item \textbf{set $S$} induces a clique found so far with $|S| < k$, 
\item \textbf{subgraph $g$} is one induced by vertices that connect each vertex in $S$, and 
\item \textbf{integer $l$} is equal to $k - |S|$. 
\end{itemize}
Essentially, branch $B$ covers all $k$-cliques, each consisting of $S$ and an $l$-clique in $g$. To list all $k$-cliques under branch $B$, it creates a group of sub-branches via a branching step such that for each sub-branch, the set $S$ is expanded with one vertex from $g$, the graph $g$ is updated accordingly (by removing those vertices that are not adjacent to the vertex included in $S$), and $l$ is decremented by 1. The recursive process continues until when the $l$ for a branch reduces to $2$, for which the $l$-cliques (which correspond to edges)  can be listed trivially in $g$. The original $k$-clique listing problem on graph $G$ can be solved by starting with the branch $(S, g, l)$ with $S = \emptyset$, $g = G$ and $l = k$. We call the branching step involved in these existing methods \emph{vertex-oriented branching} since each sub-branch is formed by including a \emph{vertex} to the set $S$.
}

\begin{algorithm}[t]
\small
\caption{The vertex-oriented {\chengB branching-based} BB framework: \texttt{VBBkC}}
\label{alg:VBBkC}
\KwIn{A graph $G=(V,E)$ and an integer $k\geq 3$}
\KwOut{All $k$-cliques within $G$}
\SetKwFunction{List}{\textsf{VBBkC\_Rec}}
\SetKwProg{Fn}{Procedure}{}{}
\List{$\emptyset, G, k$}\;
\Fn{\List{$S, g, l$}} {
    \tcc{Pruning}
    \lIf{$|V(g)|<l$}{\textbf{return}}
    \tcc{Termination when $l=2$}
    \If{$l=2$}{
       \lFor{each edge $(u,v)$ in $g$}{
                \textbf{Output} a $k$-clique $S\cup\{u,v\}$
       }
       \textbf{return}\;
    }
    \tcc{Branching when $l\geq 3$}
    \For{each vertex $v_i\in V(g)$ based on a given vertex ordering}{
        Create branch $B_i=(S_i,g_i,l_i)$ based on Eq.~(\ref{eq:vertex-oriented-branching})\;
        \List{$S_i, g_i, l_i$}\;
    }
}
\end{algorithm}

Consider the branching step at a branch $B=(S,g,l)$.
Let $\langle v_1,v_2,\cdots,v_{|V(g)|} \rangle$ be an arbitrary ordering of vertices in $g$. The branching step would produce $|V(g)|$ new sub-branches from branch $B$. The $i$-th sub-branch, denoted by $B_i=(S_i,g_i,l_i)$, includes $v_i$ to $S$ and excludes $\{v_1,v_2,\cdots,v_{i-1}\}$ (and also those that are not adjacent to $v_i$).
Formally, for $1\leq i\leq |V(g)|$, we have
\begin{equation}
    \label{eq:vertex-oriented-branching}
    S_i=S\cup \{v_i\},\quad g_i=\widehat{g}_i[N(v_i,\widehat{g}_i)],\quad l_i=l-1,
\end{equation}
where {\cheng $\widehat{g}_i$ is a subgraph of $g$ induced by the set of vertices $\{v_{i}, v_{i+1}, \cdots, v_{|V(g)|}\}$, i.e., $\widehat{g}_i = g[v_i,\cdots,v_{|V(g)|}]$}. The branching can be explained by a recursive binary partition process, as shown in Figure~\ref{subfig:bbkc}. 
Specifically, it first divides the current branch $B$ into two sub-branches based on $v_1$: one branch {\cheng moves} $v_1$ from $g$ to $S$ (this is the branch $B_1$ which will list those $k$-cliques in $B$ that include $v_1$), {\cheng and} the other removes $v_1$ from $g$ 
({\cheng so that it} will list others in $B$ that exclude $v_1$). 
For branch $B_1$, it also removes from $g$ those vertices that are not adjacent to $v_1$ since they cannot form any $k$-clique with $v_1$. In summary, we have $g_1=\widehat{g}_1[N(v_1,\widehat{g}_1)]$.
Then, it recursively divides the latter into two new sub-branches: one branch {\cheng moves} $v_2$ from $\widehat{g}_2$ to $S$ (this is the branch $B_2$ which will list those $k$-cliques in $B$ that exclude $v_1$ and include $v_2$), and the other removes $v_2$ from $\widehat{g}_2$ 
({\cheng so that it} will list others in $B$ that exclude $\{v_1,v_2\}$).
It continues the process, until the last branch $B_{|V(g)|}$ is formed. 
{\cheng In summary,}
branch $B_i$ will list those $k$-cliques in $B$ that include $v_i$ and exclude $\{v_1,\cdots,v_{i-1}\}$. Consequently, all $k$-cliques in $B$ will be listed exactly once after branching.

{\cheng We call the \underline{v}ertex-oriented branching-based \underline{BB} framework for \underline{$k$}-\underline{c}lique listing \texttt{VBBkC}}. We present its pseudo-code in Algorithm~\ref{alg:VBBkC}. In particular, when $l=2$, a branch $(S,g,l)$ 
can be terminated by listing each of edges in $g$ together with $S$ (lines 4-6). We remark that Algorithm~\ref{alg:VBBkC} presents the algorithmic idea only while the detailed implementations (e.g., the data structures used for representing a branch) 
{\cheng often}
vary in existing algorithms~\cite{danisch2018listing, li2020ordering, yuan2022efficient}.
%
%
{\cheng Different variants of \texttt{VBBkC} have different time complexities. The state-of-the-art algorithms, including \texttt{DDegCol}~\cite{li2020ordering}, \texttt{DDegree}~\cite{li2020ordering}, \texttt{BitCol}~\cite{yuan2022efficient} and \texttt{SDegree}~\cite{yuan2022efficient}, all share the time complexity of $O(km(\delta/2)^{k-2})$, where $\delta$ is the degeneracy of the graph.} Some more details of variants of \texttt{VBBkC} will be {\chengB provided} in the related work (Section~\ref{sec:related}).

\begin{figure}[t]
\vspace{-4mm}
\subfigure[Branching in \texttt{VBBkC} framework. The notation ``$+$'' means to include a vertex by adding it $S$ and ``$-$'' means to exclude a vertex by removing it from the graph $g$.]{
    \label{subfig:bbkc}
    \includegraphics[width=0.47\textwidth]{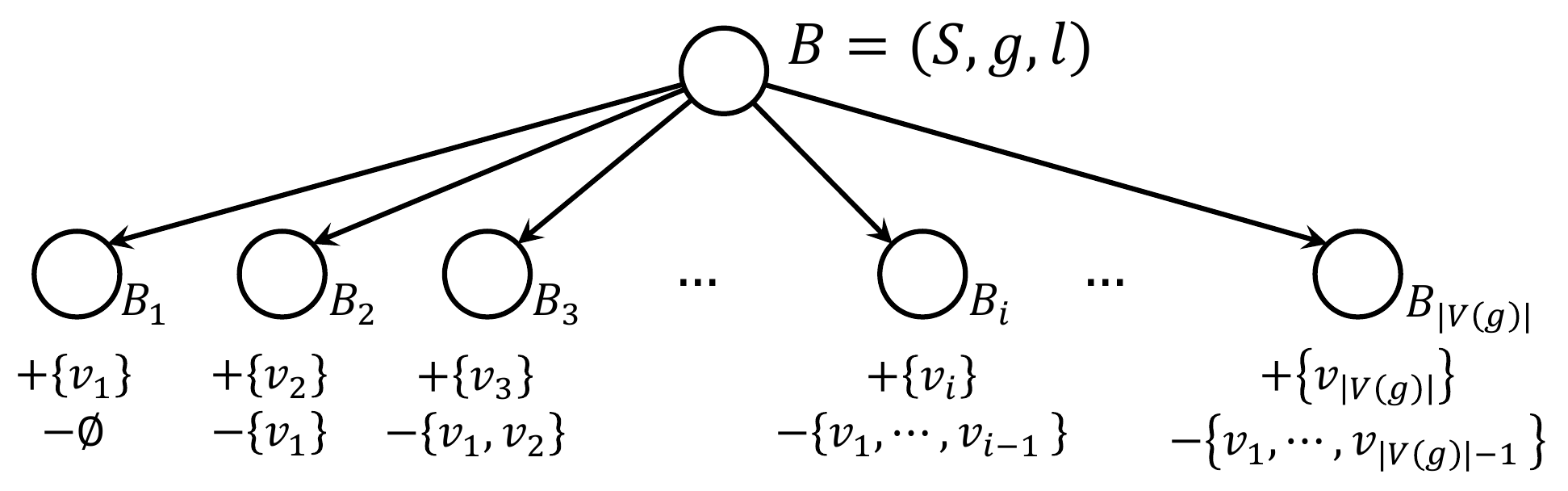}}
\subfigure[Branching in \texttt{EBBkC} framework. The notation ``$+$'' means to include two {\chengB vertices incident to} an edge by adding them to $S$ and ``$-$'' means to exclude an edge by removing it from the graph $g$.]{
    \label{subfig:ebbkc}
    \includegraphics[width=0.47\textwidth]{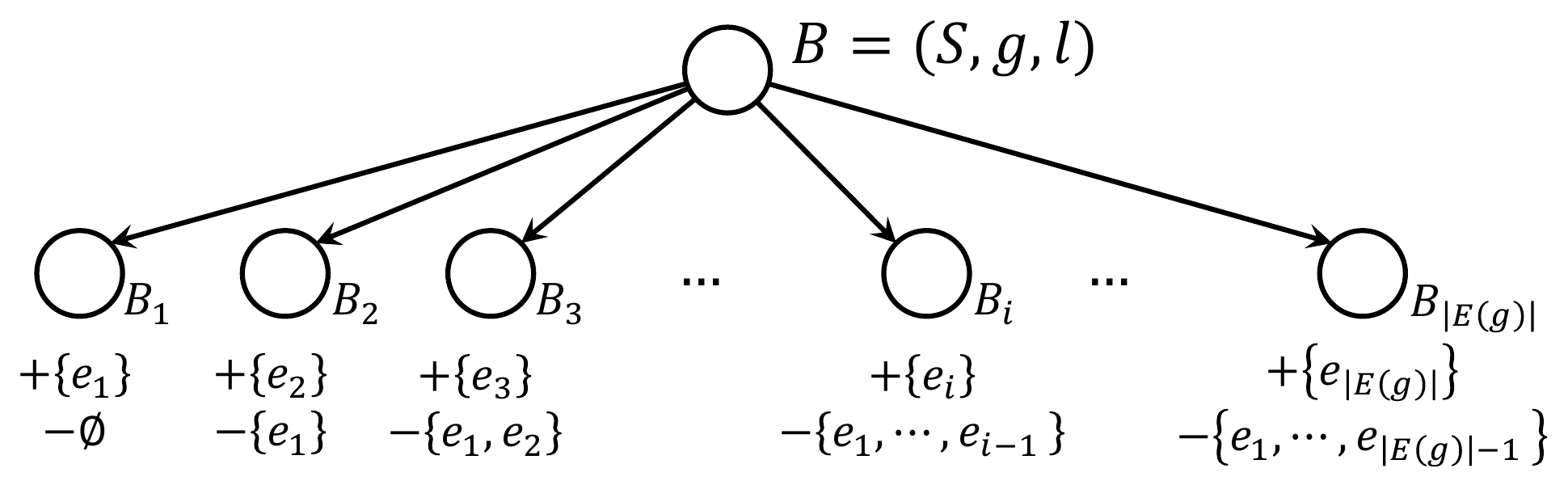}}
\vspace{-5mm}
\caption{Illustration of \texttt{VBBkC} and \texttt{EBBkC}.}
\label{fig:framework}
\vspace{-3mm}
\end{figure}

%% file: EBBkC.tex
\section{A New Branch-and-Bound Framework: \texttt{EBBkC}}

\label{sec:edge-oriented-bb-framework}


\subsection{{\cheng Motivation and Overview of \texttt{EBBkC}}}
\label{subsec:motivation}

Recall that for a branch $B = (S, g, l)$, the vertex-oriented branching forms a sub-branch by moving one vertex $v_i$ from $g$ to $S$. To improve the performance, existing studies consider the information of each vertex in $g$ towards pruning more sub-branches and/or forming sub-branches with smaller graph instances. Specifically, with the newly added vertex $v_i$, the produced sub-branch $B_i$ can shrink the graph instance $g$ as a smaller one induced by the neighbors of $v_i$ (e.g., removing those vertices that are not adjacent to $v_i$). As a result, one can prune the branch $B_i$ if $v_i$ has less than $(l-1)$ neighbors in $g$. In addition, one can produce a group of sub-branches with small graph instances by specifying the ordering of vertices in $g$ for branching.

In this paper, we propose to form a branch by \emph{moving two vertices that connect with each other} (correspondingly, an edge) from $g$ to $S$ at the same time. Note that the two vertices are required to be connected since otherwise they will not form a larger partial $k$-clique with $S$. 
{\cheng The intuition} is that it {\cheng would allow} us to consider more information (i.e., an edge {\cheng or two vertices instead of a single vertex}) towards pruning more sub-branches and/or forming sub-branches with smaller graph instances. Specifically, with the newly added edge, the produced branch can shrink the graph instance as the one induced by the common neighbors of two vertices of the edge, which is smaller than that produced by the vertex-oriented branching (details {\chengB can be found} in Section~\ref{subsec:EBBkC-T}).
%
{\cheng In addition}, we can prune more branches {\cheng based on} the information of {\cheng the added edge}, e.g., a produced branch can be pruned if the vertices in the newly added edge have less than $(l-2)$ common neighbors in $g$ (details {\chengB can be found} in Section~\ref{subsec:EBBkC-C}).
We call the above branching strategy \emph{edge-oriented branching}, which we introduce as follows. 

\begin{algorithm}[t]
\small
\caption{The edge-oriented {\chengB branching-based} BB framework: \texttt{EBBkC}}
\label{alg:EBBkC}
\KwIn{A graph $G=(V,E)$ and an integer $k\geq 3$}
\KwOut{All $k$-cliques within $G$}
\SetKwFunction{List}{\textsf{EBBkC\_Rec}}
\SetKwProg{Fn}{Procedure}{}{}
\List{$\emptyset, G, k$}\;
\Fn{\List{$S, g, l$}} {
    \tcc{Pruning}
    \lIf{$|V(g)|<l$}{\textbf{return}}
    \tcc{Termination when $l=1$ or $l=2$}
    \uIf{$l=1$}{
       \lFor{each vertex $v$ in $g$}{
                \textbf{Output} a $k$-clique $S\cup\{v\}$
       }
       \textbf{return}\;
    }\ElseIf{$l=2$}{
       \lFor{each edge $(u,v)$ in $g$}{
                \textbf{Output} a $k$-clique $S\cup\{u,v\}$
       }
       \textbf{return}\;
    }
    \tcc{Branching when $l\geq 3$}
    \For{each edge $e_i\in E(g)$ based on a given edge ordering}{
        Create branch $B_i=(S_i,g_i,l_i)$ based on Eq.~(\ref{eq:edge-oriented-branching})\;
        \List{$S_i, g_i, l_i$}\;
    }
}
\end{algorithm}

Consider the branching step at a branch $B=(S,g,l)$. Let $\langle e_1,e_2,\cdots,e_{|E(g)|} \rangle$ be an \emph{arbitrary} ordering of edges in $g$. Then, the branching step would produce $|E(g)|$ new sub-branches from $B$. The $i$-{th} branch, denoted by $B_i=(S_i,g_i,l_i)$, includes to $S$ 
{\kaixin{{\chengB the two vertices incident to} the edge $e_i$, i.e., $V(e_i)$,}} 
and excludes from $g$ those edges in $\{e_1,e_2,\cdots,e_{i-1}\}$ (and those vertices that {are disconnected from} the vertices incident to $e_i$). Formally, for $1\leq i \leq |E(g)|$, we have
\begin{equation}
\label{eq:edge-oriented-branching}
    S_i=S\cup V(e_i),\quad     g_i=\overline{g}_i[N(V(e_i),\overline{g}_{i})],\quad l_i=l-2,
\end{equation}
{\cheng where $\overline{g}_{i}$ is a subgraph of $g$ induced by the set of edges $\{e_i, e_{i+1}, \cdots, e_{|E(g)|}\}$, i.e., $\overline{g}_{i} = g[e_i,\cdots,e_{|E(g)|}]$.}
{\kaixin 
We note that (1) $N(V(e_i),\overline{g}_{i})$ is to filter out those vertices that are not adjacent to
{\chengB the two vertices incident to}
$e_i$ since they cannot form any $k$-cliques with $e_i$ and (2) $\overline{g}_i[N(V(e_i),\overline{g}_{i})]$ is the graph instance for the sub-branch induced by the vertex set $N(V(e_i),\overline{g}_{i})$ in $\overline{g}_{i}$. }

%
The edge-oriented branching also corresponds to a recursive binary partition process, as illustrated in Figure~\ref{subfig:ebbkc}. Specifically, it first divides branch $B$ into two sub-branches based on $e_1$: one moves $e_1$ from $g$ to $S$ (this is the branch $B_1$ which will list those $k$-cliques in $B$ that include edge $e_1$), and the other removes $e_1$ from $g$
(so that it will list others that exclude $e_1$). For branch $B_1$, it also removes from $g$ those vertices that are not adjacent to the vertices in $V(e_1)$ (i.e., $g_1=\overline{g}_1[N(V(e_1),\overline{g}_{1})]$) since they cannot form any $k$-clique with $e_1$. Then, it recursively divides the latter into two new sub-branches: one moves $e_2$ from $\overline{g}_2$ to $S$ (this the branch $B_2$ which will list those $k$-cliques in $B$ that exclude $e_1$ and include $e_2$), and the other removes $e_2$ from $\overline{g}_2$ 
(so that it will list others that exclude $\{e_1,e_2\}$). It continues the process, until the last branch $B_{|E(g)|}$ is formed. 

{\cheng We call the \emph{\underline{e}dge-oriented branching-based \underline{BB} framework for \underline{$k$}-\underline{c}lique listing \texttt{EBBkC}.}}
We present in Algorithm~\ref{alg:EBBkC} the pseudo-code of \texttt{EBBkC}, which differs with \texttt{VBBkC} mainly in the branching step (lines {\cheng 10-}11). 
{\cheng
{\chengB We note that while} a branching step in our edge-oriented branching (i.e., including an edge) can be treated as two branching steps in the existing vertex-oriented branching (i.e., including two vertices of an edge via two steps), our \texttt{EBBkC} has a major difference from \texttt{VBBkC} as follows. For the former, we can explore \emph{arbitrary} edge orderings for each branching step while for the latter, the underlying edge orderings are \emph{constrained} by the adopted vertex ordering. For example, once a vertex ordering is decided for vertex-oriented branching with vertex $v_i$ appearing before $v_j$, then the edges that are incident to $v_i$ would appear before those that are incident to $v_j$ for the corresponding edge-oriented branching. For this reason, the existing \texttt{VBBkC} framework is covered by our \texttt{EBBkC} framework, i.e., for any instance of \texttt{VBBkC} with a vertex ordering, there exists an edge ordering such that the corresponding \texttt{EBBkC} instance is equivalent to the \texttt{VBBkC} instance, \emph{but not vice versa}. This explains why the time complexity of \texttt{EBBkC} based on some edge ordering is better than that of \texttt{VBBkC} (details can be found in Section~\ref{subsec:EBBkC-T}).
}

In the sequel, we explore different {\cheng orderings} of edges in \texttt{EBBkC}.
{\chengB Specifically, with the proposed truss-based edge ordering}, \texttt{EBBkC} would have the worst-case time complexity of {\Revise $O(\delta m + k m  (\tau/2)^{k-2})$} with $\tau < \delta$. {\Revise We note that the time complexity is better than that of the state-of-the-art \texttt{VBBkC} algorithms (which is $O(k m  (\delta/2)^{k-2})$)} {\Yui when $k>3$} 
(Section~\ref{subsec:EBBkC-T}). 
{\chengB With the proposed color-based edge ordering}, \texttt{EBBkC} can apply some pruning rules to improve the efficiency in practice (Section~\ref{subsec:EBBkC-C}). Then, with the proposed hybrid {\chengB edge} ordering, \texttt{EBBkC} would inherit both the above theoretical {\chengB result} and practical {\chengB performance} (Section~\ref{subsec:EBBkC-H}). 
Finally, we discuss other potential applications of \texttt{EBBkC} (Section~\ref{subsec:discussion}).

\begin{algorithm}[t]
\small
\caption{{\cheng \texttt{EBBkC} with truss-based edge ordering}: \texttt{EBBkC-T}}
\label{alg:EBBkC-T}
\KwIn{A graph $G=(V,E)$ and an integer $k\geq 3$}
\KwOut{All $k$-cliques within $G$}
\SetKwFunction{List}{\textsf{EBBkC-T\_Rec}}
\SetKwProg{Fn}{Procedure}{}{}
\tcc{Initialization and branching at $(\emptyset,G,k)$}
$\pi_\tau(G)\leftarrow$ the truss-based ordering of edges in $G$\;
\For{{\kaixin each edge $e_i\in E(G)$ following $\pi_\tau(G)$}}{
    {\kaixin Obtain $S_i$ and $g_i$ according to Eq.~(\ref{eq:edge-oriented-branching})}\;
    {Initialize $VSet(e_i)\leftarrow V(g_i)$ and $ESet(e_i)\leftarrow E(g_i)$}\;
    \List{$S_i, g_i, k-2$}\;
}
\Fn{\List{$S, g, l$}} {
    Conduct Pruning and Termination (lines 3-9 of Algorithm~\ref{alg:EBBkC})\; 
    \tcc{Branching when $l\geq 3$}
    \For{each edge $e$ in $E(g)$}{
        {$S'\leftarrow S\cup V(e)$, $g'\leftarrow (V(g)\cap VSet(e), E(g)\cap ESet(e))$}\;
        \List{$S', g', l-2$}\;
    }
}
\end{algorithm}

\subsection{\texttt{EBBkC-T}: {\cheng \texttt{EBBkC} with the} Truss-based Edge Ordering}
\label{subsec:EBBkC-T}

Consider the edge-oriented branching at $B=(S,g,l)$ based on an ordering of edges $\langle e_1,e_2,\cdots,e_{|E(g)|}\rangle$. For a {\chengB sub-branch} $B_i$ ($1\leq i \leq |E(g)|$) produced from $B$, we observe that the size of graph instance $g_i$ {\cheng (i.e., the number of vertices in $g_i$)}
{\chengB is equal to}
the number of common neighbors of vertices in $V(e_i)$ {\chengB in $\overline{g}_i$}, which depends on the ordering of edges. Formally, we have 
\begin{equation}
|V(g_i)|=|N(V(e_i),\overline{g}_i)|,\ \text{where}\ \overline{g}_i=g[e_i,\cdots,e_{|E(g)|}].    
\end{equation}
Recall that the smaller a graph instance is, the faster the corresponding branch can be solved. %
Therefore, to reduce the time costs, we {\cheng aim} to minimize the {\cheng sizes} of graph instances by determining the ordering of edges via the following greedy procedure.

\smallskip
\noindent\textbf{Truss-based edge ordering.} We determine the ordering of edges by an iterative process. Specifically, it iteratively removes from $g$ the edge whose vertices have the smallest number of common {neighbors} (correspondingly, the smallest size of a graph instance), and then adds it to the end of the ordering. {\chengB Consequently}, for the $i$-th edge $e_i$ in the produced ordering ($1\leq i\leq |E(g)|$), we have
\begin{equation}
\label{eq:truss-based-edge}
    e_i = \min_{e\in E(g)\setminus \{e_1,\cdots,e_{i-1}\}} |N(V(e_i), g[E(g)\setminus \{e_1,\cdots,e_{i-1}\}])|.
\end{equation}
We call the above ordering \emph{truss-based edge ordering} of $g$ and denote it by $\pi_\tau(g)$ since the corresponding iterative process is the same to the \emph{truss decomposition}~\cite{che2020accelerating,wang2012truss,cohen2008trusses}, which can be done in {\Revise $O(\delta m)$ {\chengB time}, where $\delta$ is the degeneracy of the graph \cite{che2020accelerating}}. 

\smallskip
\noindent\textbf{The \texttt{EBBkC-T} algorithm.} 
{\cheng When the truss-based edge ordering is adopted in \texttt{EBBkC}, we call the resulting algorithm \texttt{EBBkC-T}.}
The pseudo-code of \texttt{EBBkC-T} is presented in Algorithm~\ref{alg:EBBkC-T}. 
In particular, it only computes the truss-based edge ordering of $G$ (i.e., $\pi_{\tau}(G)$) for the branching at the initial branch $(\emptyset,G,k)$ (lines 2-5). Then, for any other branching step at a following branch $(S,g,l)$, edges in $\langle e_1,e_2,\cdots,e_{|E(g)|} \rangle$ adopt the same ordering to those used in $\pi_{\tau}(G)$, which could differ with the truss-based edge ordering of $g$. Formally, $e_i$ comes before $e_j$ in $\langle e_1,e_2,\cdots,e_{|E(g)|} \rangle$ (i.e., $i<j$) if and only if it does so in $\pi_{\tau}(G)$. 
To implement this efficiently, it maintains two additional auxiliary sets, i.e., {$VSet(\cdot)$ and $ESet(\cdot)$} (lines 2-4). The idea is that for an edge $e$, all edges in {$ESet(e)$} are ordered behind $e$ in $\pi_{\tau}$ and the vertices incident to these edges are both connected with those incident to $e$. Therefore, the branching steps (with the introduced edge ordering) can be efficiently conducted via set intersections (line 9) for those branches following $(\emptyset,G,k)$. The correctness of this implementation can be easily verified. 

\smallskip
\noindent\textbf{Complexity.} 
{\cheng Given a branch $B=(S,g,l)$, let $\tau(g)$ be the largest size of a produced graph instance, i.e., }
\begin{equation}
\label{eq:tau}
    \tau(g) = \max_{e_i\in E(g)} |V(g_i)|.
\end{equation}
{\cheng We have the following observation.}
\begin{lemma}
\label{lemma:comparison}
    When applying the truss-based {\chengB edge} ordering $\pi_\tau(g)$ at $(S,g,l)$, we have $\tau(g)<\delta(g)$, where $\delta(g)$ is the degeneracy of $g$.
\end{lemma}
\begin{proof}
    We prove by contradiction. Suppose that $\tau\ge \delta$. Since $\delta$ is defined as the largest value of $k$ such that the $k$-core of $g$ is non-empty, there must not have a $(\delta+1)$-core in $g$, i.e., $V(C_{\delta+1})=\emptyset$. {\chengB Here, $C_k$ refers to a $k$-core.}
    Consider the branching step at such an edge $e_i\in E(g)$ that produces the largest size of graph instance, i.e., $|V(g_i)|=\tau$. According to Eq.~(\ref{eq:truss-based-edge}), $e_i$ has the minimum number of common neighbors of its end points in the graph $\overline{g}_i=g[e_{i},\cdots, e_{|E(g)|}]$. This means for each edge $e\in E(\overline{g}_i)$, the number of common neighbors of the end points of $e$ is no less than $\tau$, i.e., $|N(V(e), \overline{g}_i)| \ge \tau$,
    Obviously, $\overline{g}_{i}$ is non-empty, i.e., $V(\overline{g}_{i})\neq \emptyset$. 
    Then for each vertex $v\in V(\overline{g}_{i})$, the number of its neighbors is at least $\tau+1$, i.e., $|N(v,\overline{g}_{i})|\ge \tau+1$. Therefore, $\overline{g}_{i}$ is a subgraph of $(\tau+1)$-core, i.e., $V(\overline{g}_{i})\subseteq V(C_{\tau+1})$. According to the hereditary property of $k$-core\footnote{{\kaixin The hereditary property claims that given a graph $G$ and two integers $k$ and $k'$ with $k\le k'$, then the $k'$-core of $G$ is a subgraph of the $k$-core of $G$ \cite{batagelj2003m}.}} and the hypothesis $\tau\ge\delta$, we have $V(C_{\tau+1}) \subseteq V(C_{\delta+1})$, which leads to a contradiction that $V(\overline{g}_{i})\subseteq V(C_{\tau+1})\subseteq V(C_{\delta+1})= \emptyset$. 
\end{proof}

{\cheng Based on the above result, we derive that}
the time complexity of \texttt{EBBkC-T} is {\Revise better} than that of the state-of-the-art {\chengB algorithms}, i.e., $O(k m  (\delta/2)^{k-2})$, which we show in the following theorem.

\begin{theorem}
\label{theo:ebbkc-t}
     Given a graph $G=(V,E)$ and an integer $k\geq 3$, the time and {space complexities of \texttt{EBBkC-T} are} {\Revise $O(\delta m + k m  (\tau/2)^{k-2})$} and {$O(m+n)$}, respectively, where $\tau=\tau(G)$ and it is strictly smaller than the degeneracy $\delta$ of $G$.
\end{theorem}

\begin{proof}
    We give a sketch of the proof and put the details in the technical report \cite{wang2024technical}. The running time of \texttt{EBBkC-T} {\Revise {\chengC consists of} the time of generating the truss-based edge ordering, which is $O(\delta m)$ \cite{che2020accelerating}, and the time of the recursive listing procedure (lines 6-10 of Algorithm~\ref{alg:EBBkC-T}). Consider the latter one. }
    Given a branch $B=(S, g, l)$, we denote by $T(g, l)$ the upper bound of time cost of listing $l$-cliques under such a branch. When $k\ge 3$, with different values of $l$, we have the following recurrences. 
\begin{equation}
\label{eq:ebbkc-recurrence-main}
T(g, l) \le \left\{
\begin{array}{lc}
    O(k\cdot |V(g)|) & l = 1\\ 
    O(k\cdot |E(g)|) &  l = 2 \\ 
    \sum_{e\in E(g)} \Big( T(g', l-2) + T'(g') \Big) & 3\le l\le k-2 \\
\end{array}
\right.
\end{equation}
where $T'(g')$ is the time for constructing $g'$ given $B=(S, g, l)$ (line 9 {\chengB of Algorithm~\ref{alg:EBBkC-T}}). We show that with different values of $l$, $T'(g')$ satisfies the following equation. 
\begin{equation}
\label{eq:close-T-prime-main}
    \sum_{e\in E(g)} T'(g') = \left\{
\begin{array}{lc}
    O(\tau \cdot |E(g)|) & l = 3 \\
    O(\tau^2 \cdot |E(g)|) & l > 3 \\
\end{array}
\right.
\end{equation}
The reason is as follows. When given a branch $B=(S,g,l)$ with $l=3$, for each edge, we just need to compute $V(g')$ for the sub-branch (since the termination when $l=1$ only cares about the vertices in $g'$), which can be done in $O(\tau)$. When given a branch $B=(S,g,l)$ with $l>3$, for each edge, we need to construct both $V(g')$ and $E(g')$, which can be done in $O(\tau^2)$ since there are at most $\tau(\tau-1)/2$ edges in $g$. Besides, we show that given a branch $B=(S,g,l)$ and the sub-branches $B'=(S',g',l')$ produced at $B$, we have
\begin{equation}
    \sum_{e\in E(g)} |E(g')| < \left \{
    \begin{array}{lc}
      \frac{\tau^2}{4} \cdot |E(g)|   &  l < k \\
       \frac{\tau^2}{2} \cdot |E(g)|  & l = k
    \end{array}
\right.
\end{equation}
This inequality is proven in the technical report \cite{wang2024technical}. With above inequalities, we can prove the theorem by induction on $l$. 
\end{proof}

{\Revise
{\chengC 
\noindent\textbf{\texttt{VBBkC} v.s. \texttt{EBBkC-T} (Time Complexity).}
{\chengC The time complexity of \texttt{EBBkC-T} (i.e., $O( m \delta + k m (\tau/2)^{k-2})$) is better than that of state-of-the-art {\chengC \texttt{VBBkC} algorithms (i.e., $O(km (\delta/2)^{k-2})$) {\Yui for $k>3$}.
This is because both (1) $O(m\delta)$ and (2) $O(k m (\tau/2)^{k-2})$ are bounded by $O(k m  (\delta/2)^{k-2})$.
For (1), it is because we have $\delta < 2k (\delta/2)^{k-2}$ when $k>3$; and for (2), it is because $\tau <\delta$.
}
{\Yui We note that for $k=3$, the time complexity of \texttt{EBBkC-T} is dominated by $O(\delta m)$ since $\tau<\delta$, which is the same as that of state-of-the-art \texttt{VBBkC} algorithms (i.e., $O(\delta m)$) and is also the same as that of algorithms for listing triangles~\cite{latapy2008main,ortmann2014triangle} in the worst case (i.e., $O(m^{1.5})$) since $\delta<\sqrt{m}$.
}
}}

{
\smallskip
\noindent\textbf{Discussion on $\tau$.}
By definition, $\tau$ of a graph $G$, i.e., $\tau(G)$, corresponds to the largest integer such that there exists a non-empty subgraph where the two vertices of each edge have at least $\tau$ common neighbors.
Similar to the degeneracy $\delta$ of the graph, $\tau$ also measures the density of a graph, say, the larger the value of $\tau$, the denser the graph. However, $\tau$ is always smaller than $\delta$ as it imposes stricter constraint on connections (i.e., the two vertices of every edge have at least $\tau$ common neighbors v.s. every vertex has at least $\delta$ neighbors). 
\underline{Theoretically}, for a graph with $n$ vertices, the gap between $\delta$ and $\tau$ can be as large as $n/2$.
To see this, consider a complete bipartite graph with $p$ vertices on each side, where $p$ is a positive integer. For this graph, we have $\delta = p$ since each vertex has exactly $p$ neighbors and $\tau = 0$ since the two vertices of each edge have no common neighbors.
\underline{Practically}, the ratio $\tau/\delta$ is below 0.8 for many real-world graphs. For example, we have collected the statistics of $\tau/\delta$ on 139 more real-world graphs \cite{RealGraphs} and found that the ratio is below 0.8 for the majority of the graphs (105 out of 139). 
}

\smallskip
\noindent\textbf{Remark.} 
(1) We note that another option of designing \texttt{EBBkC-T} is to 
{\chengB compute the truss-based edge ordering for each individual branch and use it for branching at the branch.}
However, it {\chengB would} introduce additional time cost 
{\chengB without achieving} better theoretical time complexity. Thus, we {\chengB choose not to adopt this option.} 
{\kaixin
(2) It is worthy noting that for a truss-based edge ordering, {\chengB there} does not always exist a vertex ordering such that the instance of \texttt{VBBkC} {\chengB with the vertex ordering} is equivalent to the instance of \texttt{EBBkC-T}. {\chengB We include a counter-example in the technical report~\cite{wang2024technical} for illustration.}
}


\subsection{\texttt{EBBkC-C}: {\cheng \texttt{EBBkC} with the} Color-based Edge Ordering}
\label{subsec:EBBkC-C}

{\kaixin 
{\cheng While the truss-based edge ordering helps to form sub-branches with small sizes at a branch, it does not offer much power to prune the formed sub-branches - all we can leverage for pruning are some size constraints (line 3 of Algorithm~\ref{alg:EBBkC}). On the other hand, some existing studies of \texttt{VBBkC} have successfully adopted color-based \emph{vertex} ordering for effective pruning~\cite{hasenplaugh2014ordering, yuan2017effective}. Specifically, consider a branch $B=(S,g,l)$. They first color the vertices in $g$ by iteratively assigning to an uncolored vertex $v$ the smallest color value taken from $\{1,2,\cdots\}$ that has not been assigned to $v$'s neighbours. {\chengB Let $c$ be the number of color values used by the coloring procedure.} They then obtain a vertex orderng by sorting the vertices in a non-increasing order based on the color values $\langle v_1,v_2,\cdots,v_{|V(g)|} \rangle$ (with ties broken by node ID), i.e., for $v_i$ and $v_j$ with $i<j$, we have $col(v_i)\geq col(v_j)$, where $col(\cdot)$ is the color value of a vertex. As a result, they prune the sub-branch $B_i=(S_i, g_i, l_i)$, which includes $v_i$ to $S$, if $col(v_i)<l$.}
%
The rationale is that since all vertices in $g_i$ have their color values strictly smaller than $col(v_i)$ (according to Eq.~(\ref{eq:vertex-oriented-branching}) and the definition of the color-based vertex ordering), they do not have $l-1$ different color values, indicating $g_i$ does not contain any $(l-1)$-cliques, {\cheng and therefore} $B_i$ can be pruned.

\begin{algorithm}[t]
\small
\caption{{\cheng \texttt{EBBkC} with color-based edge ordering}: \texttt{EBBkC-C}}
\label{alg:EBBkC-C}
\KwIn{A graph $G=(V,E)$ and an integer $k\geq 3$}
\KwOut{All $k$-cliques within $G$}
\SetKwFunction{List}{\textsf{EBBkC-C\_Rec}}
\SetKwProg{Fn}{Procedure}{}{}
Conduct vertex coloring on $G$ and get $id(v)$ for each vertex in $V$\;
$\overrightarrow{G}\leftarrow (V,\overrightarrow{E})$ where $\overrightarrow{E}=\{u\rightarrow v \mid (u,v)\in E \land id(u)<id(v)\}$\;
\List{$\emptyset,\overrightarrow{G}, k$}\;

\Fn{\List{$S, \overrightarrow{g}, l$}} {
    Conduct Pruning and Termination (line 3-9 of Algorithm~\ref{alg:EBBkC})\; 
    \tcc{Branching when $l\geq 3$}
    \For{each edge $u\rightarrow v$ in $E(\overrightarrow{g})$}
    {
        $S'\leftarrow S\cup \{u,v\}$ and $\overrightarrow{g}'\leftarrow \overrightarrow{g}[N^+(\{u,v\},\overrightarrow{g})]$\;
        \lIf{either of the rules of pruning applies}{\textbf{continue}}
        \List{$S', \overrightarrow{g}', l-2$}\;
    }
}
\end{algorithm}



Recall that in Section~\ref{subsec:motivation}, for an instance of \texttt{VBBkC} with a vertex ordering, there {\cheng would} exist an edge ordering such that the corresponding \texttt{EBBkC} instance {\cheng based on the edge ordering} is equivalent to the \texttt{VBBkC} instance. 
{\cheng Motivated by this, we propose to adopt the edge ordering that corresponds to the color-based vertex ordering, which we call \emph{color-based edge ordering}, for our \texttt{EBBkC}. One immediate benefit is that it would naturally inherit the pruning power of the color-based vertex ordering, which has been demonstrated for \texttt{VBBkC}~\cite{hasenplaugh2014ordering, yuan2017effective}. Furthermore, it would introduce new opportunities for pruning, compared with existing \texttt{VBBkC} with the color-based vertex ordering, since it can leverage the two {\chengB vertices incident to} an edge collectively (instead of a single vertex twice as in \texttt{VBBkC}) for designing new pruning rules, which we explain next.}


}

\begin{figure}[t]
\vspace{-5mm}
\subfigure[An example graph $G$.]{
    \label{subfig:g2}
    \includegraphics[width=0.23\textwidth]{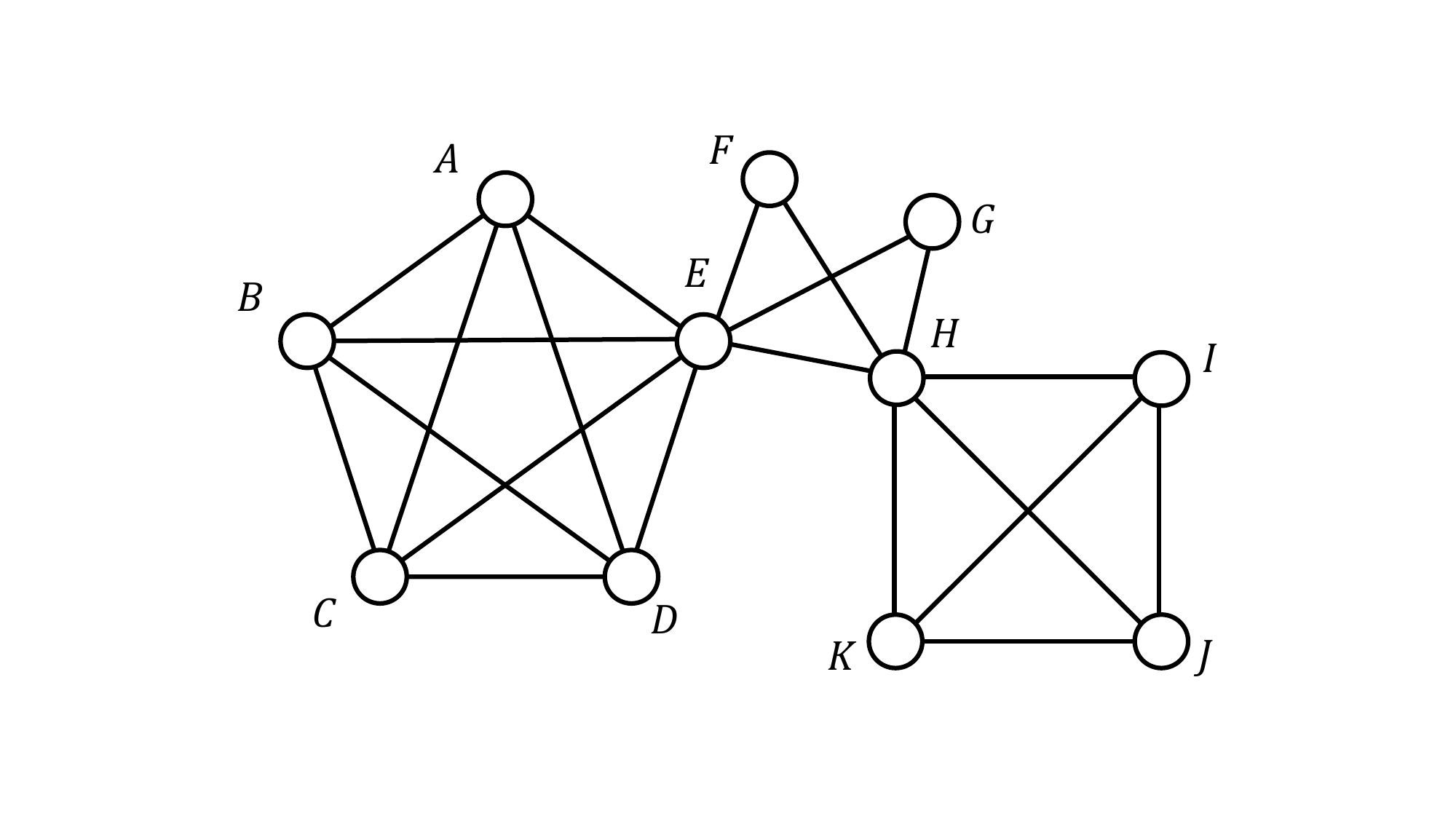}}
\subfigure[The DAG built upon $G$ {\chengB (the numbers inside circles are the color values)}.]{
    \label{subfig:color-dag}
    \includegraphics[width=0.23\textwidth]{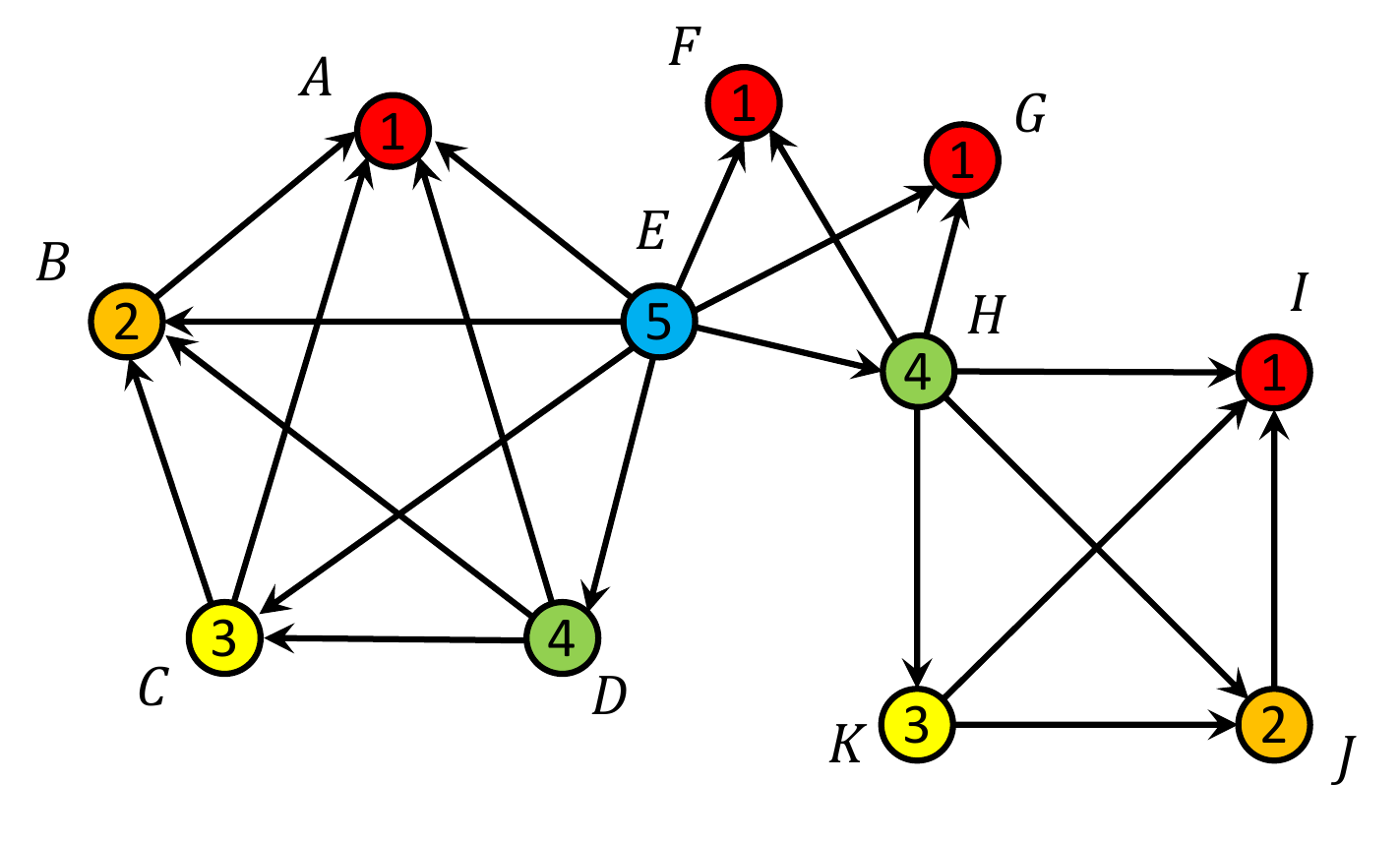}}
\vspace{-3mm}
\caption{Color-based edge ordering and pruning rules.}
\label{fig:example2}
\vspace{-5mm}
\end{figure}

\smallskip
\noindent\textbf{Color-based edge ordering and pruning rules.} 
Consider the branch $B=(S,g,l)$. 
We first color the graph $g$ {\cheng using} the graph coloring technique~\cite{hasenplaugh2014ordering, yuan2017effective} and obtain the color-based vertex ordering $\langle v_1,v_2,\cdots,v_{|V(g)|} \rangle$ such that $col(v_i)\geq col(v_j)$ for $i<j$. 
%
Then for each vertex $u\in V(g)$, let $id(u)$ be the position of a vertex $u$ in the ordering. 
{\kaixin{For each edge $e=(u,v)\in E(g)$ {\cheng with $id(u)<id(v)$}, we define $str(e)$ 
as a string,
{\cheng which is the concatenation of}
$id(u)$ and $id(v)$ (i.e., $str(e)=$`$id(u)$+$id(v)$').}} 
Finally, we {\cheng define} the color-based edge ordering 
{\cheng as}
the alphabetical ordering {\cheng based on} $str(e)$ for {\cheng edges} $e\in E(g)$. {\cheng That is}, one edge $e = (u,v)$ 
{\cheng with $id(u)<id(v)$}
comes before {\chengB another edge} $e' = (u', v')$ {\cheng with $id(u')<id(v')$}
if (1) $id(u)<id(u')$ or (2) $id(u)=id(u')$ and $id(v)<id(v')$.
%

Consider a branch $B=(S,g,l)$ and 
{\cheng and a sub-branch $B_i$ that includes edge $e_i=(u,v)$ with $col(u)>col(v)$ to $S$.}
We can apply the following two pruning rules.
%
\begin{itemize}[leftmargin=*]
    \item {\cheng \textbf{Rule (1).}} If $col(u)<l$ or $col(v)<l-1$, we prune sub-branch $B_i$;  
    \item {\cheng \textbf{Rule (2).}} If the vertices in produced sub-branch have less than $l-2$ distinct color values, we prune sub-branch $B_i$.  
\end{itemize}

{\kaixin 
We note that Rule (1) is 
{\cheng equivalent to the pruning that \texttt{VBBkC} with the color-based vertex ordering conducts at two branching steps of including $u$ and $v$~\cite{li2020ordering}.
Rule (2) is a new one, which applies only in our \texttt{EBBkC} framework with the color-based edge ordering. In addition, Rule (2) is more powerful than Rule (1) in the sense that if Rule (2) applies, then Rule (1) applies, but not vice versa.}
}
{\kaixin The reason is that the color values of $u$ and $v$ are sometimes much larger than the number of distinct color values in the sub-branch since both color values consider the information of their own neighbors instead of their common neighbors. 
{\cheng For illustration}, consider the example in Figure~\ref{fig:example2} and assume that we aim to list 4-cliques (i.e., $k=4$). We focus on the edge ${EH}$. It is easy to check that 
{\cheng Rule (1) does not apply, but Rule (2) applies since the vertices in the produced sub-branch, i.e., $F$ and $G$, have only one color value.
Given that it takes $O(1)$ time to check if Rule (1) applies and $O(|V(g_i)|)$ time to check if Rule (2) applies, our strategy is to check Rule (1) first, and if it does not apply, we further check Rule (2).}
}
{We remark that Rule (2) can be adapted to some \texttt{VBBkC} based algorithms including \texttt{DDegCol}~\cite{li2020ordering} and \texttt{BitCol}~\cite{yuan2022efficient}, that is, we prune sub-branch $B_i$, if the vertices inside have less than $l-1$ distinct color values.
}

%

\smallskip
\noindent\textbf{The \texttt{EBBkC-C} algorithm.} 
{\cheng When the color-based edge ordering is adopted in \texttt{EBBkC}, we call the resulting algorithm \texttt{EBBkC-C}.}
The pseudo-code of \texttt{EBBkC-C} is presented in Algorithm~\ref{alg:EBBkC-C}. {With the color-based edge ordering, a directed acyclic graph (DAG), denoted by $\overrightarrow{G}$, is built for efficiently conducting the branching steps \cite{li2020ordering, yuan2022efficient}. Specifically, $\overrightarrow{G}$ is built upon $G$ by orienting each edge $(u,v)$ in $E$ with $id(u)<id(v)$ 
{\cheng from $u$ to $v$}
(line 2). 
{\chengB For illustration, consider Figure~\ref{fig:example2}.}
Given $V_{sub}\subseteq V$, let $N^+(V_{sub},\overrightarrow{g})$ be the common out-neighbours of the vertices in $V_{sub}$ in $\overrightarrow{g}$. Then, given a branch $B=(S,\overrightarrow{g},l)$\footnote{Note that $\overrightarrow{g}$ is a subgraph of $\overrightarrow{G}$, whose edge {\chengB orientations are the same as those} of $\overrightarrow{G}$.}, the edge-oriented branching at a branch with the color-based ordering can be easily conducted by calculating the common out-neighbors of the {\chengB vertices incident to} an edge in $\overrightarrow{g}$ (line 7). Then, we will prune the produced branch if either of the above two 
{\cheng rules applies}
(line 8). 
}


\smallskip
\noindent\textbf{Complexity.} The time cost of \texttt{EBBkC-C} is {$O(k m  (\Delta/2)^{k-2})$}, {\cheng which we show in the following theorem.}
In practice, \texttt{EBBkC-C} runs faster than \texttt{EBBkC-T}. 

\begin{theorem}
\label{theo:ebbkc-c}
    Given a graph $G=(V,E)$ and an integer $k\geq 3$, the time and {space complexities of \texttt{EBBkC-C} are} {$O(k m  (\Delta/2)^{k-2})$} and {$O(m+n)$}, respectively, where $\Delta$ is the maximum degree of $G$.
\end{theorem}

\begin{proof}
The proof is similar to that of Theorem~\ref{theo:ebbkc-t}. The difference is that the largest size of the produced graph instance in \texttt{EBBkC-C} can only be bounded by $\Delta$ {\chengC and the time of generating the color-based edge ordering is $O(m)$, which is dominated by the cost of the recursive procedure}. 
\end{proof}


\subsection{\texttt{EBBkC-H}: {\cheng \texttt{EBBkC} with} Hybrid Edge Ordering}
\label{subsec:EBBkC-H}

Among \texttt{EBBkC-T} and \texttt{EBBkC-C}, the former has a better theoretical time complexity and the latter enables effective pruning in practice. 
{
Inspired by the hybrid vertex ordering used for \texttt{DDegCol}~\cite{li2020ordering} and \texttt{BitCol}~\cite{yuan2022efficient}, we aim to achieve the merits of both algorithms by adopting both the truss-based edge ordering (used by \texttt{EBBkC-T}) and the color-based edge ordering (used by \texttt{EBBkC-C}) in the \texttt{EBBkC} framework.}
%
%
%
Specifically, we first apply the truss-based edge ordering for the branching step at the initial branch $(\emptyset,G,k)$. Then, for the following branches, we adopt the color-based ordering for their branching steps. We call this algorithm based on the hybrid edge ordering \texttt{EBBkC-H}. The pseudo-code of \texttt{EBBkC-H} is presented in Algorithm~\ref{alg:EBBkC-H} and the implementations of branching steps are similar to those in \texttt{EBBkC-T} and \texttt{EBBkC-C}. 
{\cheng The size of a produced problem instance for \texttt{EBBkC-H} is bounded by $\tau$ (due to the branching at the initial branch based on the truss-based edge ordering),} and thus \texttt{EBBkC-H} achieves the same time complexity as \texttt{EBBkC-T}.
{\cheng In addition, \texttt{EBBkC-H} enables effective pruning for all branches except for the initial branch (since color-based edge coloring is adopted at these branches), and thus it runs fast in practice as \texttt{EBBkC-C} does.}

\begin{algorithm}[t]
\small
\caption{{\chengB \texttt{EBBkC} with hybrid edge ordering:}
\texttt{EBBkC-H}}
\label{alg:EBBkC-H}
\KwIn{A graph $G=(V,E)$ and an integer $k\geq 3$}
\KwOut{All $k$-cliques within $G$}
\SetKwFunction{List}{\textsf{EBBkC-C\_Rec}}
\SetKwProg{Fn}{Procedure}{}{}
\tcc{Initialization and branching at $(\emptyset,G,k)$}
$\pi_\tau(G)\leftarrow$ the truss-based ordering of edges in $G$\;
\For{each edge $e_i\in E(G)$ following $\pi_\tau(G)$}{
    Obtain $S_i$ and $g_i$ according to Eq.~(\ref{eq:edge-oriented-branching})\;
    Do vertex coloring on $g_i$ and get $id(v)$ for each vertex in $V(g_i)$\;
    $\overrightarrow{g_i}\leftarrow(V(g_i),\{u\rightarrow v \mid (u,v)\in E(g_i)  \land id(u)<id(v)\})$\;
    \List{$S_i, \overrightarrow{g_i}, k-2$}\;    
}
\end{algorithm}

\smallskip
\noindent\textbf{Complexity.} The time complexity of \texttt{EBBkC-H} is {\Revise $O(\delta m + k m  (\tau/2)^{k-2})$}, which is the same as that of \texttt{EBBkC-T}.

\begin{theorem}
\label{theo:ebbkc-h}
     Given a graph $G=(V,E)$ and an integer $k\geq 3$, the time and {space complexities of \texttt{EBBkC-H} are} {\Revise $O(\delta m + k m  (\tau/2)^{k-2})$} and {$O(m+n)$}, respectively, where $\tau=\tau(G)$ is strictly smaller than the degeneracy $\delta$ of $G$.
\end{theorem}

\begin{proof}
The proof is similar to that of Theorem~\ref{theo:ebbkc-t}. Since the largest size of the produced graph instance in \texttt{EBBkC-H} can also be bounded by $\tau$, it has the same worst-case time complexity as that of \texttt{EBBkC-T}.
\end{proof}

\subsection{{Other Potential Applications of \texttt{EBBkC}}}
\label{subsec:discussion}
{
Our \texttt{EBBkC} framework can potentially be applied to other problems than the $k$-clique listing problem, which we discuss as follows.
\underline{First}, our framework can be easily adapted to solve other clique mining tasks, including maximal clique enumeration (MCE) \cite{eppstein2010listing, naude2016refined, tomita2006worst,jin2022fast}, maximum clique search (MCS) \cite{chang2019efficient, chang2020efficient} and diversified top-$k$ clique search (DCS) \cite{yuan2016diversified, wu2020local}. The rationale is that our framework can explore all possible cliques in an input graph and thus can output only the desired cliques that satisfy some properties (e.g., maximality and diversity) by filtering out others.
%
%
\underline{Second}, our framework can be potentially extended to mining other types of \textit{connected} dense subgraphs, e.g., connected $k$-plex. This is because our framework can be used to explore all possible subsets of edges by recursively including an edge, and thus the induced subgraphs will cover all possible connected dense subgraphs.
\underline{Third}, there are some potential benefits when adapting our framework to the above tasks. As discussed in Section~\ref{subsec:motivation}, the edge-oriented branching can provide more information (i.e., an edge involving two vertices instead of one vertex) towards designing more effective pivot techniques and/or pruning rules than the existing vertex-oriented branching.

}

%% file: early-stop.tex
\section{Early Termination Technique}
\label{sec:early-termination}

\begin{figure}[t]
\vspace{-5mm}
\subfigure[An example 2-plex graph $g$.]{
    \label{subfig:2-plex}
    \includegraphics[width=0.20\textwidth]{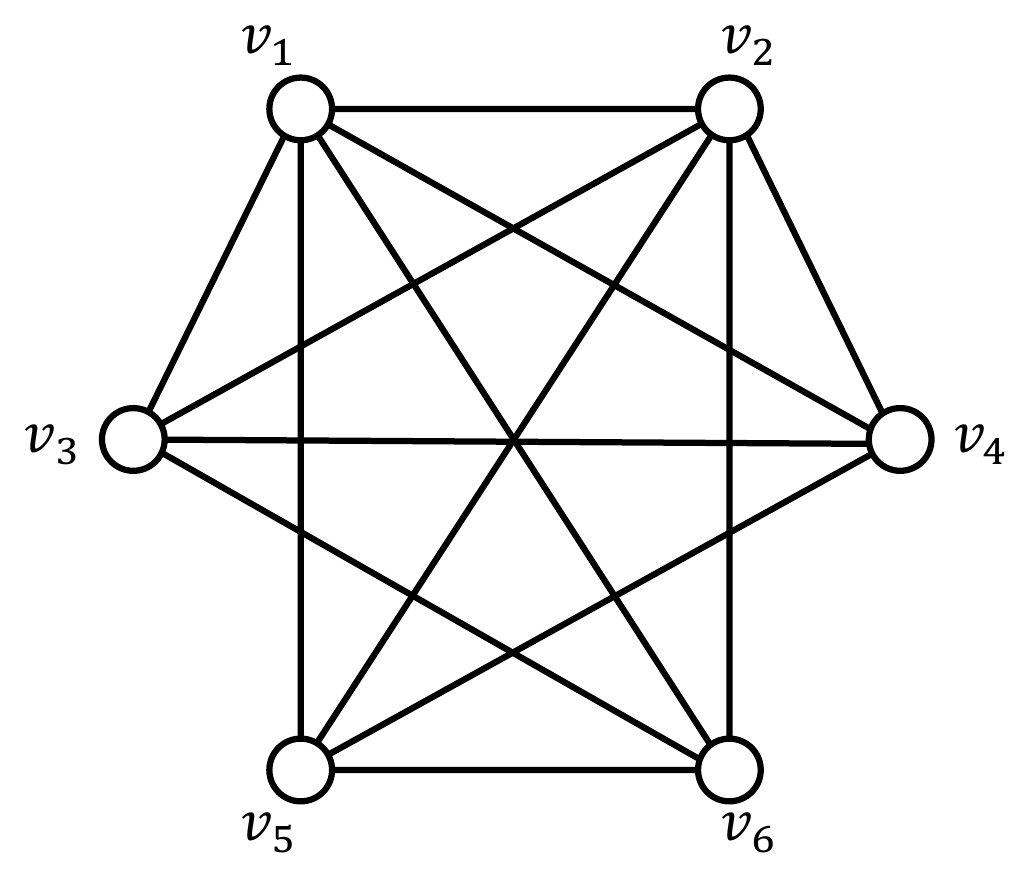}}
\subfigure[Inverse graph $g_{inv}$ of 2-plex $g$.]{
    \label{subfig:2-plex-inv}
    \includegraphics[width=0.20\textwidth]{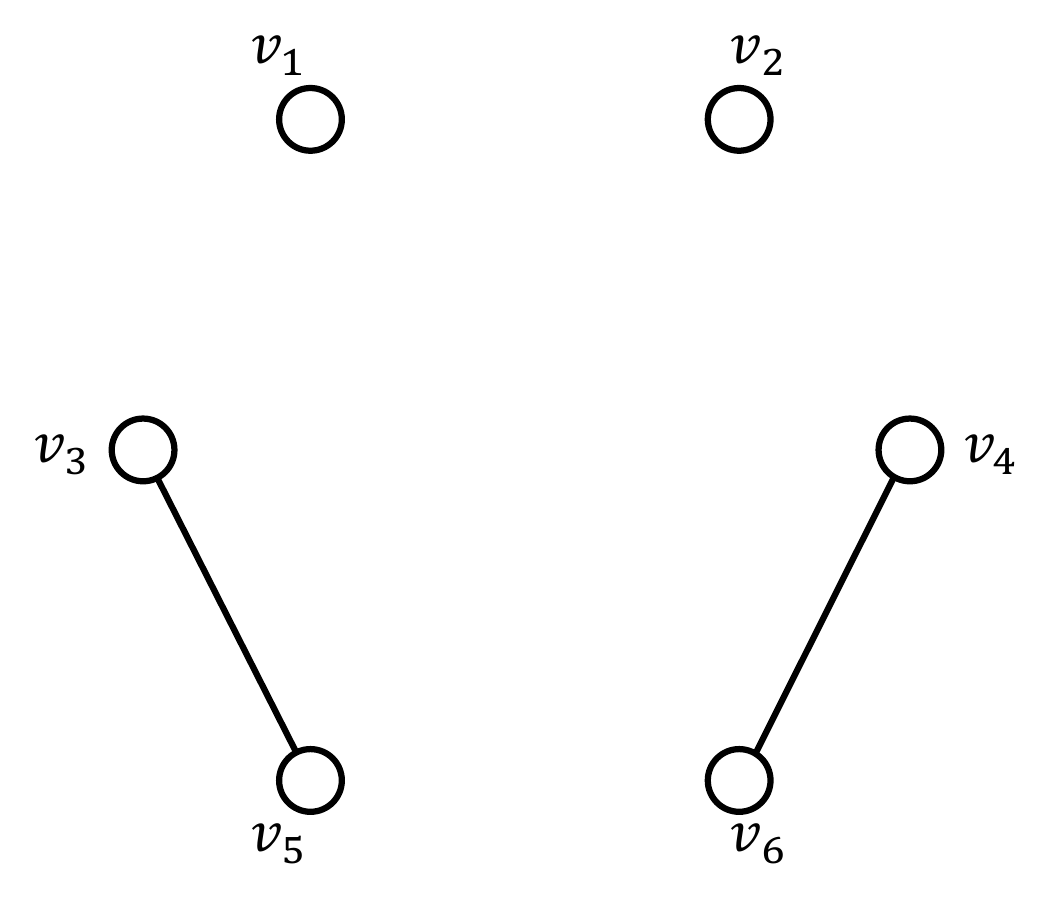}}
\vspace{-3mm}
\caption{{Examples of a 2-plex and its inverse graph.}} 
\label{fig:t-plex}
\vspace{-5mm}
\end{figure}

{\yui
{\cheng Suppose that we are at a branch $B = (S, g, l)$, where graph $g$ is dense (e.g., $g$ is a clique or nearly a clique), and the goal is to list $l$-cliques in $g$ (and merge them with $S$). Based on the \texttt{EBBkC} framework, we would conduct branching at branch $B$ and form sub-branches. Nevertheless, since $g$ is dense, there would be many sub-branches to be formed (recall that we form $|E(g)|$ sub-branches), which would be costly. Fortunately, for such a branch, we can list the $l$-cliques efficiently without continuing the recursive branching process of \texttt{EBBkC}, i.e., we can \emph{early terminate} the branching process. Specifically, we have the following two observations.}
\begin{itemize}[leftmargin=*]
    \item 
    If $g$ is a clique or a $2$-plex (recall that a $t$-plex is a graph where each vertex inside has at most $t$ non-neighbors including itself,
    we can list $l$-cliques in $g$ efficiently in a combinatorial manner. For the former case, we can directly enumerate all possible sets of $l$ vertices in $g$. For the latter case, we can do similarly, but in a bit more complex manner (details will be discussed in Section~\ref{subsec:2-plex}).
    \item 
    If $g$ a $t$-plex (with $t\ge 3$), the branching procedure based on $g$ can be converted to that on its inverse graph $g_{inv}$ (recall that $g_{inv}$ has the same set of vertices as $g$, with an edge between two vertices in $g_{inv}$ if and only if they are disconnected from each other in $g$), which is sparse, and the converted procedure would run faster (details will be discussed in Section~\ref{subsec:3-plex}).
\end{itemize}

{We determine whether $g$ is a $t$-plex for some $t$ by checking the minimum degree of a vertex in $g$ - if it is no less than $|V(g)|-t$, $g$ is a $t$-plex; otherwise, $g$ is not a $t$-plex.
This can be done while constructing the corresponding branch $B=(S,g,l)$ in $O(V(g))$ time.}


\subsection{{\cheng Listing $k$-Cliques} from 2-Plex in Nearly Optimal Time}
\label{subsec:2-plex}

Consider a branch $B=(S,g,l)$ with $g$ as a 2-plex.
{\cheng The} procedure for listing $k$-cliques inside, called \texttt{kC2Plex}, utilizes the \emph{combinatorial technique}. The rationale behind is that listing $k$-cliques from a large clique can be solved in the optimal time by directly enumerating all possible combinations of $k$ vertices. 

{\cheng Specifically,}
we first partition $V(g)$ into three disjoint sets, namely $F$, $L$ and $R$, each of which induces a clique. 
%
This can be done in two steps.
\underline{First}, it partitions $V(g)$ into two disjoint parts: one containing those vertices that are adjacent to all other vertices in $V(g)$
(this is $F$) and the other containing the remaining vertices that are not adjacent to two vertices including itself (this is $L\cup R$). Note that $L\cup R$ always involves an even number of vertices and can be regarded as a collection of pairs of vertices $\{u,v\}$ such that $u$ is disconnected from $v$. \underline{Second}, it further partitions $L\cup R$ into two parts by breaking each pair in $L\cup R$, that is, $L$ and $R$ 
contain the first and the second vertex in each pair, respectively. As a result, every vertex in one set connects all others within the same set  and is disconnected from one vertex from the other set. Note that the partition of $L\cup R$ is not unique and can be an arbitrary one. 
%
{For illustration, consider the example in Figure~\ref{subfig:2-plex}~and~\ref{subfig:2-plex-inv}. The vertices $v_1$ and $v_2$ are adjacent to all other vertices. Thus, $F=\{v_1,v_2\}$ and $L\cup R=\{v_3, v_4, v_5, v_6\}$. One possible partition is $L=\{v_3,v_4\}$ and $R=\{v_5, v_6\}$.}
Below, we elaborate on how the partition $V(g)=F\cup L \cup R$ helps to speedup the $k$-clique listing.

{\cheng Recall}
that the set of $k$-cliques in $B$ can be listed by finding all $l$-cliques in $g$ and merging each of them with $S$. Consider a $l$-clique in $g$. Based on the partition $V(g)=F\cup L \cup R$, it consists of three disjoint subsets of $F$, $L$ and $R$, namely $F_{sub}$, $L_{sub}$ and $R_{sub}$, each of which induces a small clique. Therefore, all $l$-cliques with the form of $F_{sub}\cup L_{sub}\cup R_{sub}$ can be found by iteratively enumerating all possible $|F_{sub}|$-combinations over $F$, $|L_{sub}|$-combinations over $L$ and $|R_{sub}|$-combinations over $R$ such that $|F_{sub}|+|L_{sub}|+|R_{sub}|=l$.

\begin{algorithm}[t]
\yui
\small
\caption{List {\cheng $k$-cliques} in a 2-plex: \texttt{kC2Plex}}
\label{alg:kC2Plex}
\KwIn{A branch $(S,g,l)$ with $g$ corresponding to a 2-plex}
\KwOut{All $k$-cliques within $(S,g,l)$}
Partition $V(g)$ into three disjoint sets $F$, $L$ and $R$\;
\lIf{$|F|+|L|<l$}{\Return}
\For{$c_1\in [\max\{0,l-|L|\}, \min\{l,|F|\}]$ and each $c_1$-combination $F_{sub}$ over $F$}{
    \For{$c_2\in [0, \min\{l-c_1,|L|\}]$ and each $c_2$-combination $L_{sub}$ over $L$}{
        \For{$c_3 \gets l-c_1-c_2$ and each $c_3$-combination $R_{sub}$ over $R\setminus \overline{N}(L_{sub},g)$}{
            \textbf{Output} a $k$-clique $S\cup F_{sub}\cup  L_{sub}\cup R_{sub}$;
        }
    }
}
\end{algorithm}

We present the pseudo-code of \texttt{kC2Plex} in Algorithm~\ref{alg:kC2Plex}. In particular, the integers $c_1$, $c_2$ and $c_3$ are used to ensure the satisfaction of $|F_{sub}|+|L_{sub}|+|R_{sub}|=l$. Specifically, it first finds a $c_1$-clique $F_{sub}$ from $F$ and a $c_2$-clique $L_{sub}$ from $L$. Recall that every vertex in $L$ is disconnected from one vertex in $R$ and vice versa. Hence, it removes from $R$ those vertices that is disconnected from one vertex in $L_{sub}$, which we denote by $\overline{N}(L_{sub},g)$, and this can be done efficiently in $\Theta(|L_{sub}|)$ (as verified by Theorem~\ref{theo:2-plex}), and then finds a $c_3$-clique from $R\setminus \overline{N}(L_{sub},g)$. Besides, when $|F|+|L|<l$, it terminates the procedure since no $l$-clique will be found in $g$ (line 2).

\smallskip
\noindent\textbf{Time complexity.} We analyze the time complexity of 
\texttt{kC2Plex} as follows. 

\begin{theorem}
\label{theo:2-plex}
Given a branch $B=(S,g,l)$ with $g$ being a 2-plex, \texttt{kC2Plex} lists all $k$-cliques within $B$ in $O(|E(g)|+k\cdot c(g,l))$ {\cheng time} where $c(g,l)$ is the number of $l$-cliques in $g$.
\end{theorem}

\begin{proof}
Algorithm~\ref{alg:kC2Plex} takes $O(|E(g)|)$ for partitioning $V(g)$ by obtaining the degree of each vertex inside (line 1). For each round of lines 3-6, the algorithm can guarantee exactly one $k$-clique to be outputted at line 6 based on the settings of $c_1$, $c_2$ and $c_3$. Besides, the operation $R \setminus \overline {N}(L_{sub},g)$ only takes $\Theta(|L_{sub}|)$ time. Specifically, we (1) maintain two arrays $L=\{u_1,u_2,\cdots\}$ and $R=\{v_1,v_2,\cdots\}$ such that $u_i$ is not adjacent to $v_i$ for $1 \le i \le |L|$, and (2) reorder $R$ by switching $|L_{sub}|$ vertices with the same indices as those in $L_{sub}$ to the tail of $R$ which runs in $\Theta(|L_{sub}|)$ time, and take the first $|R|-|L_{sub}|$ vertices in $R$ as $R\setminus  \overline {N}(L_{sub},g)$. 
\end{proof}

\noindent\textbf{Remark.} We remark that \texttt{kC2Plex} achieves the \emph{input-output sensitive} time complexity since the time cost depends on both the size of input $|E(g)|$ and the number of $k$-cliques within the branch. 
Besides, we note that $O(k\cdot c(g,k))$ is the optimal time for listing $k$-cliques within the branch and thus \texttt{kC2Plex} only takes extra $O(|E(g)|)$ time, {\cheng i.e., \texttt{kC2Plex}} 
is \emph{nearly optimal}.

\subsection{{\cheng Listing $k$-Cliques} from $t$-Plex with $t\ge 3$ }
\label{subsec:3-plex}

\begin{algorithm}[t]
\small
\yui
\caption{List {\cheng $k$-cliques} in a $t$-plex ($t\ge 3$): \texttt{kCtPlex}}
\label{alg:k-plex}
\KwIn{A branch $(S,g,l)$ with $g$ corresponding to a $t$-plex}
\KwOut{All $k$-cliques within $(S,g,l)$}
\SetKwFunction{List}{\textsf{kCtPlex\_Rec}}
\SetKwProg{Fn}{Procedure}{}{}

Construct the inverse graph $g_{inv}$ of $g$\;
$I\gets$ set of vertices in $V(g_{inv})$ that  are disconnected from all others\;
\List{$S$, $V(g_{inv})\setminus I$, $l$}

\Fn{\List{$S'$, $C$, $l'$}} {

    {\kaixin
        \tcc{Termination when $l'=0$}
        \If{$l'=0$} {\textbf{Output} a $k$-clique $S'$\; \Return\;}
    }

    \tcc{Choose all rest $l'$ vertices from $I$}
    \If {$|I|\geq l'$} {
    \For{each $l'$-combination $I_{sub}$ over $I$}{
            \textbf{Output} a $k$-clique $S'\cup I_{sub}$\;
        }
        }
    
    \tcc{Choose at least one vertex from $C$}
    \For{each $v_i\in C$}{
        Create a branch $(S_i,C_i,l_i)$ based on Eq.~(\ref{eq:kCtPlex})\;

        \lIf{$|C_i|+|I|\geq l_i$}{\List{$S_i$, $C_i$, $l_i$}}
    }
    
}
\end{algorithm}

Consider a branch $B=(S,g,l)$ with $g$ as a $t$-plex {\chengB with $t\ge 3$}.
{\cheng The} 
procedure for listing $k$-clique inside, called \texttt{kCtPlex}, differs in the way of branching (i.e., forming new branches). 
%
{\cheng Specifically, it branches} based on {\cheng the} inverse graph $g_{inv}$
{\cheng instead of $g$}. 
{\cheng The rationale is that since $g$ is a $t$-plex and tends to be dense, its inverse graph $g_{inv}$ would be sparse. As a result, branching on $g_{inv}$ would} run empirically faster. Below, we give the details.

Specifically, it maintains the inverse graph $g_{inv}$ of $g$ and represents a branch $(S,g,l)$ by the new form of $(S,C,l)$ where $C$ is the set of vertices in $g$, i.e., $C=V(g)$. We note that the new form omits the information of edges in $g$, which is instead stored in $g_{inv}$. Consider the branching step of \texttt{kCtPlex} at a branch $(S,C,l)$. Let $\langle v_1,v_2,\cdots,v_{|C|} \rangle$ be an arbitrary ordering of vertices in $C$. Then, the branching step would produce $|C|$ new sub-branches. The $i$-th branch, denoted by $B_i=(S_i,C_i,l_i)$, includes $v_i$ to $S$ and excludes $\{v_1,v_2,..,v_i\}$ from $C$. Formally, for $1\leq i \leq |C|$, we have
\begin{equation}
\label{eq:kCtPlex}
    S_i=S\cup \{v_i\},\ C_i=C\setminus \{v_1,v_2,\cdots,v_i\} \setminus N(v_i,g_{inv}),\ l_i=l-1.
\end{equation}
Note that we need to remove from $C$ those vertices that are not adjacent to $v_i$ in $g$ (since they cannot form any $k$-clique with $v_i$) and they connect to $v_i$ in $g_{inv}$. Clearly, all $k$-cliques in $(S,C,l)$ will be listed exactly once after branching. 
We note that the branching strategy we used in \texttt{kCtPlex} differs from that for \texttt{EBBkC} (and that for \texttt{VBBkC}). Specifically, the former (resp. the latter) is based on a sparse inverse graph $g_{inv}$ (resp. a dense $t$-plex $g$) and maintains a vertex set $C_i$ (resp. a graph instance $g_i$) for the produced branches. We remark that the former (correspondingly, the early stop strategy with $t$ at least 3) runs faster than the latter in practice, as verified in our experiments.


}


We summarize the procedure \texttt{kCtPlex} in Algorithm~\ref{alg:k-plex}. In particular, it also utilizes the combinatorial technique for boosting the performance (lines 8-10). Specifically, it figures out the set of vertices, denoted by $I$,  in $g_{inv}$, which are disconnected from all others (line 2). Consider a $k$-clique in $B$. It may involve $c$ vertices in $I$ where $c\in [0,\min\{|I|,k\}]$. Hence, we remove from $V(g_{inv})$ those vertices in $I$ for branching at line 3 while adding them back to a $k$-clique at lines 8-10. It is not difficult to verify the correctness. Due to the page limit, we include the time complexity analysis of \texttt{kCtPlex} in the technical report~\cite{wang2024technical}.

\smallskip
\noindent\textbf{Remark.}
(1) With the early termination strategy, the BB algorithms retain the same time complexity provided before but run practically faster as verified in the experiments. (2) The early termination strategy is supposed to set a small threshold of $t$ so as to apply the alternative procedures only on dense graph instances (i.e., $t$-plexes). We test different choices of $t$ in the experiments; the results suggest that \texttt{EBBkC} with $t$ being set from 2 to 5 runs comparably faster than other choices while the best one among them varies for different settings of $k$.

%% file: exp.tex
\section{Experiments}
\label{sec:exp}

\subsection{Experimental Setup}
\label{subsec:setup}

\begin{table}[t]
\vspace{-3mm}
\footnotesize
    \centering
    \caption{Dataset Statistics.}
\vspace{-2mm}
    \label{tab:data}
    \begin{tabular}{r|cc|cccc}
    \hline
        \textbf{Graph (Name)} & $|V|$ & $|E|$ & $\Delta$ & $\delta$ & $\tau$ & $\omega$ \\
    \hline
        \textsf{nasasrb (NA)} & 54,870 & 1,311,227 & 275 & 35 & 22 & 24  \\
        \textsf{fbwosn (FB)} & 63,731 & 817,090 & 2K & 52 & 35 & 30  \\
        \textbf{\textsf{wikitrust (WK)}} & 138,587 & 715,883 & 12K  & 64 & 31 & 25 \\
        \textsf{shipsec5 (SH)} & 179,104 & 2,200,076 & 75  & 29 & 22 & 24  \\
        \textsf{socfba (SO)} & 3,097,165 & 23,667,394 & 5K & 74 & 29 & 25 \\
        \textbf{\textsf{pokec (PO)}} & 1,632,803 & 22,301,964 & 15K  & 47 & 27 & 29 \\
        \textsf{wikicn (CN)} & 1,930,270 & 8,956,902 & 30K  & 127 & 31 & 33 \\    
        \textsf{baidu (BA)} & 2,140,198 & 17,014,946 & 98K & 82 & 29 & 31 \\
    \hline
        \textsf{websk (WE)} & 121,422 & 334,419 & 590 & 81 & 80 & 82  \\
        \textsf{citeseer (CI)} & 227,320 & 814,134 & 1K & 86 & 85 & 87 \\
        \textbf{\textsf{stanford (ST)}} & 281,904 & 1,992,636 & 39K & 86 & 61 & 61 \\
        \textsf{dblp (DB)} & 317,080 & 1,049,866 & 343 & 113 & 112 & 114 \\
        \textsf{dielfilter (DE)} & 420,408 & 16,232,900 & 302 & 56 & 43 & 45 \\ 
        \textsf{digg (DG)} & 770,799 & 5,907,132 & 18K & 236 & 72 & 50  \\
        \textsf{skitter (SK)} & 1,696,415 & 11,095,298 & 35K  & 111 & 67 & 67 \\
        \textbf{\textsf{orkut (OR)}} & 2,997,166 & 106,349,209 & 28K  & 253 & 74 & 47 \\ 
    \hline
        \textsf{allwebuk (UK)}  & 18,483,186  & 261,787,258 & 3M & {943} & {942} & {944}  \\
        {\textsf{clueweb (CW)}} & {147,925,593} & {446,766,953} &{1M} & {192} & {83} &  {56} \\
        {\textsf{wikipedia (WP)}} & {25,921,548} & {543,183,611} & {4M} & {1120} & {426} & {428} \\
    \hline
    \end{tabular}
\vspace{-2mm}
\end{table}

\begin{figure*}[t]
\vspace{-4mm}
\subfigure[\textsf{NA}]{
    \includegraphics[width=0.24\textwidth]{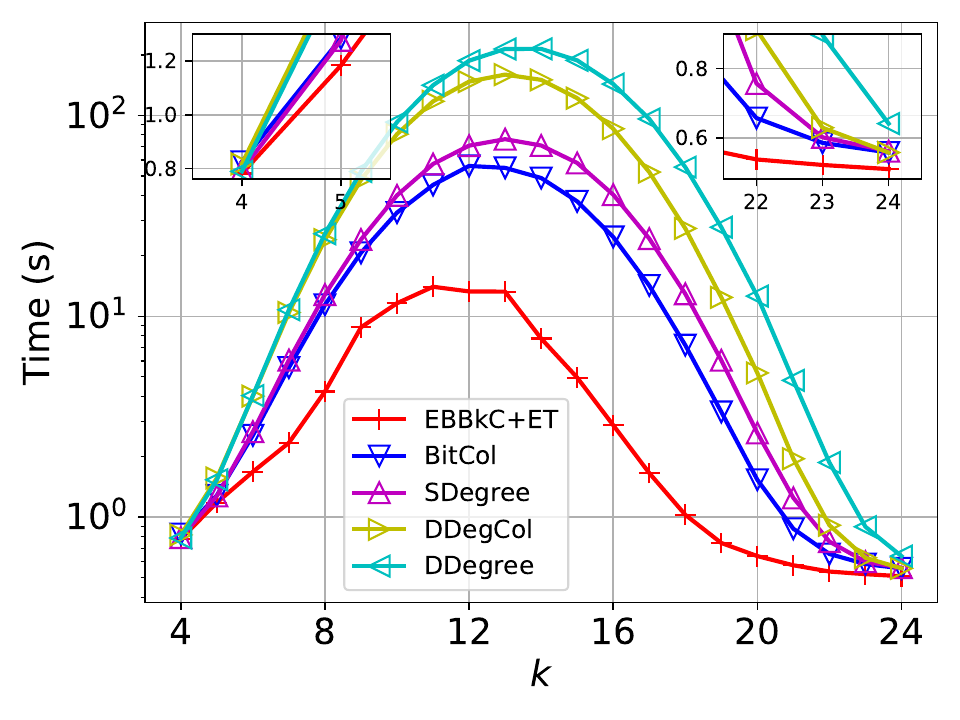}}
\subfigure[\textsf{FB}]{
    \includegraphics[width=0.24\textwidth]{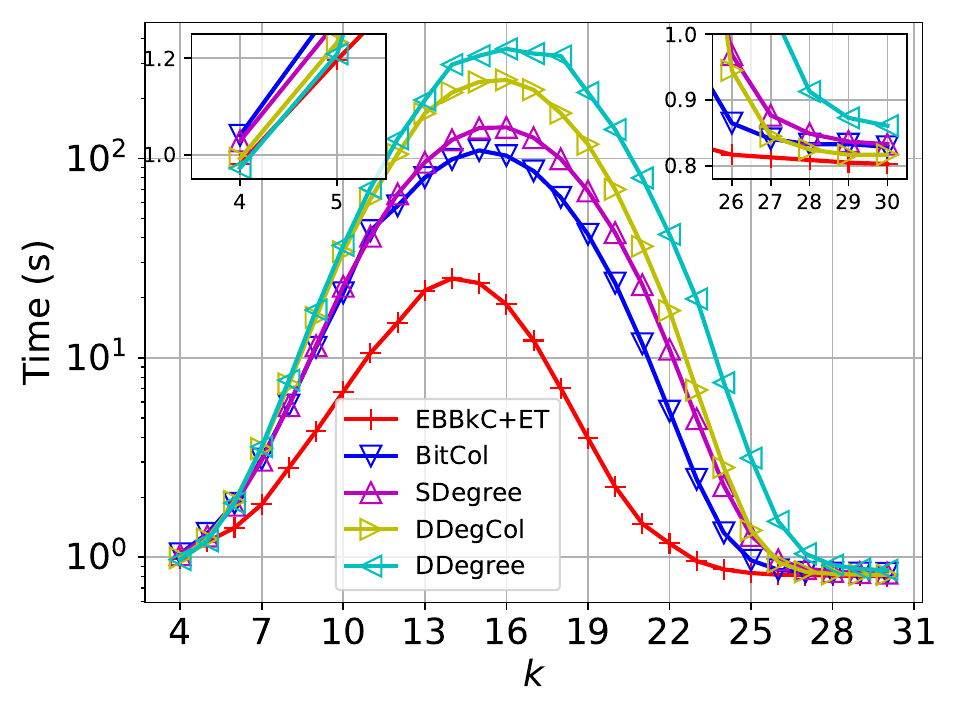}}
\subfigure[\textsf{WK}]{
    \includegraphics[width=0.24\textwidth]{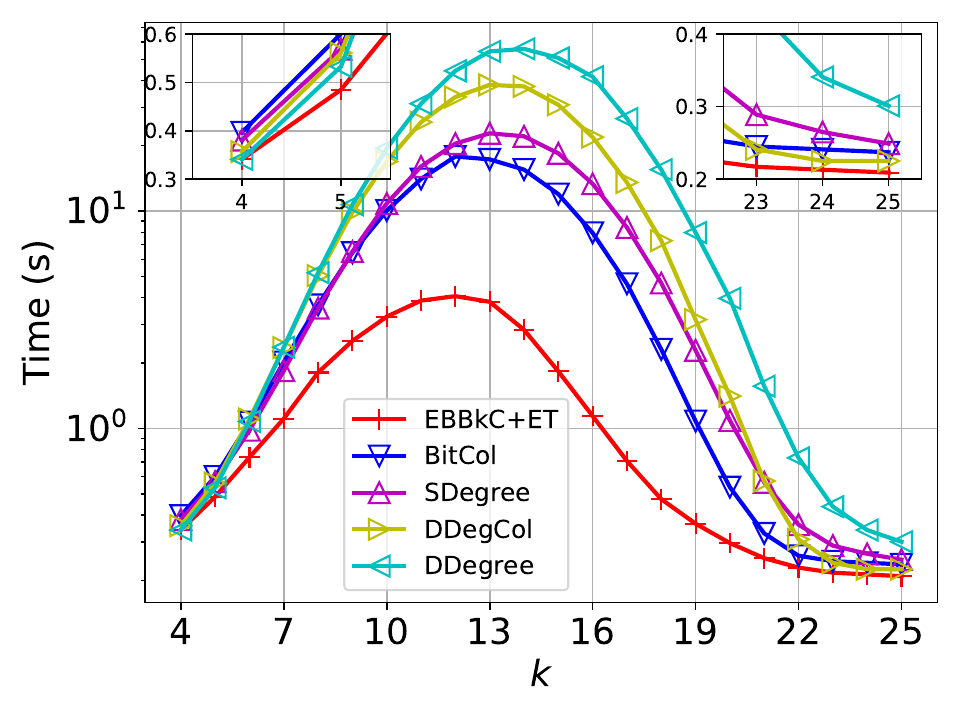}}
\subfigure[\textsf{SH}]{
    \includegraphics[width=0.24\textwidth]{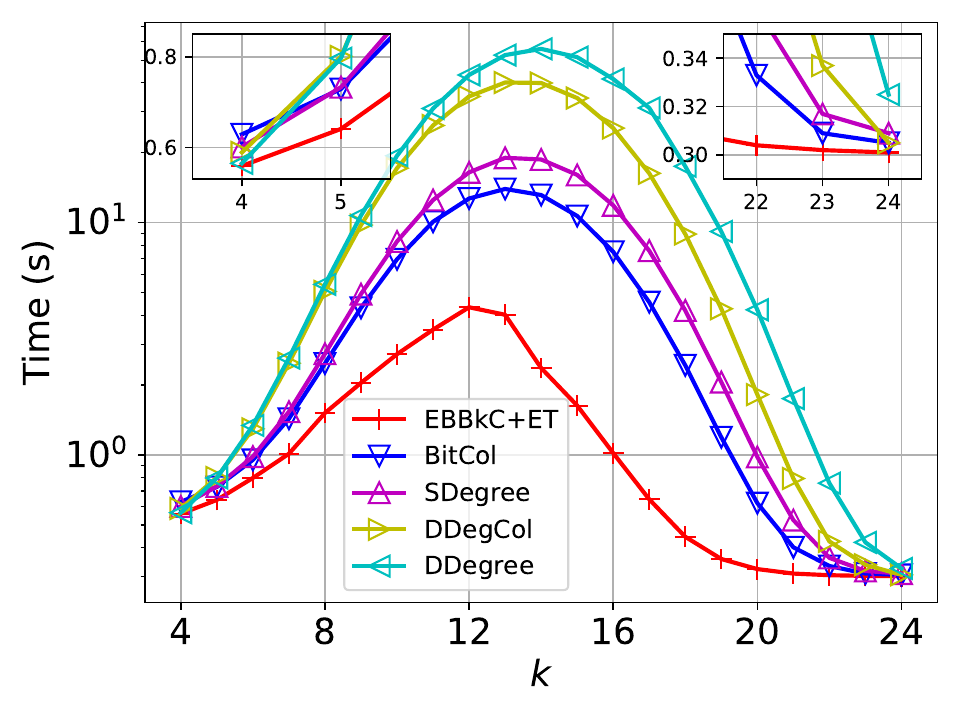}}
\subfigure[{\textsf{SO}}]{
\label{subfig:a}
    \includegraphics[width=0.24\textwidth]{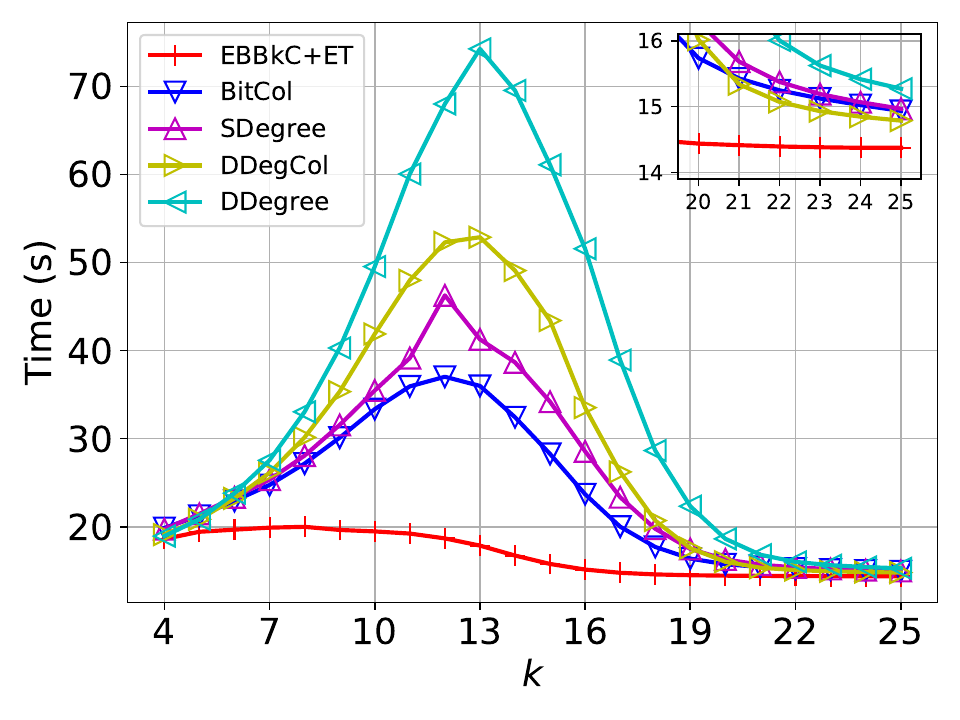}}
\subfigure[\textsf{PO}]{
    \includegraphics[width=0.24\textwidth]{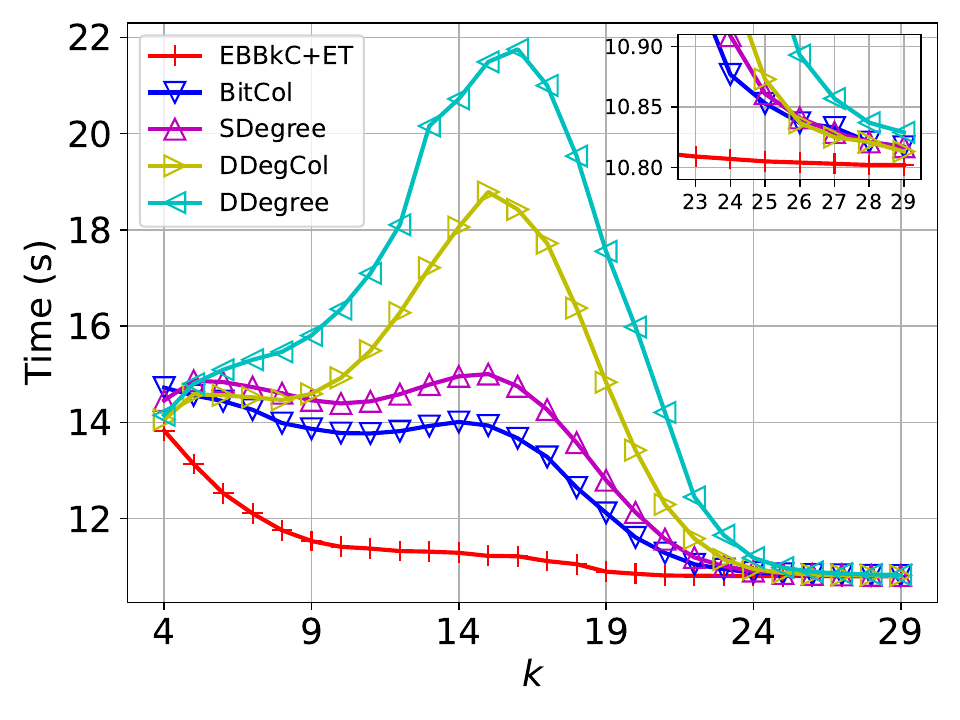}}
\subfigure[\textsf{CN}]{
    \includegraphics[width=0.24\textwidth]{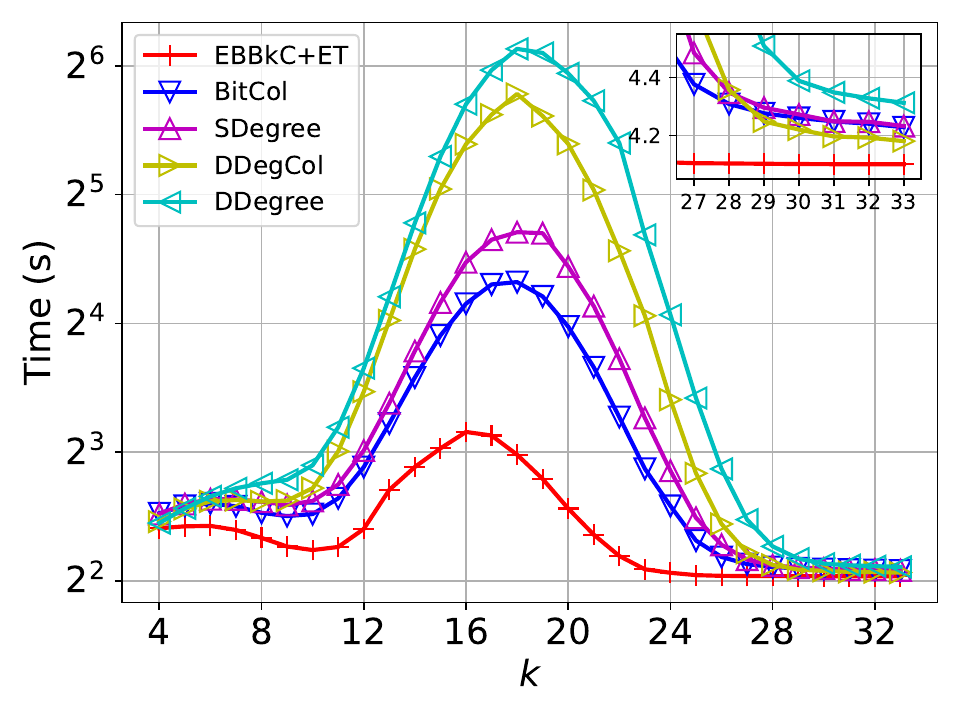}}
\subfigure[\textsf{BA}]{
    \includegraphics[width=0.24\textwidth]{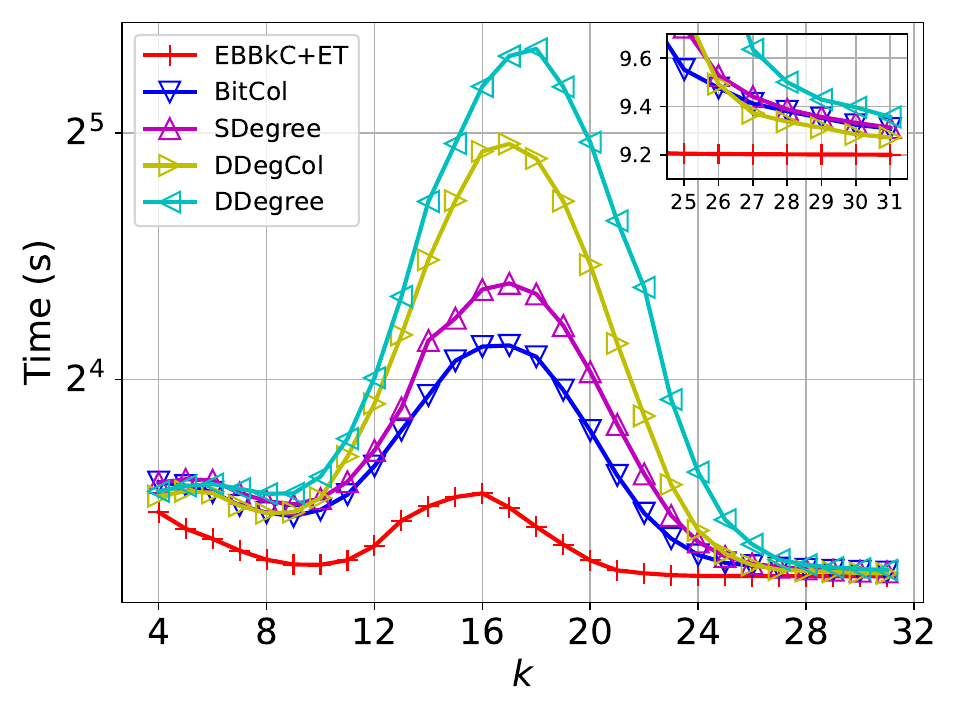}}
\vspace{-3mm}
\caption{Comparison with baselines on the small-$\omega$ graphs, varying $k$ from 4 to $\omega$.}
\vspace{-2mm}
\label{fig:real-small}
\end{figure*}

\begin{figure*}[t]
\vspace{-2mm}
\subfigure[\textsf{WE}]{
    \includegraphics[width=0.24\textwidth]{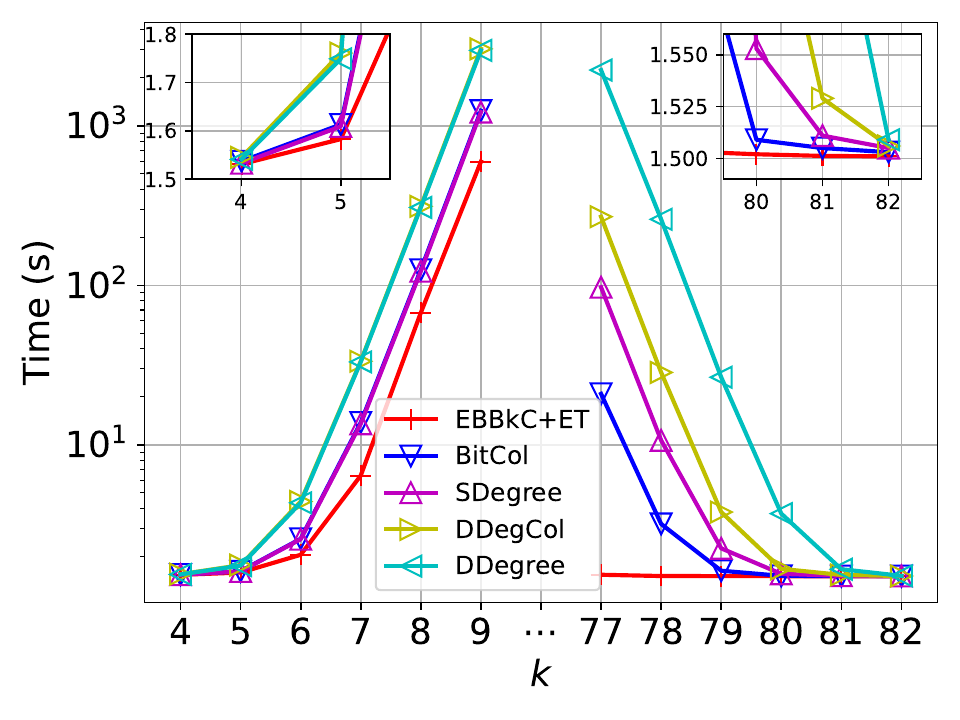}}
\subfigure[\textsf{CI}]{
    \includegraphics[width=0.24\textwidth]{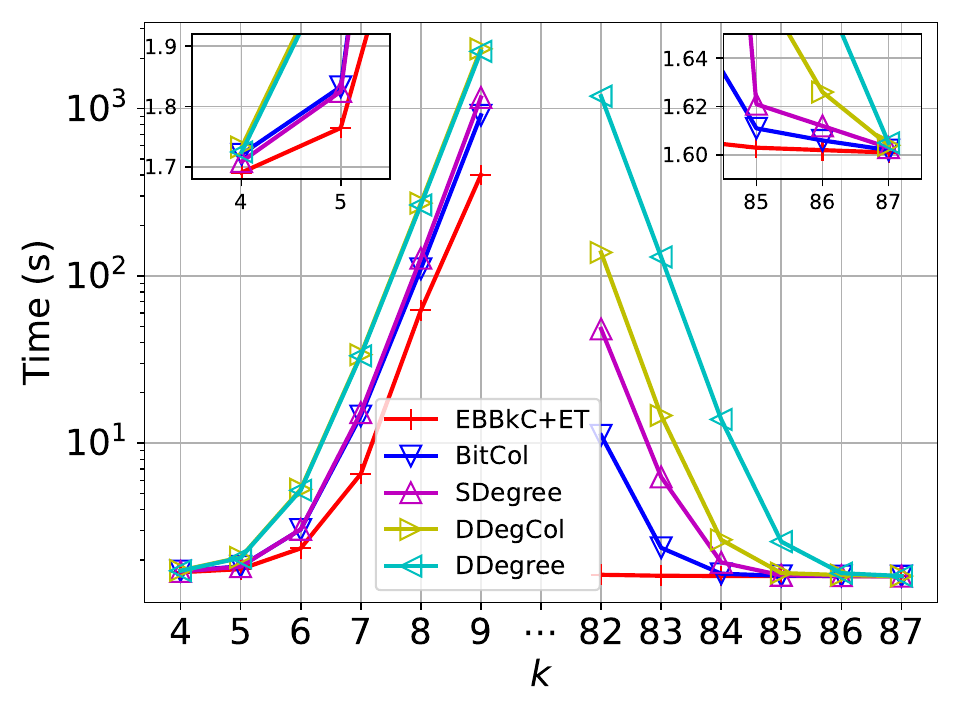}}
\subfigure[\textsf{ST}]{
    \includegraphics[width=0.24\textwidth]{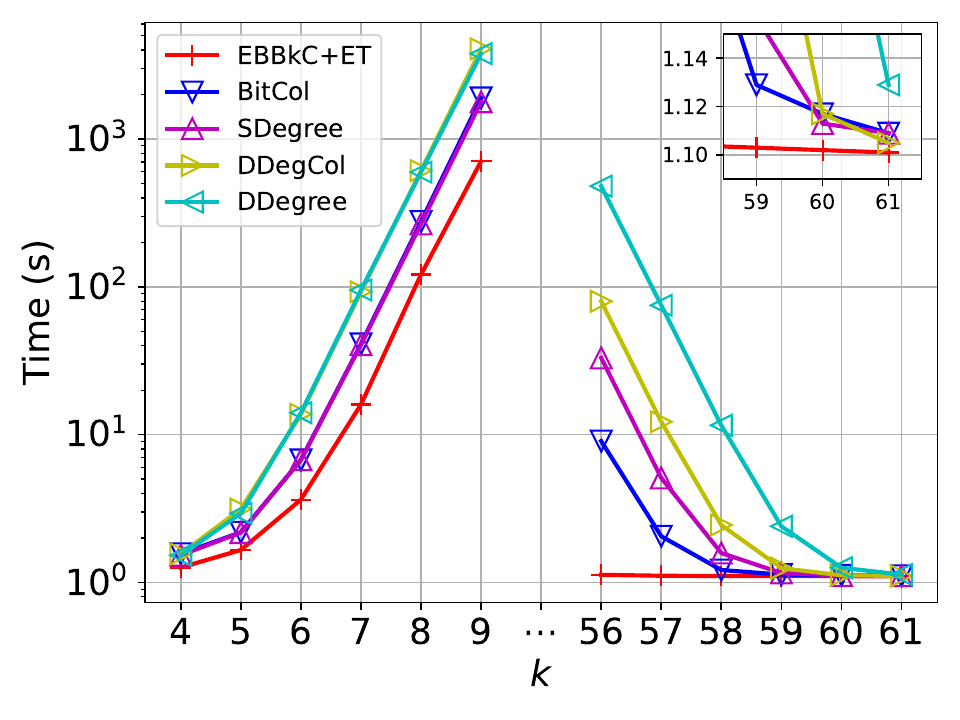}}
\subfigure[\textsf{DB}]{
    \includegraphics[width=0.24\textwidth]{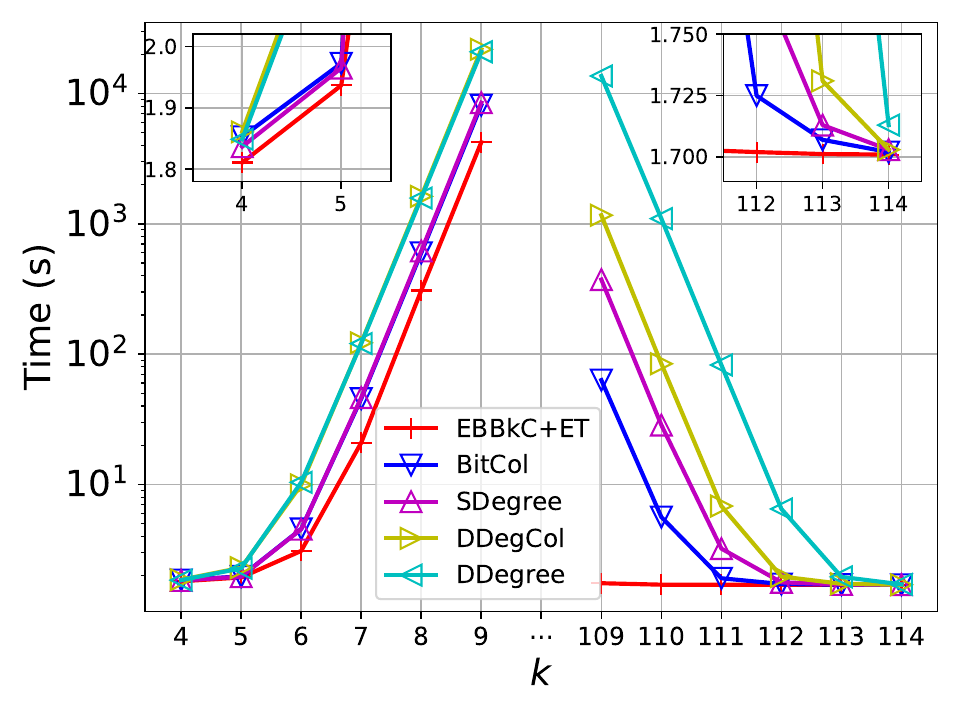}}
\subfigure[\textsf{DE}]{
    \includegraphics[width=0.24\textwidth]{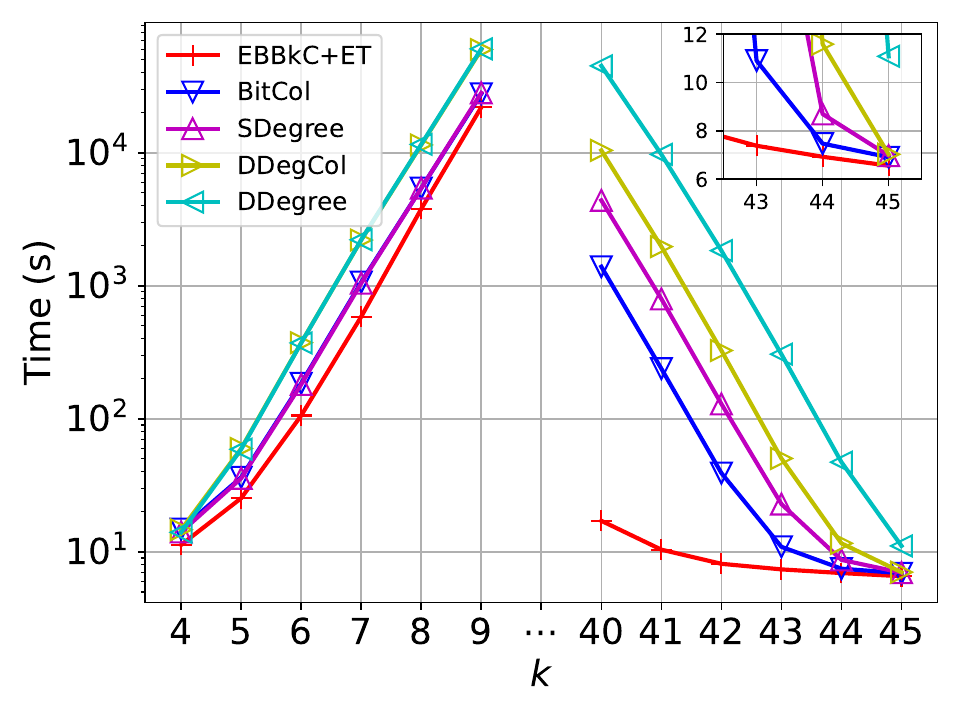}}
\subfigure[\textsf{DG}]{
    \includegraphics[width=0.24\textwidth]{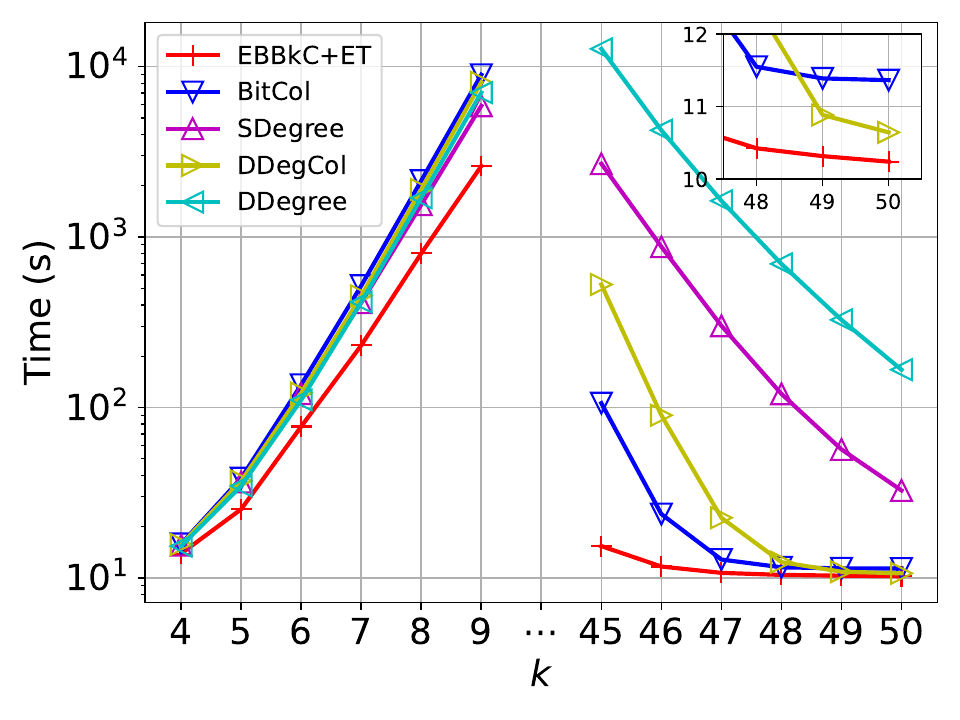}}
\subfigure[\textsf{SK}]{
    \includegraphics[width=0.24\textwidth]{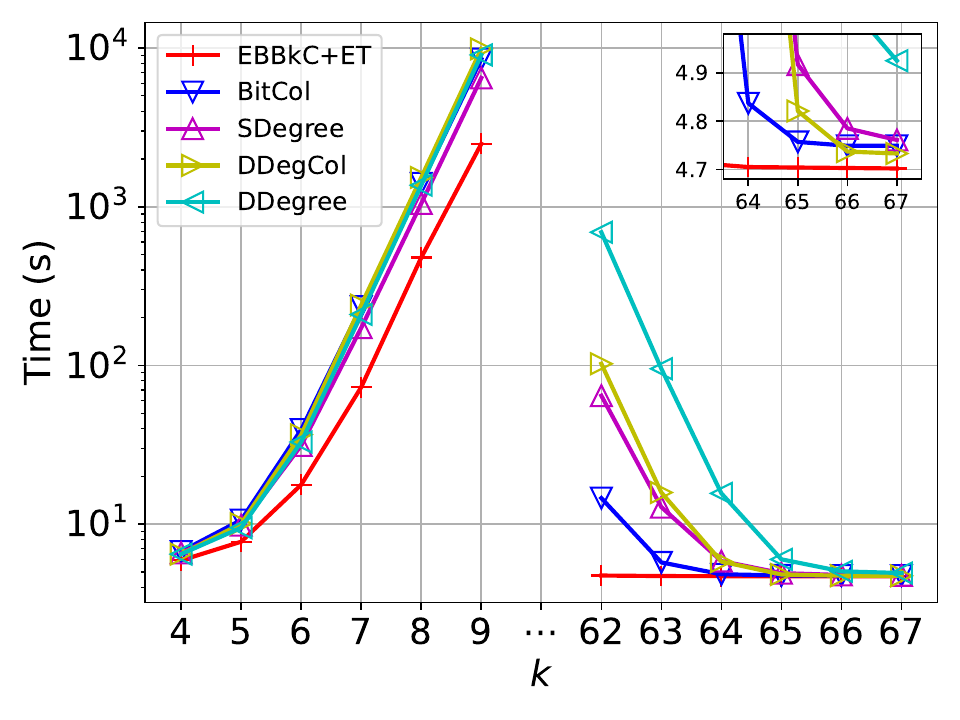}}
\subfigure[\textsf{OR}]{
    \includegraphics[width=0.24\textwidth]{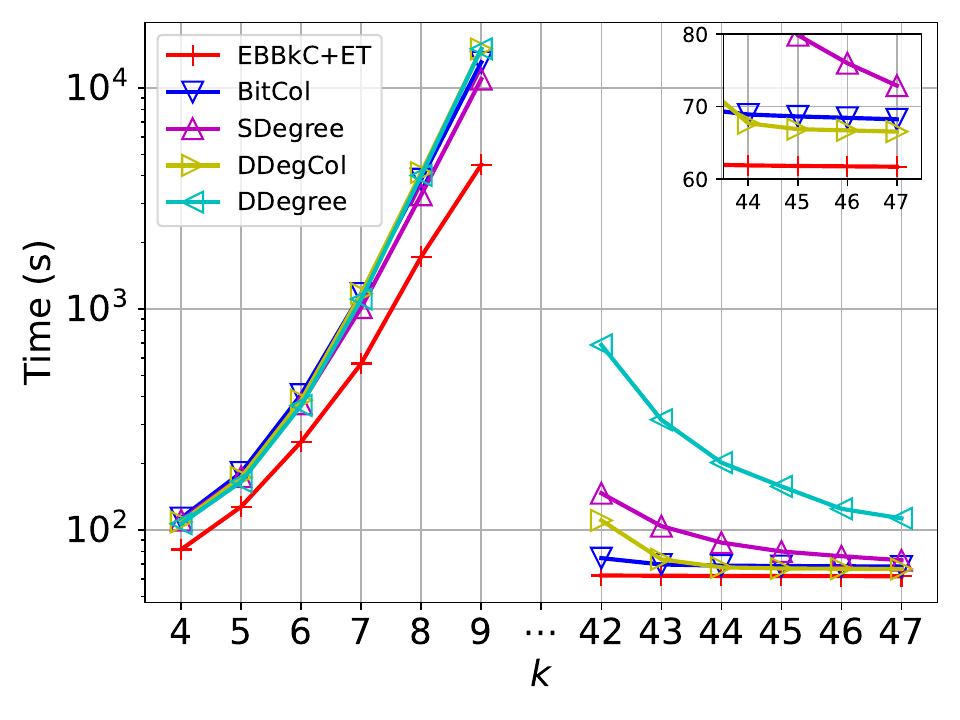}}
\caption{Comparison with baselines on the large-$\omega$ graphs, varying $k$ from 4 to 9 and from $\omega-5$ to $\omega$.}
\label{fig:real-large}
\vspace{-3mm}
\end{figure*}

\smallskip\noindent
\textbf{Datasets.}
We use {19} real datasets in our experiments, which can be obtained from an open-source network repository \cite{rossi2015network}. For each graph, we ignore the directions, weights and self-loops (if any) at the very beginning. Following the existing study \cite{li2020ordering}, we divide the real datasets into two groups based on the size of a maximum clique $\omega$: small-$\omega$ graphs and large-$\omega$ graphs. For small-$\omega$ graphs, we 
list all $k$-cliques for all $k$, while for large-$\omega$ graphs, we 
only list $k$-cliques for small $k$ values {\cheng and} large $k$ values which are near $\omega$. 
We collect the graph statistics and report the maximum degree $\Delta$, the degeneracy number $\delta$, the 
{\chengB truss related number}
$\tau$ and the maximum clique size $\omega$, which are shown in Table~\ref{tab:data}. 
{\cheng We select four datasets, namely \textsf{WK}, \textsf{PO}, \textsf{ST} and \textsf{OR}, which are bold in the table, as default ones since they cover different size of graphs.}

\smallskip\noindent
\textbf{Baselines and Metrics.} 
We choose \texttt{EBBkC-H} as the default edge-oriented branching-based BB framework for comparison, and denote it by \texttt{EBBkC} for brevity in the experimental results.
We compare our algorithm \texttt{EBBkC+ET} with four existing algorithms, namely \texttt{DDegCol} \cite{li2020ordering}, \texttt{DDegree} \cite{li2020ordering}, \texttt{SDegree} \cite{yuan2022efficient} and \texttt{BitCol} \cite{yuan2022efficient} in terms of running time. Specifically, \texttt{EBBkC+ET} employs the edge-oriented branching-based BB framework with hybrid edge ordering, and applies the early-termination technique. 
{The running times of all algorithms reported in our paper include the time costs of (1) conducting the pre-processing techniques (if any) and (2) generating orderings of vertices. For our early termination technique, we set the parameter $t=2$ when $k\le \tau/2$ and $t=3$ when $\tau/2<k\le \omega$ for early-termination.}
All baselines are the state-of-the-art algorithms for $k$-cliques listing with vertex-oriented {\chengB branching-based} BB framework. 
%
{\kaixin Following the existing study \cite{li2020ordering}, we vary $k$ from 4 since $k$-clique listing problem reduces to triangle listing problem when $k=3$ and there are efficient algorithms \cite{latapy2008main,ortmann2014triangle} for triangle listing which run in polynomial time.}


%
%

\smallskip\noindent
\textbf{Settings.}
The source codes of all algorithms are written in C++ and the experiments are conducted on a Linux machine with a 2.10GHz Intel CPU and 128GB memory. 
{\cheng We note that we have not utilized SIMD instructions for data-level parallelism in our implementation, despite the potential to further enhance the acceleration of our algorithm.}
We set the time limit as 24 hours (i.e., 86,400 seconds) and  the running time of any algorithm that exceeds the time limit is recorded with ``{INF}''. Implementation codes and datasets can be found via this link \url{https://github.com/wangkaixin219/EBBkC}. 

\subsection{Experimental Results}
\label{subsec:result}

\smallskip\noindent
\textbf{(1) Comparison among algorithms (on small-$\omega$ graphs).}
Figure~\ref{fig:real-small} shows the results of listing $k$-cliques on small-$\omega$ graphs. We observe that \texttt{EBBkC+ET} (indicated with red lines) runs faster than all baselines on all datasets. 
This is consistent with our theoretical analysis that the 
{\chengB time complexity}
of \texttt{EBBkC+ET} is {\cheng better than those of the baseline algorithms.}
Besides, on \textsf{PO}, \textsf{CN} and \textsf{BA}, we observe that the running time of \texttt{EBBkC+ET} first decreases when $k$ is small. 
There are two possible reasons: (1) on some dataset, the number of $k$-cliques decreases as $k$ increases when $k$ is small. On \textsf{BA}, for example, the number of 4-cliques ($k=4$) is nearly 28M while the number of 7-cliques ($k=7$) is 21M; (2) {\cheng as $k$ increases, a larger} number of branches can be pruned by the size constraint with the truss-based edge ordering,
which provides the opportunity to run faster. To see this, we collect the number of promising branches after the branching step with the truss-based edge ordering. On \textsf{PO}, there are 11M branches left when enumerating 4-cliques ($k=4$) while there are only less than 1M branches left when enumerating 10-cliques ($k=10$). 

\smallskip\noindent
\textbf{(2) Comparison among algorithms (on large-$\omega$ graphs).} 
Figure~\ref{fig:real-large} shows the results of listing $k$-cliques on large-$\omega$ graphs. We omit the results for some values of $k$ since the time costs are beyond 24 hours due to the large number of $k$-cliques. We find that \texttt{EBBkC+ET} still runs the fastest. It is worthy noting that \texttt{EBBkC+ET} can greatly improve the efficiency by 1-2 orders of magnitude over the baselines when $k$ is near the size of a maximum clique $\omega$, {\kaixin e.g., 
\texttt{EBBkC+ET} runs 9.2x and 97.7x faster than \texttt{BitCol} 
on \textsf{DB} (when $k=109$) and on \textsf{DE} (when $k=40$), respectively}. The reasons are two-fold: (1) when $k$ is near $\omega$, a large number of branches can be pruned by the size constraint with the truss-based edge ordering,
and as a result, the remaining branches are relatively dense; (2) \texttt{EBBkC+ET} can quickly enumerate cliques within a dense structure, e.g., $t$-plex, without {\cheng making branches for} the search space, which dramatically reduces the running time.

{
\smallskip\noindent
\textbf{(3) Ablation studies.}
We compare two variants of our method, 
namely \texttt{EBBkC+ET} (the full version) and \texttt{EBBkC} (the full version without the early termination), with two \texttt{VBBkC} algorithms, namely \texttt{DDegCol+} (the \texttt{DDegCol} algorithm with Rule (2) proposed in this paper) and \texttt{BitCol+} (the \texttt{BitCol} algorithm with Rule (2) and without SIMD-based implementations).
We note that \texttt{DDegCol+} and \texttt{BitCol+} correspond to the SOTA \texttt{VBBkC} algorithms without SIMD-based implementations.
We report the results in Figure~\ref{fig:variant} and have the following observations. 
\underline{First}, \texttt{DDegCol+} and \texttt{BitCol+} have very similar running times, which shows that pre-processing techniques in the BitCol paper are not effective. \underline{Second}, \texttt{EBBkC} runs clearly faster than \texttt{DDegCol+} and \texttt{BitCol+}, which demonstrates the clear contribution of our edge-oriented BB framework (note that \texttt{EBBkC} and \texttt{DDegCol+} differ only in their frameworks). We note that it is not fair to compare \texttt{EBBkC} with \texttt{BitCol} directly since the latter is based on SIMD-based implementations while the former is not. \underline{Third}, \texttt{EBBkC+ET} runs clearly faster than \texttt{EBBkC}, which demonstrates the clear contribution of the early termination technique. For example, on dataset \textsf{WK}, the edge-oriented BB framework and early-termination contribute 28.5\% and 71.5\% to the efficiency improvements over \texttt{BitCol+} when $k=13$, respectively.
We also show the time costs of generating the truss-based ordering for edges in \texttt{EBBkC} and those of generating the degeneracy ordering for vertices in \texttt{VBBkC} in Table~\ref{tab:ordering}. We note that while the former are slightly larger than the latter, the overall time cost of \texttt{EBBkC} is smaller than that of \texttt{VBBkC} (as shown in Figure~\ref{fig:variant}).
}

\begin{table}[t]
\centering
    \caption{{Time for generating truss-based edge ordering and degeneracy ordering (unit: sec). }}
\vspace{-2mm}
    \label{tab:ordering}
    \begin{tabular}{|c|c|c|c|c|}
    \hline
         & \textsf{WK} & \textsf{PO} & \textsf{ST} & \textsf{OR} \\ \hline
    Truss (s) & 0.2 & 10.7 & 1.1 & 60.4 \\ \hline
    Degen. (s) & 0.1 & 7.3 & 0.6 & 53.3 \\ \hline
    \end{tabular}
\end{table}

\begin{figure}[t]
\subfigure[{\textsf{WK}}]{
    \includegraphics[width=0.23\textwidth]{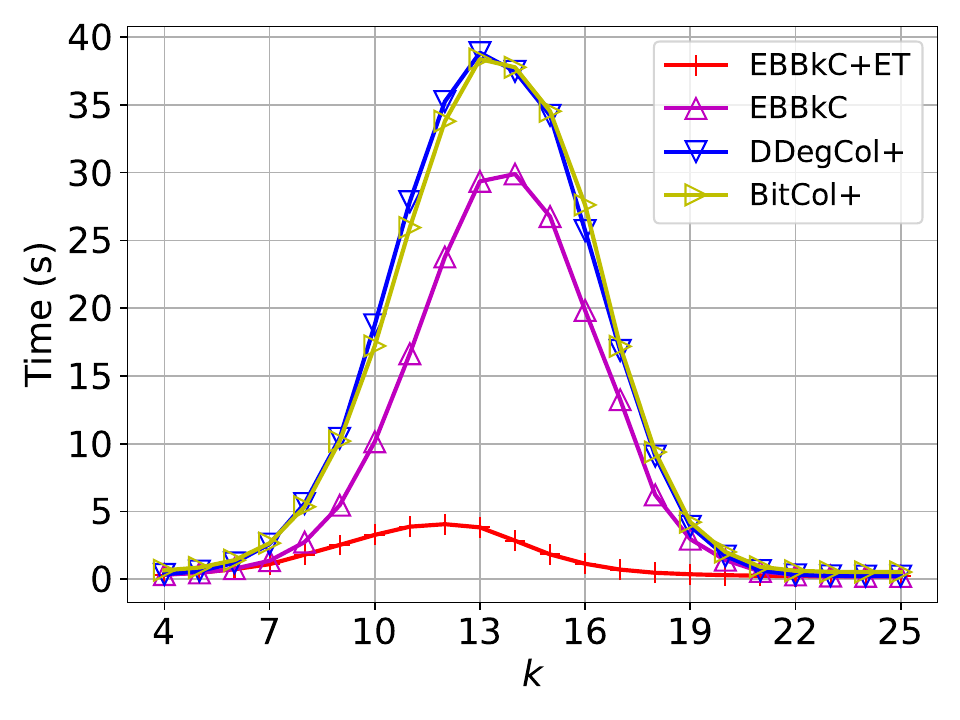}}
\subfigure[{\textsf{PO}}]{
    \includegraphics[width=0.23\textwidth]{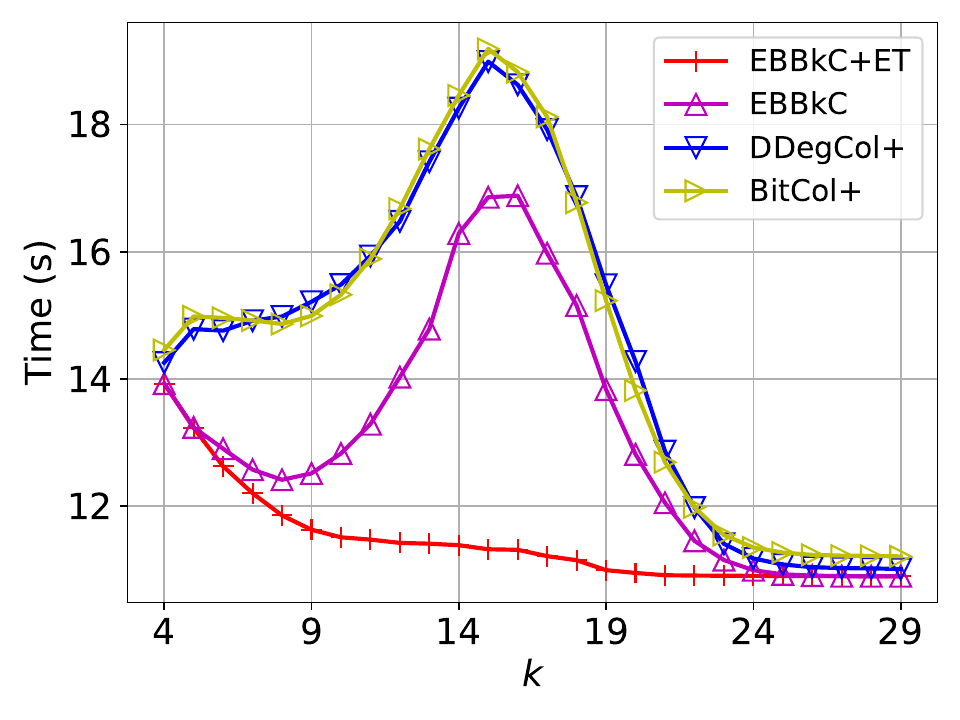}}
\vspace{-3mm}
\caption{{Ablation studies.} }
\vspace{-2mm}
\label{fig:variant}
\end{figure}

\begin{figure}[t]
 \vspace{-2mm}
\subfigure[\textsf{WK}]{
    \includegraphics[width=0.23\textwidth]{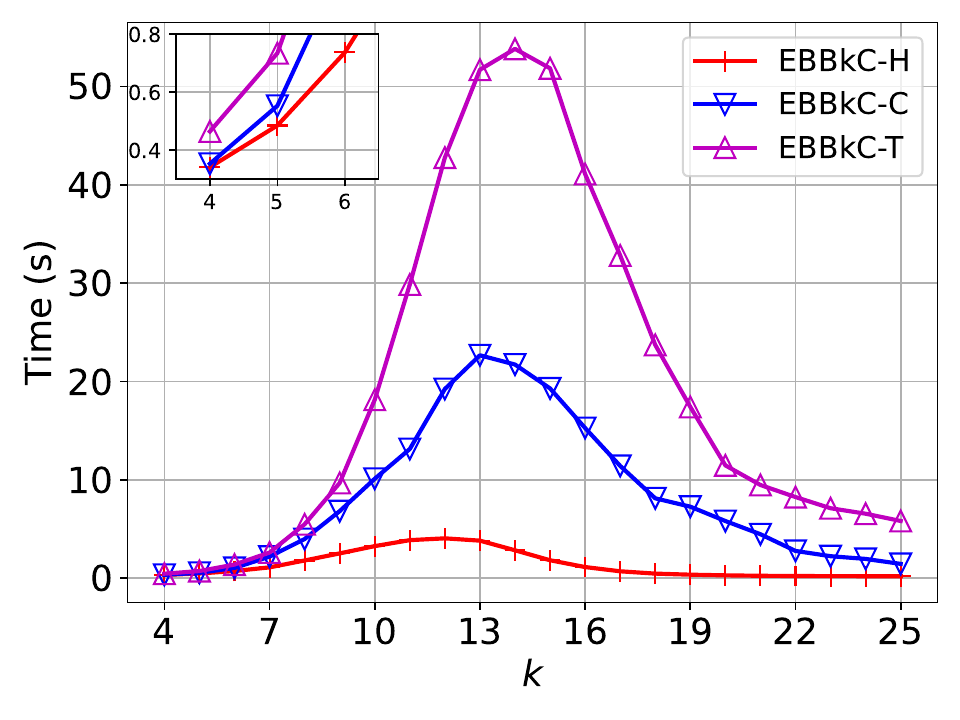}}
\subfigure[\textsf{PO}]{
    \includegraphics[width=0.23\textwidth]{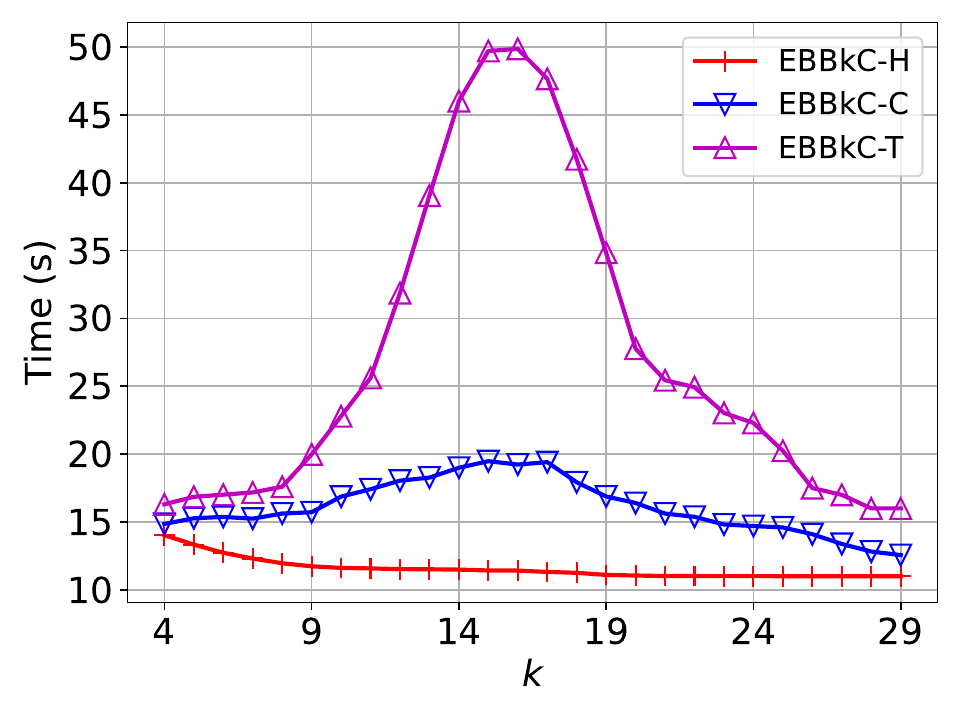}}
\vspace{-3mm}
\caption{Effects of the {\cheng edge ordering}. 
}
\vspace{-2mm}
\label{fig:ebbkc-frame}
\end{figure}

{\kaixin
\smallskip\noindent
\textbf{(4) Effects of the {\cheng edge ordering} (comparison among  \texttt{EBBkC-T}, \texttt{EBBkC-C} and \texttt{EBBkC-H}).}
%
For the sake of fairness, we employ all color-based pruning rules for \texttt{EBBkC-C} and \texttt{EBBkC-H} frameworks and employ the early-termination technique for all frameworks. The results are shown in Figure~\ref{fig:ebbkc-frame}. Consider \texttt{EBBkC-H} and \texttt{EBBkC-T}. Although both frameworks have the same time complexity, \texttt{EBBkC-H} runs much faster since it can prune more unpromising search paths by color-based pruning rules than \texttt{EBBkC-T}. Consider \texttt{EBBkC-H} and \texttt{EBBkC-C}. \texttt{EBBkC-H} outperforms \texttt{EBBkC-C} since the largest sub-problem instances produced by \texttt{EBBKC-H} is smaller than that of \texttt{EBBkC-C}, which also conforms our theoretical analysis.}

\begin{figure}[t]
\vspace{-1mm}
\subfigure[\textsf{WK}]{
    \includegraphics[width=0.23\textwidth]{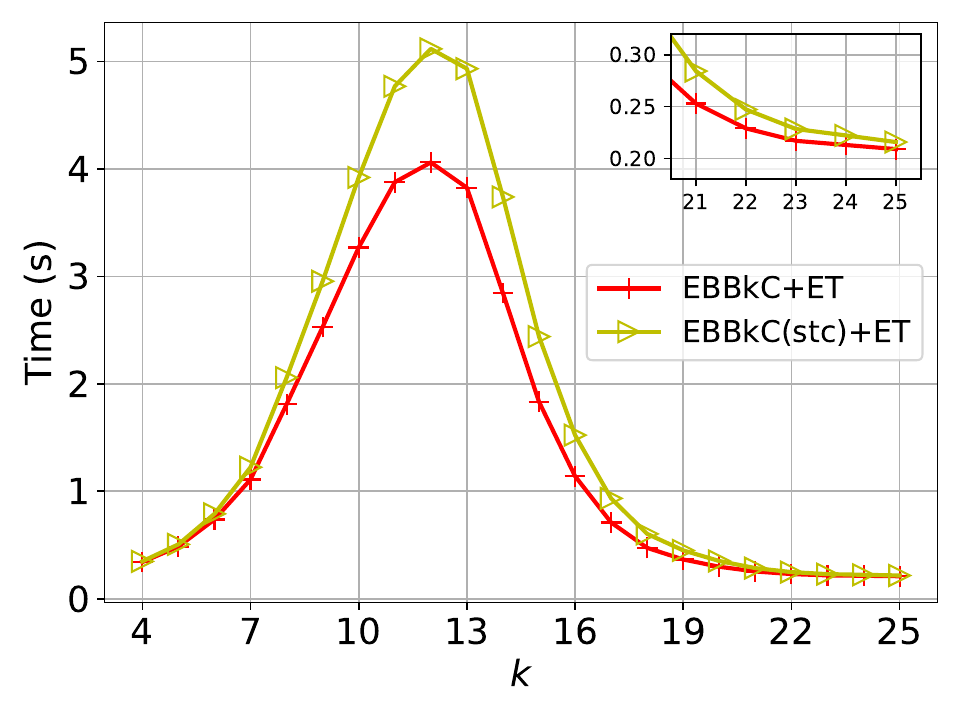}}
\subfigure[\textsf{PO}]{
    \includegraphics[width=0.23\textwidth]{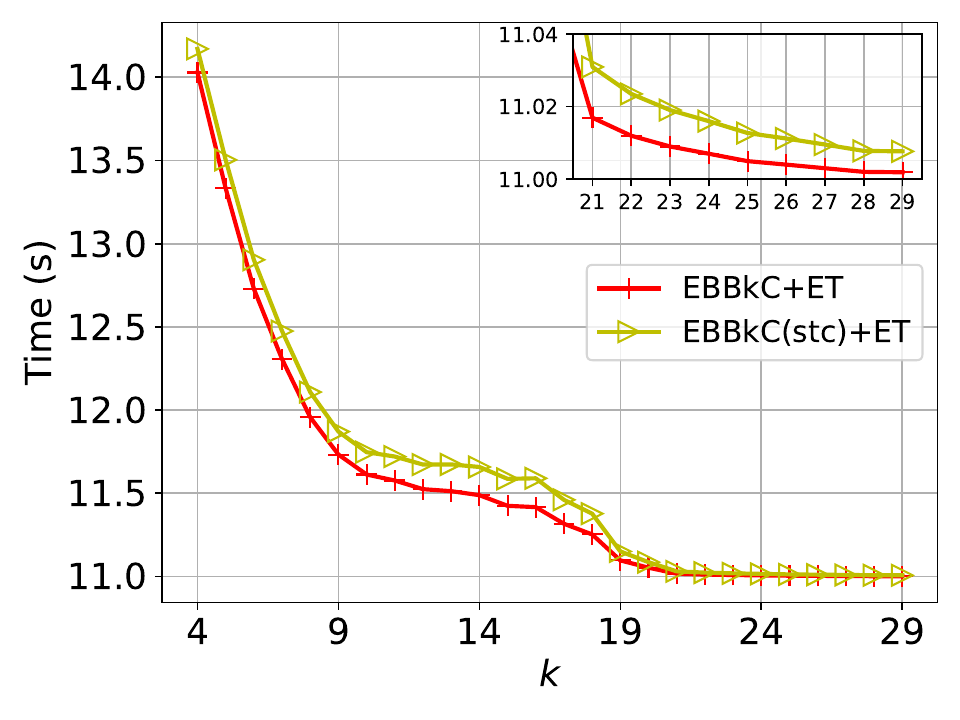}}
\vspace{-2mm}
\caption{Effects of the pruning rules (comparison between the algorithms w/ and w/o the Rule (2)).}
\vspace{-3mm}
\label{fig:color}
\end{figure}

\begin{figure}[t]
\centering
    \subfigure[Results on \textsf{WK} varying $k$]{
    \includegraphics[width=0.23\textwidth]{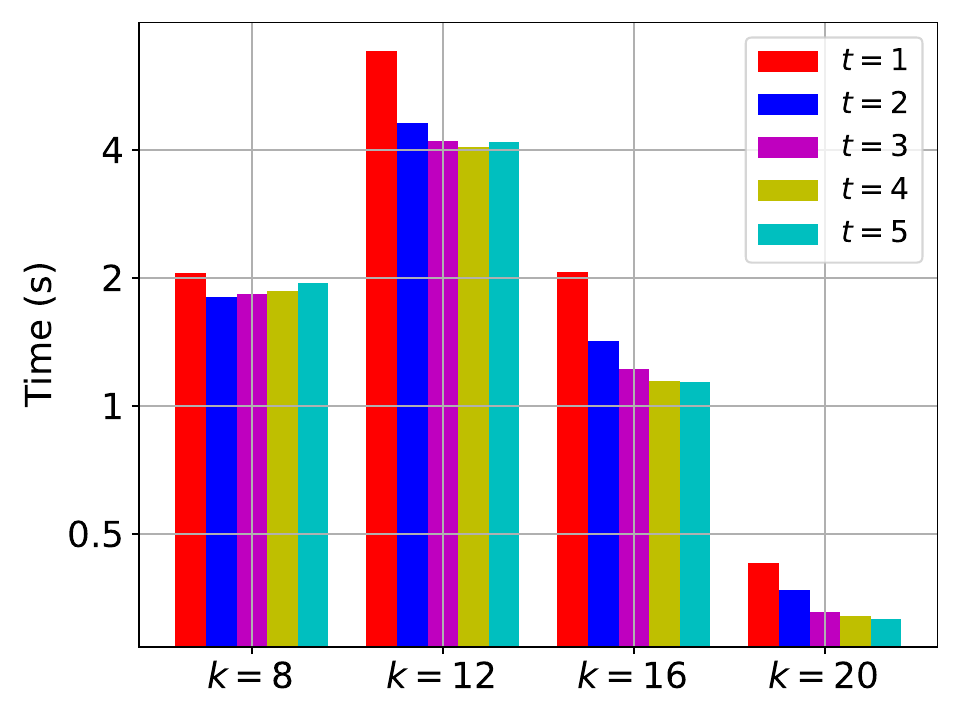}}
\subfigure[Results on \textsf{PO} varying $k$]{
    \includegraphics[width=0.23\textwidth]{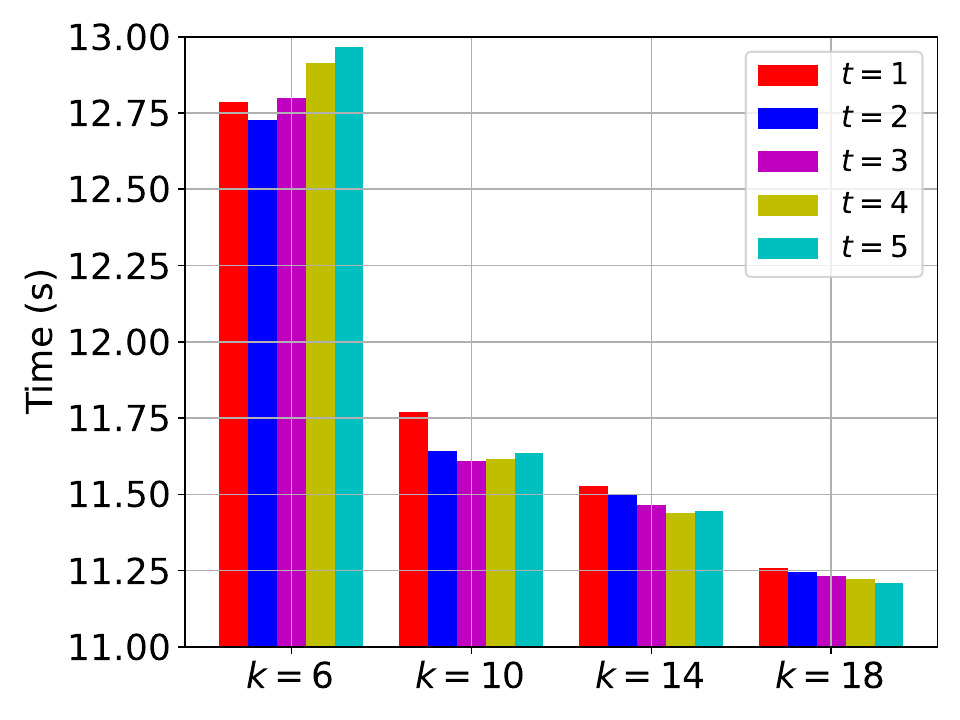}}
\vspace{-3mm}
\caption{Effects of early-termination technique (varying $t$).}
\label{fig:stop}
\end{figure}

\begin{figure}[t]
\vspace{-2mm}
\subfigure[{Results on \textsf{ST} ($k=8$)}]{
\label{subfig:parallel-stanford}
    \includegraphics[width=0.23\textwidth]{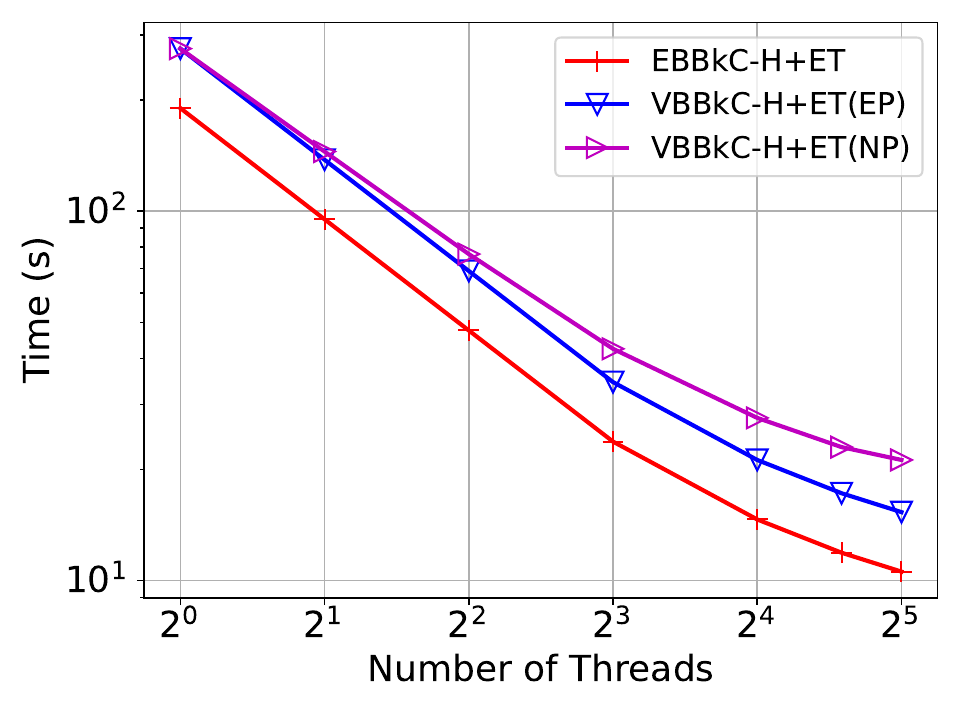}}
\subfigure[{Results on \textsf{OR} ($k=8$)}]{
\label{subfig:parallel-orkut}
    \includegraphics[width=0.23\textwidth]{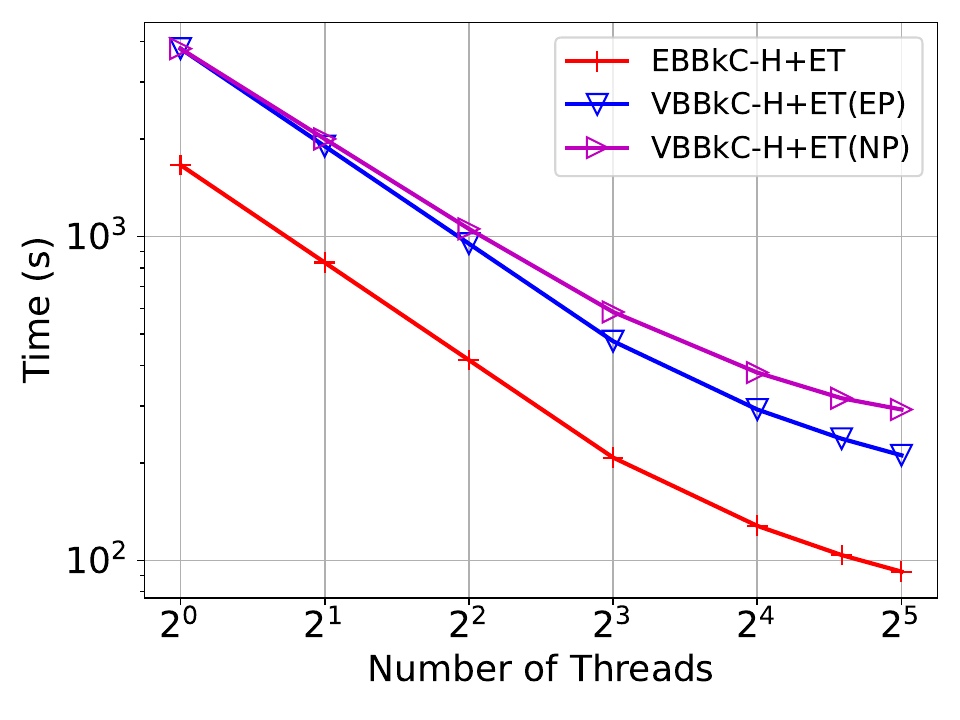}}
\vspace{-3mm}
\caption{{Comparison among different parallel schemes, varying the number of threads.}}
\label{fig:parallel}
\end{figure}

\begin{figure*}[t]
\vspace{-2mm}
\begin{minipage}{0.23\linewidth}
    \centering
     \includegraphics[width=\textwidth]{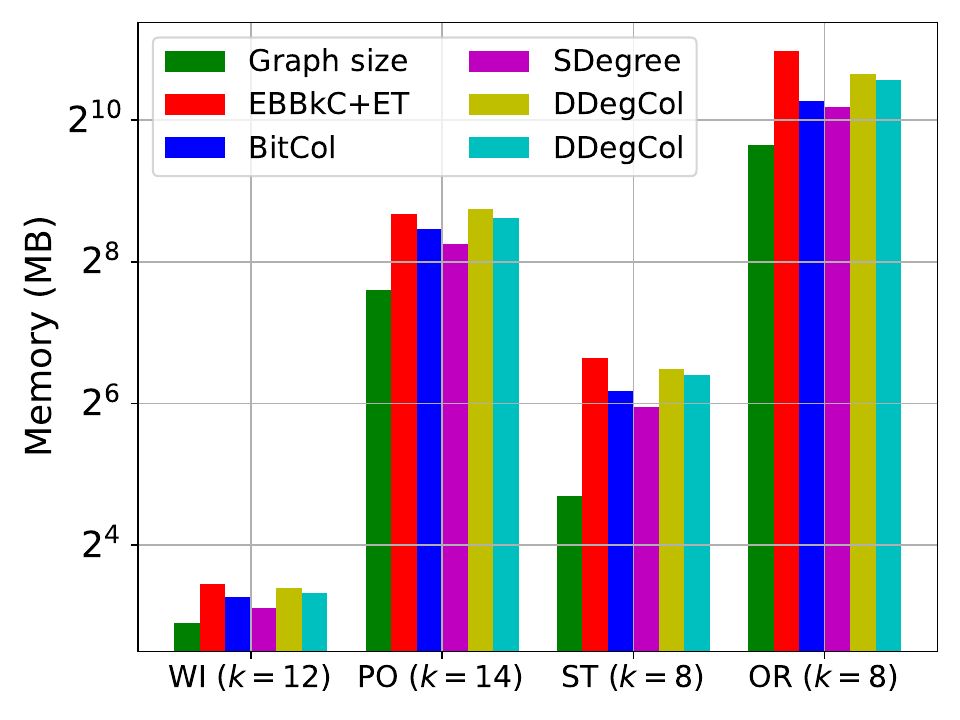}
     \caption{{Space costs.}}\label{fig:mem} 
\end{minipage}
\begin{minipage}{0.72\linewidth} 
\centering
    \subfigure[{\textsf{UK}}]{
\label{subfig:scale-uk}
    \includegraphics[width=0.32\textwidth]{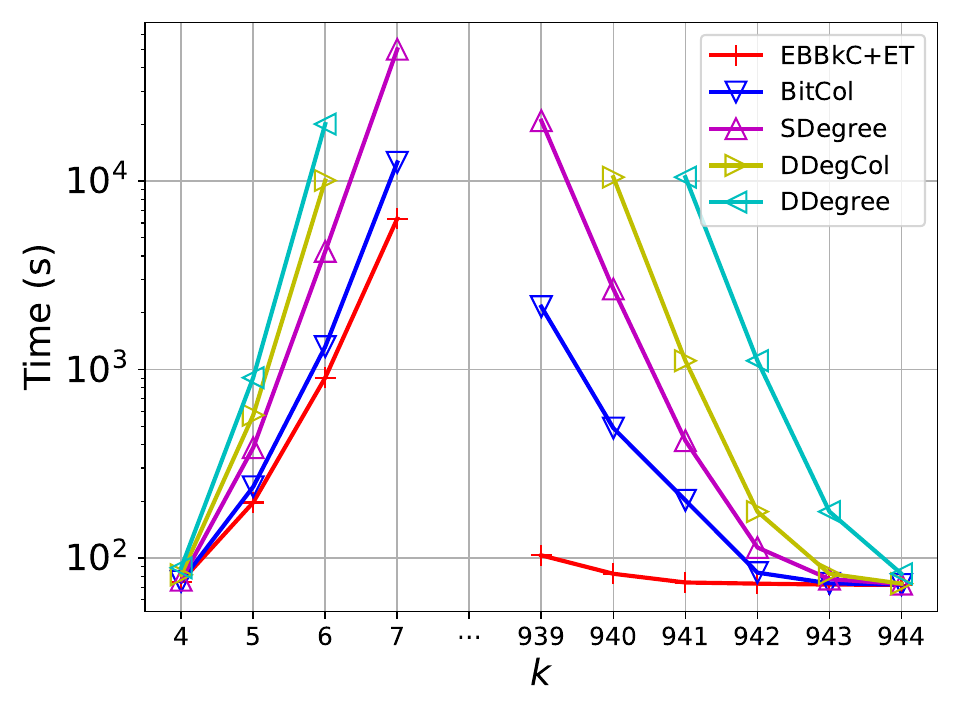}}
\subfigure[{\textsf{CW}}]{
\label{subfig:scale-cw}
    \includegraphics[width=0.32\textwidth]{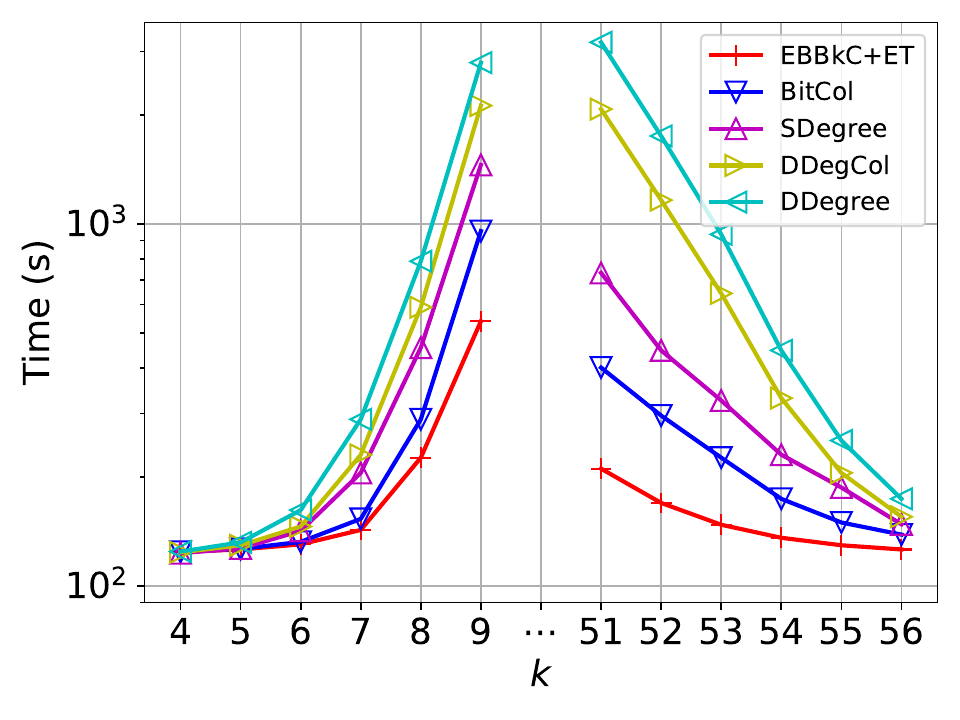}}
\subfigure[{\textsf{WP} }]{
\label{subfig:scale-wp}
    \includegraphics[width=0.32\textwidth]{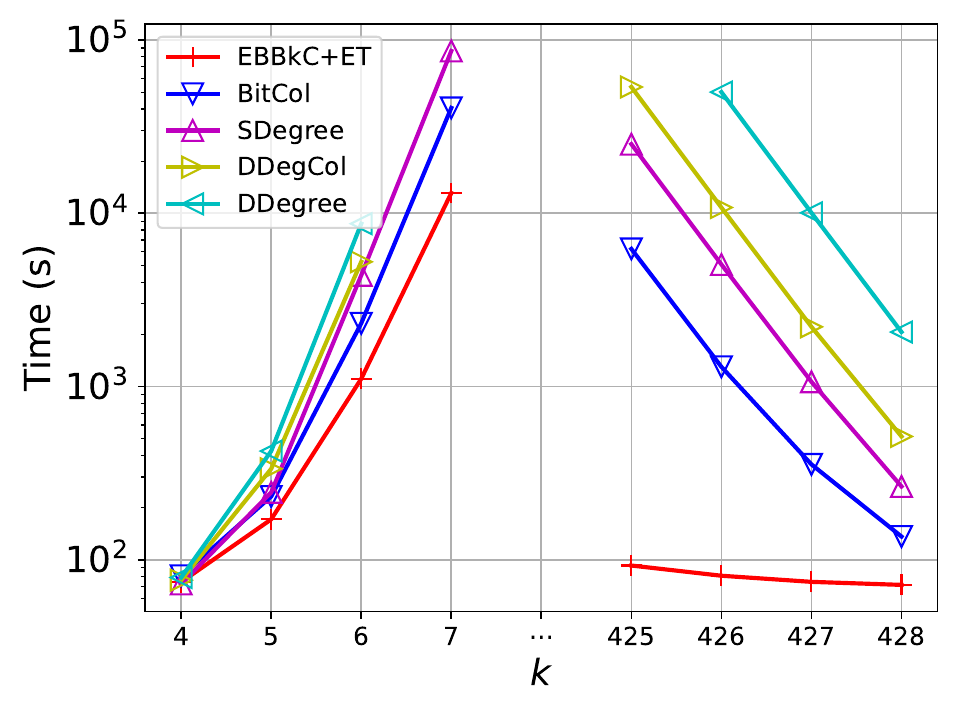}}
\vspace{-3mm}
\caption{{Results on scalability test. }} \label{fig:scale}
\end{minipage}
\vspace{-3mm}
\end{figure*}

\smallskip\noindent
\textbf{(5) Effect of the color-based pruning rules.}
Recall that in Section~\ref{subsec:EBBkC-C}, we introduce two color-based pruning rules, where the first rule is adapted from the existing studies \cite{li2020ordering} and the second rule is newly proposed in this paper. Therefore, we study the effect of the second pruning rule by making comparison between the running time of the algorithms with and without this rule, respectively. We denote the algorithm without this rule by \texttt{EBBkC(stc)+ET}. The results are shown in Figure~\ref{fig:color}, indicated by red lines and yellow lines. We observe that the second pruning rule brings more advantages as $k$ increases. {\cheng A} possible reason is that when $k$ is small, the graph instance $g$ usually has more than $l$ colors ($l\le k$), which cannot {\cheng be} pruned by the second rule while as $k$ increases, sparse graph instances can easily violate the rule and they can be safely pruned.





\smallskip\noindent
\textbf{{(6) Effects of early-termination technique (varying $t$).}}
%
We study the effects of choosing different parameter $t$ (in $t$-plex) for early-termination. We vary $t$ in the range $\{1,2,3,4,5\}$ and report the corresponding running time of \texttt{EBBkC+ET} under different values of $k$. The results are shown in Figure~\ref{fig:stop}. We have the following observations. First, for all values of $k$, \texttt{EBBkC+ET} with $t=2$ always runs faster than that with $t=1$, 
since we can list all $k$-cliques inside a $t$-plex in optimal time when $t=1$ and $t=2$ but we can early-terminate the recursion \textsf{EBBkC\_Rec} in an earlier phase when $t=2$ than that when $t=1$ (Section~\ref{subsec:2-plex}). Second, when the value of $k$ is small, \texttt{EBBkC+ET} with smaller $t$ runs faster while when the value of $k$ is large, \texttt{EBBkC+ET} with larger $t$ runs faster. On \textsf{WK}, for example, when $k=8$, \texttt{EBBkC+ET} with $t=2$ runs the fastest; when $k=12$ and $k=16$, \texttt{EBBkC+ET} with $t=4$ runs the fastest; and when $k=20$, \texttt{EBBkC+ET} with $t=5$ runs the fastest. 
This phenomenon conforms our theoretical analysis in Section~\ref{sec:early-termination} that when the value of $k$ increases, listing $k$-cliques with early-termination on sparser plexes, i.e., $t$-plex with larger $t$, can also be efficient.

\smallskip\noindent
\textbf{(7) Parallelization.}
%
{\cheng We compare different algorithms in a parallel computing setting.}
{\kaixin Specifically, for \texttt{EBBkC} framework, since each produced sub-branch produced from $B=(\emptyset, G, k)$ can be solved independently (see line 6 of Algorithm~\ref{alg:EBBkC-H}), we process all of such sub-branches in parallel.
For existing studies under \texttt{VBBkC} framework \cite{li2020ordering, danisch2018listing, yuan2022efficient}, there are two parallel schemes. One is called \texttt{NodeParallel} (\texttt{NP} for short). This strategy processes each produced sub-branch at $B=(\emptyset, G, k)$ in parallel since they are independent (see line 9 of Algorithm~\ref{alg:VBBkC}). The other is called \texttt{EdgeParallel} (\texttt{EP} for short). This strategy first aggregates the first two consecutive branching steps at $B=(\emptyset, G, k)$ as a unit, and as a result, it would produce $|E(G)|$ sub-branches, then it processes each sub-branch in parallel since they are independent \cite{li2020ordering, danisch2018listing, yuan2022efficient}.}
The results of the comparison are shown in Figure~\ref{fig:parallel}. 
%
{\cheng Consider the comparison between \texttt{VBBkC+ET(NP)} and \texttt{VBBkC+ET(EP)}},
indicated by blue lines and magenta lines. 
The algorithm with edge-level parallelization strategy achieves a higher degree of parallelism since it would produce a number of problem instances with smaller {\cheng and} similar scales, which balances the computational loads across the threads.
%
Then, we consider the comparison between \texttt{EBBkC+ET} and \texttt{VBBkC+ET(EP)}, indicated by red lines and blue lines. Both algorithms can be regarded as adopting edge-level parallelization strategy but differ in the ordering of the edges. We observe that \texttt{EBBkC+ET} runs faster than \texttt{VBBkC+ET(EP)}. The reason is that \texttt{EBBkC} uses truss-based edge ordering, which would produce even smaller problem instances than those of \texttt{VBBkC+ET(EP)}. 

{
\smallskip\noindent
\textbf{(8) Space costs.}
We report the space costs of different algorithms in Figure~\ref{fig:mem}. We have the following observations.
First, the space costs of all algorithms are comparable and a few times larger than the graph size, which is aligned with the space complexity of $O(n+m)$.
Second, \texttt{EBBkC+ET} has the space cost slightly larger than those of others since it employs extra data structures for maintaining edge ordering and conducting early-termination.}



{
\smallskip\noindent
\textbf{(9) Scalability test.}
We test the scalability of the algorithms on three large graphs under the parallel setting with 48 threads and report the results in Figure~\ref{fig:scale}.
All algorithms use the \texttt{EP} parallel scheme.
%
\texttt{EBBkC+ET} outperforms other baselines consistently. In particular, on the largest graph \textsf{WP}, \texttt{EBBkC+ET} is up to about 100$\times$ faster than  \texttt{BitCol} (when $k=425$). 
}

%% file: related.tex
\section{Related Work}
\label{sec:related}

\smallskip\noindent
\textbf{Listing $k$-cliques for arbitrary $k$ values.} 
Existing exact $k$-clique listing algorithms for arbitrary $k$ values can be classified into two categories: backtracking based algorithms \cite{chiba1985arboricity} and branch-and-bound based algorithms \cite{yuan2017effective, li2020ordering, yuan2022efficient}.
Specifically, the algorithm \texttt{Arbo} \cite{chiba1985arboricity} is the first practical algorithm for listing all $k$-cliques, whose time complexity ({\cheng we focus on the worst-case time complexies in this paper}) is $O(km\alpha^{k-2})$, where $m$ and $\alpha$ are the number of edges and the arboricity of the graph, respectively. However, {\cheng it is difficult to paralleize} \texttt{Arbo} 
since it involves a depth-first backtracking procedure. 
To solve this issue, several (vertex-oriented {\chengB branching-based}) branch-and-bound (BB) based algorithms are proposed, 
{\cheng including}
\texttt{Degree} \cite{ortmann2014triangle, finocchi2015clique}, \texttt{Degen} \cite{danisch2018listing}, \texttt{DegenCol} \cite{li2020ordering}, \texttt{DegCol} \cite{li2020ordering}, \texttt{DDegCol} \cite{li2020ordering}, \texttt{DDegree} \cite{li2020ordering}, \texttt{SDegree} \cite{yuan2022efficient} and \texttt{BitCol} \cite{yuan2022efficient}. As introduced in Section~\ref{sec:vbbkc}, these algorithms follow the 
{\cheng same}
{\chengB framework} but differ in how the vertex orderings are adopted {\cheng and some implementation details}.
Specifically, \texttt{Degree} uses a global degree ordering, with which {\cheng the size of} the largest problem instance produced 
can be bounded by $\eta$ (i.e., the $h$-index of the graph).
\texttt{Degree} has a time complexity of $O(km(\eta/2)^{k-2})$. \texttt{Degen} uses a global degeneracy ordering, with which {\cheng the size of} the largest problem instance produced 
can be bounded by $\delta$ (i.e., the degeneracy of the graph).
\texttt{Degen} has a time complexity of $O(km(\delta/2)^{k-2})$. Since it has been proven that $\delta\le 2\alpha-1$ \cite{zhou1994edge}, \texttt{Degen} is the first algorithm that outperforms \texttt{Arbo} theoretically and practically. However, \texttt{Degen} suffers from the issue that it cannot {\chengC efficiently} list the clique whose size is near $\omega$ (i.e., the size of a maximum clique). 
To solve this issue, the authors in \cite{li2020ordering} propose several {\cheng algorithms, which are based on} {\chengB color-based vertex orderings}. \texttt{DegCol} and \texttt{DegenCol} first color the graph with some graph coloring algorithms (e.g., inverse degree based \cite{yuan2017effective} and inverse degeneracy based \cite{hasenplaugh2014ordering}) and generate the an ordering of vertices {\cheng based on the color values of the vertices}.
{\cheng While the color values} of the vertices can significantly prune the unpromising search paths and make the algorithm efficient to list near-$\omega$ cliques, \texttt{DegCol} and \texttt{DegenCol} both {\cheng have} the time complexity of $O(km(\Delta/2)^{k-2})$, where $\Delta$ is the maximum degree of a vertex in the graph since the {\chengB color-based} ordering cannot guarantee a tighter size bound for the produced problem instance. 
To overcome this limitation, \texttt{DDegCol} adopts a hybrid ordering, where it first uses degeneracy ordering to branch the universal search space such that the size of each produced problem instance is bounded by $\delta$, then it uses {\chengB color-based} orderings to branch each produced sub-branch. Following a similar procedure, \texttt{DDegree} also combines degeneracy ordering and degree ordering to branch the universal search space and the sub-spaces, respectively. In this way, \texttt{DDegCol} and \texttt{DDegree} have the 
time complexity of $O(km(\delta/2)^{k-2})$. 
\texttt{BitCol} and \texttt{SDegree} implement \texttt{DDegCol} and \texttt{DDegree} in a more efficient way with SIMD instructions, respectively, 
{\cheng and}
retain the same time complexity as that of \texttt{DDegCol} and \texttt{DDegree}. 
{\cheng In contrast, our \texttt{EBBkC} algorithm is an edge-oriented {\chengB branching-based BB} algorithm based on edge orderings and has its time complexity (i.e., {\chengC $O(m\delta + km(\tau/2)^{k-2})$}) better than 
that of the state-of-the-art algorithms including \texttt{DDegCol}, \texttt{DDegree}, \texttt{BitCol} and \texttt{SDegree} (i.e., $O(km(\delta/2)^{k-2})$), where $\tau$ is strictly smaller than $\delta$. 
}
Some other algorithms, e.g., \texttt{MACE} \cite{makino2004new}, 
{\cheng which are originally designed for}
the maximal clique enumeration problem, can also be adapted to listing $k$-cliques problem \cite{makino2004new, schmidt2009scalable, takeaki2012implementation}. These adapted algorithms are mainly based on the well-known Bron-Kerbosch (BK) algorithm \cite{bron1973algorithm} with some size constraints to ensure that each clique to be outputted has exactly $k$ vertices. However, these adapted algorithm are even less efficient than \texttt{Arbo} theoretically and practically, e.g., \texttt{MACE} has a time complexity of $O(kmn\alpha^{k-2})$ \cite{schmidt2009scalable, li2020ordering}, 
{\cheng and cannot}
handle large real graphs. 

\smallskip
\noindent
\textbf{Listing $k$-cliques for special $k$ values.} There are two special cases for $k$-clique listing problem: when $k=3$ and when $k=\omega$. When $k=3$, the problem reduces to triangle listing problem \cite{chiba1985arboricity, latapy2008main, ortmann2014triangle}. The state-of-the-art algorithm for triangle listing problem 
follows a vertex ordering based framework \cite{ortmann2014triangle}, whose time complexity is $O(m\alpha)$, where $\alpha$ is the arboricity of the graph. When $k=\omega$, the problem reduces to maximum clique search problem \cite{ostergaard2002fast, pattabiraman2015fast, lu2017finding, chang2019efficient}. The state-of-the-art algorithm for maximum clique search problem first transforms the maximum clique problem to a set of clique finding sub-problems, then it conducts a branch-and-bound framework to iteratively check whether a clique of a certain size can be found in the sub-problem \cite{chang2019efficient}, whose time complexity is $O(n 2^n)$. 
{\cheng We note that these existing algorithms cannot solve the $k$-clique listing problem for arbitrary $k$'s.}

%% file: conclusion.tex
\section{Conclusion}
\label{sec:conclusion}
{\chengB In this paper, we study the $k$-clique listing problem, a fundamental graph mining operator with diverse applications in various networks. We propose a new branch-and-bound framework, named \texttt{EBBkC}, which incorporates an edge-oriented branching strategy. This strategy expands a partial $k$-clique using two connected vertices (an edge), offering new opportunities for optimization. Furthermore, to handle dense graph sub-branches more efficiently, we develop specialized algorithms that enable early termination, contributing to improved performance. We conduct extensive experiments on 19 real graphs, and the results consistently demonstrate \texttt{EBBkC}'s superior performance compared to state-of-the-art \texttt{VBBkC}-based algorithms.
In the future, {we plan to explore the potential applications of our \texttt{EBBkC} technique to other cohesive subgraph mining tasks including clique and connected dense subgraph mining}. In addition, we will explore the possibility of adopting our algorithms to list $k$-cliques' counterparts in other types of graphs such as bipartite graphs.
}
%

\begin{acks}
We would like to thank the anonymous reviewers for providing constructive feedback and valuable suggestions. This research is supported by the Ministry of Education, Singapore, under its Academic Research Fund (Tier 2 Award MOE-T2EP20221-0013, Tier 2 Award MOE-T2EP20220-0011, and Tier 1 Award (RG77/21)). Any opinions, findings and conclusions or recommendations expressed in this material are those of the author(s) and do not reflect the views of the Ministry of Education, Singapore.
\end{acks}

%% file: appendix.tex
\appendix


\section{Proof of Theorem~\ref{theo:ebbkc-t}}

{\Revise The running time of \texttt{EBBkC-T} {\chengC consists of} the time of generating the truss-based edge ordering, which is $O(\delta \cdot |E(G)|)$ \cite{che2020accelerating}, and the time of the recursive listing procedure (lines 6-10 of Algorithm~\ref{alg:EBBkC-T}). Consider the latter one.}
Given a branch $B=(S, g, l)$, we denote by $T(g, l)$ the upper bound of time cost of listing $l$-cliques under such branch. When $k\ge 3$, with different values of $l$, we have the following recurrences. 
\begin{equation}
\label{eq:ebbkc-recurrence}
T(g, l) \le \left\{
\begin{array}{lc}
    O(k\cdot |V(g)|) & l = 1\\
    O(k\cdot |E(g)|) &  l = 2 \\
    \sum_{e_i\in E(g)} \Big( T(g_i, l-2) + T'(g_i) \Big) & 3\le l\le k-2 \\
\end{array}
\right.
\end{equation}
where $T'(g_i)$ is the time for constructing $g_i$ given $B=(S, g, l)$ (line 9 {\chengB of Algorithm~\ref{alg:EBBkC-T}}). We first show the a lemma which builds the relationship between $g_i$ and $g$, then we present the complexity of $T'(g_i)$. 

\begin{lemma}
\label{lemma:sum-edge-bound}
Given a branch $B=(S,g,l)$ and the sub-branches $B_i=(S_i,g_i,l_i)$ produced at $B$. We have (1) when $l<k$, 
\begin{equation}
    \sum_{e_i\in E(g)} |E(g_i)| < \frac{\tau^2}{4} \cdot |E(g)|
\end{equation}
and (2) when $l=k$ (we note that the branch $B$ corresponds to the universal branch $B=(\emptyset, G, k)$),
\begin{equation}
    \sum_{e_i\in E(G)} |E(g_i)| < \frac{\tau^2}{2} \cdot |E|
\end{equation}
\end{lemma}

\begin{proof}
\smallskip
\noindent\textbf{Case $l < k$.} $E(g_i)$ can be obtained by checking {\chengB for} each edge in $\{e_{i+1}, \cdots, e_{|E(g)|}\}$ whether it is in $E[e_i]$. Clearly, we have $|E(g_i)|\le |E(g)| - i$. Then 
\begin{equation}
\label{eq:intermediate}
    \sum_{e_i\in E(g)} |E(g_i)| \le \sum_{i=1}^{|E(g)|} (|E(g)| - i) =\frac{|E(g)|(|E(g) - 1|)}{2} 
\end{equation}
According to Lemma~\ref{lemma:comparison}, each branch contains at most $\tau$ vertices, which indicates that $|E(g)|$ is at most $\tau(\tau-1)/2$. Therefore, 
\begin{equation}
    \frac{|E(g)|(|E(g) - 1|)}{2} \le \frac{\tau(\tau-1)}{4} \cdot (|E(g)-1|) < \frac{\tau^2}{4} \cdot |E(g)|
\end{equation}

\smallskip
\noindent\textbf{Case $l = k$.} According Lemma~\ref{lemma:comparison}, $V[e]\le \tau$. Then $E[e]$ contains at most $\tau(\tau-1)/2$ edges. Thus,
\begin{equation}
    \sum_{e_i\in E} |E(g_i)| \le \sum_{e_i\in E} \frac{\tau(\tau-1)}{2} < \frac{\tau^2}{2} \cdot |E|
\end{equation}
which completes the proof. 
\end{proof}
Consider $T'(g_i)$. When $l=3$, for each edge $e_i\in E(g)$, we just need to compute $V(g_i)$ by checking {\chengB for} each vertex in $V(g)$ whether it is in $V[e_i]$, which costs at most $O(\tau)$ {\chengB time} since $|V(g)|\le \tau$. When $l>3$, for each edge $e_i\in E(g)$, we need to compute both $V(g_i)$ and $E(g_i)$. Therefore, 
\begin{equation}
\label{eq:close-T-prime}
    \sum_{e_i\in E(g)} T'(g_i) = \left\{
\begin{array}{lc}
    O(\tau \cdot |E(g)|) & l = 3 \\
    O(\tau^2 \cdot |E(g)|) & l > 3 \\
\end{array}
\right.
\end{equation}
Given the above analyses, we prove the Theorem~\ref{theo:ebbkc-t} as follows.


\begin{proof}
We prove by induction on $l$ to show that 
\begin{equation}
\label{eq:T-close}
T(g,l) \le \lambda \cdot (k+2l) \cdot |E(g)| \cdot \big( \frac{\tau}{2} \big)^{l-2}
\end{equation}
where $\lambda$ is positive constant.
Since the integer $l$ decreases by 2 when branching, then when $k$ is odd (resp. even), $l$ would always be odd (resp. even) in all branches. Thus, we need to consider both cases. 

\smallskip\noindent
\textit{\underline{Base Case ($l=2$ and $l=3$).}} When $l=2$, it is easy to verify that there exists $\lambda$ such that $T(g,2) \le \lambda \cdot k \cdot |E(g)|$ satisfies Eq.~(\ref{eq:T-close}). When $l=3$, we put Eq.~(\ref{eq:close-T-prime}) into Eq.~(\ref{eq:ebbkc-recurrence}), {\chengB and} it is easy to verify that $T(g,3)\le \lambda \cdot k \cdot |E(g)|\cdot \tau$, which also satisfies Eq.~(\ref{eq:T-close}).


\smallskip\noindent
\textit{\underline{Induction Step.}} We first consider the case when $3<l<k$. 
Assume the induction hypothesis, i.e., Eq.~(\ref{eq:T-close}), is true for $l=p$ ($p+2<k$). When $l=p+2$, 
\begin{equation}
\begin{aligned}
& T(g,p+2) \le \sum_{e_i\in E({g})} \Big(  T(g_i,p) + T'(g_i) \Big) \\
&\le \lambda \cdot \tau^2 \cdot |E(g)| + \sum_{e_i\in E({g})}  \lambda \cdot (k+2p) \cdot |E(g_i)| \cdot \big( \frac{\tau}{2} \big)^{p-2} \\
&< \lambda \cdot \tau^2 \cdot |E(g)| + \lambda \cdot (k+2p) \cdot |E(g)| \cdot \big( \frac{\tau}{2} \big)^{p}\\
&\le \lambda \cdot (k+2p+4) \cdot |E(g)| \cdot \big( \frac{\tau}{2} \big)^{p}\\
\end{aligned}
\end{equation}
The first inequality derives from recurrence. The second inequality {\chengB derives from} Eq.~(\ref{eq:close-T-prime}) and the induction hypothesis. The third inequality derives from Lemma~\ref{lemma:sum-edge-bound} (case $l<k$). The fourth inequality {\chengB derives from} the fact that $\tau^2 \le 4\cdot ({\tau}/{2})^{p}$ when $p\ge 2$. 

Then consider the cases when $l=k$. 
\begin{equation}
\begin{aligned}
& T(G,k) \le \sum_{e_i\in E(G)} \Big( T(g_i, k-2) +T'(g_i) \Big)  \\
&\le \lambda \cdot \tau^2 \cdot E(G) + \sum_{e_i\in E(G)} \lambda \cdot (3k-4) \cdot |E(g_i)| \cdot \big( \frac{\tau}{2} \big)^{k-4} \\
&< \lambda \cdot \tau^2 \cdot E(G) + 2\lambda \cdot (3k-4) \cdot |E(G)| \cdot \big( \frac{\tau}{2} \big)^{k-2} \\
&< 6\lambda \cdot k \cdot |E(G)| \cdot  \big( \frac{\tau}{2} \big)^{k-2}
\end{aligned}
\end{equation}
The first inequality derives from recurrence. The second inequality {\chengB derives from} Eq.~(\ref{eq:close-T-prime}) and the induction hypothesis. The third inequality derives from Lemma~\ref{lemma:sum-edge-bound} (case $l=k$). The fourth inequality {\chengB derives from} the fact that $\tau^2 \le 4\cdot ({\tau}/{2})^{k-2}$ when $k\ge 4$.

\smallskip\noindent
\textit{\underline{Conclusion.}} We conclude that given a graph $G=(V,E)$ and an integer $k\ge 3$, the time complexity of \texttt{EBBkC} can be upper bounded by {\Revise $O(\delta \cdot |E(G)| + k\cdot |E(G)|\cdot ({\tau}/{2})^{k-2})$}. 
\end{proof}

\begin{figure}[t]
\subfigure[An example graph $G$.]{
    \label{subfig:g}
    \includegraphics[width=0.20\textwidth]{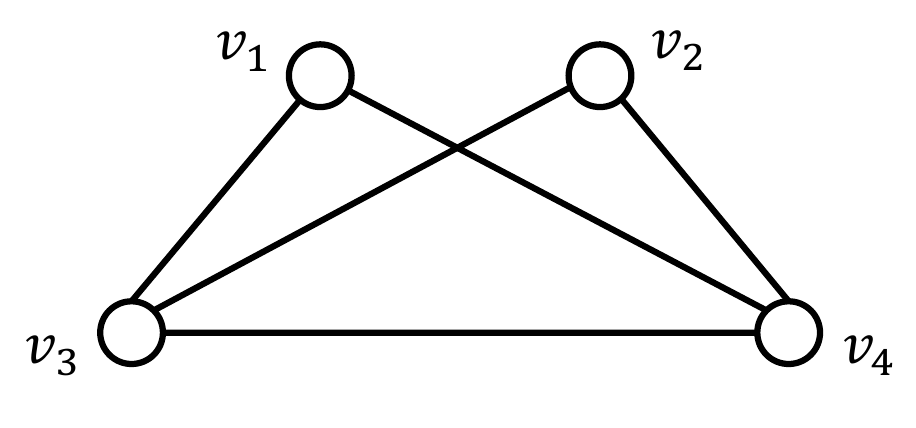}}
\subfigure[Branching with \texttt{VBBkC} on $G$.]{
    \label{subfig:example-vbbkc}
    \includegraphics[width=0.43\textwidth]{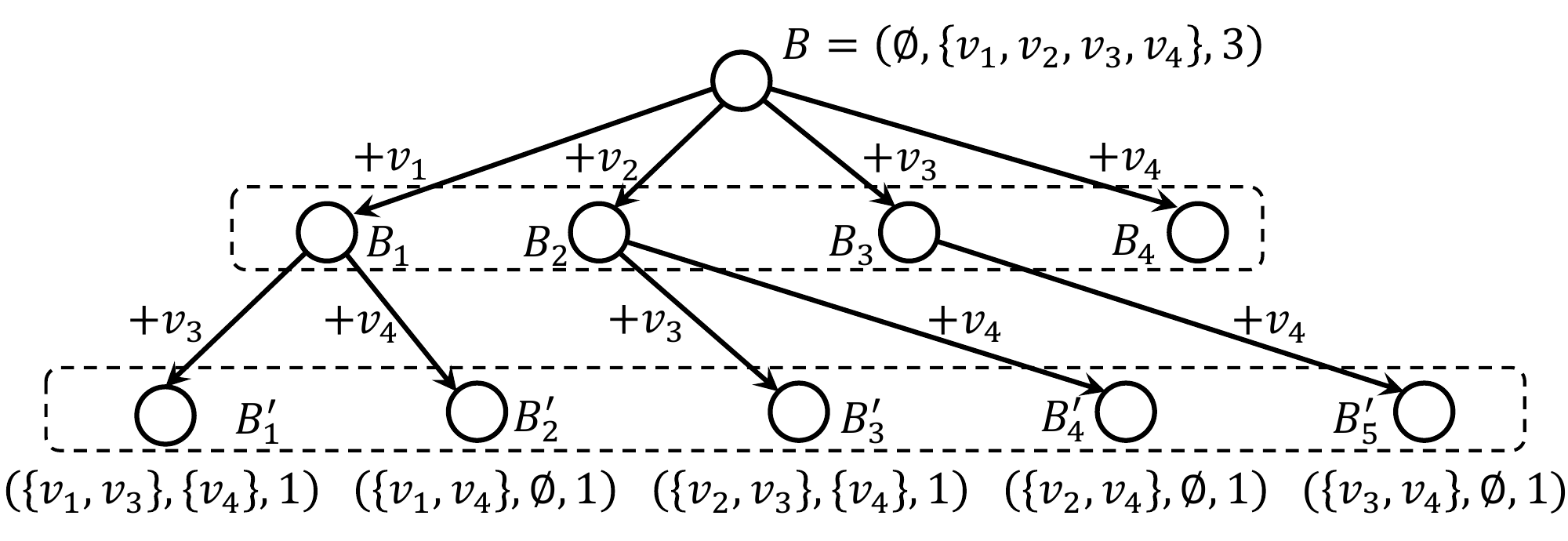}}
\subfigure[Branching with \texttt{EBBkC} that can generate the same branches as those with \texttt{VBBkC}.]{
    \label{subfig:example-ebbkc}
    \includegraphics[width=0.43\textwidth]{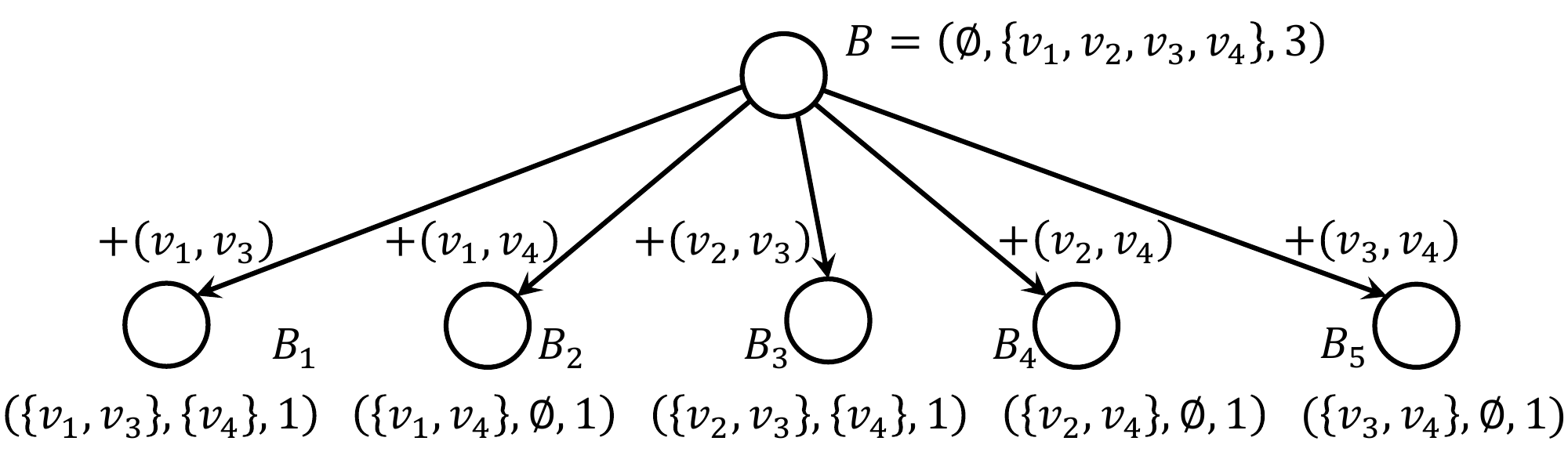}}
\subfigure[Branching with \texttt{EBBkC} that can generate the branches that \texttt{VBBkC} cannot generate.]{
    \label{subfig:example-ebbkc-2}
    \includegraphics[width=0.43\textwidth]{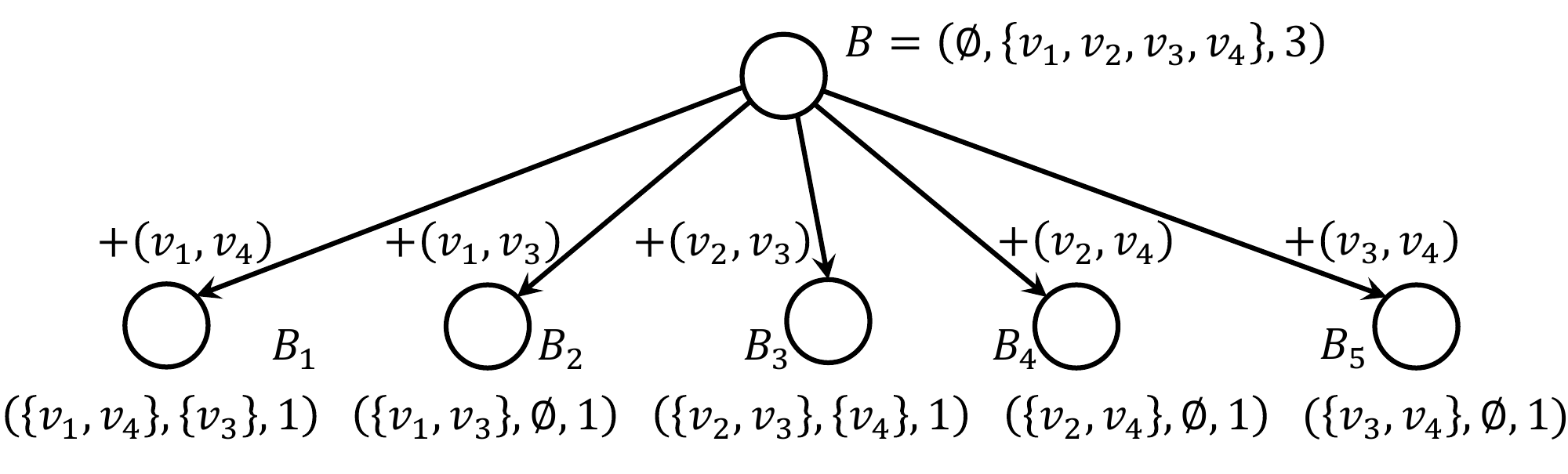}}
\vspace{-3mm}
\caption{Illustration of the advantages of \texttt{EBBkC} over \texttt{BBkC}.}
\label{fig:example}
\end{figure}

\section{A Counter-Example for the Statement in Section~\ref{subsec:EBBkC-T}}

\smallskip\noindent
\textbf{Example.} 
To illustrate the statement in Section~\ref{subsec:EBBkC-T}, we consider the example in Figure~\ref{fig:example}. The graph $G$ involves 4 vertices and 5 edges. Assume that we aim to list 3-cliques in $G$, i.e., $k=3$. 
Consider the degeneracy ordering of the vertices $\pi_{\delta}=\langle v_1, v_2, v_3, v_4 \rangle$. Then \texttt{VBBkC} would generate 5 branches following the vertex ordering {\chengB $\pi_{\delta}$}, and the illustration is shown in Figure~\ref{subfig:example-vbbkc}. Correspondingly, we can construct an edge ordering $\pi_e=\langle (v_1,v_3), (v_1,v_4), (v_2,v_3), (v_2,v_4), (v_3,v_4) \rangle$ by following which we can produce the same branches as that of \texttt{VBBkC}. Consider the truss-based edge ordering $\pi_{\tau}=\langle (v_1,v_4), (v_1,v_3), (v_2,v_3), (v_2,v_4), (v_3,v_4) \rangle$. The produced branches are shown in Figure~\ref{subfig:example-ebbkc-2}. We note that these produced branches cannot be generated with any vertex orderings. To see this, consider the branch $B_1$ and $B_3$ in Figure~\ref{subfig:example-ebbkc-2}. If there exists a vertex ordering $\pi_v$ that can generate $B_1$, then vertex $v_3$ must be ordered behind $v_4$ in $\pi_v$. Similarly, for branch $B_3$, vertex $v_4$ must be ordered behind $v_3$ in $\pi_v$. This derives a contradiction.

\begin{figure*}[t]
\begin{minipage}{0.48\linewidth}
\subfigure[\textsf{ST}]{
    \includegraphics[width=0.48\textwidth]{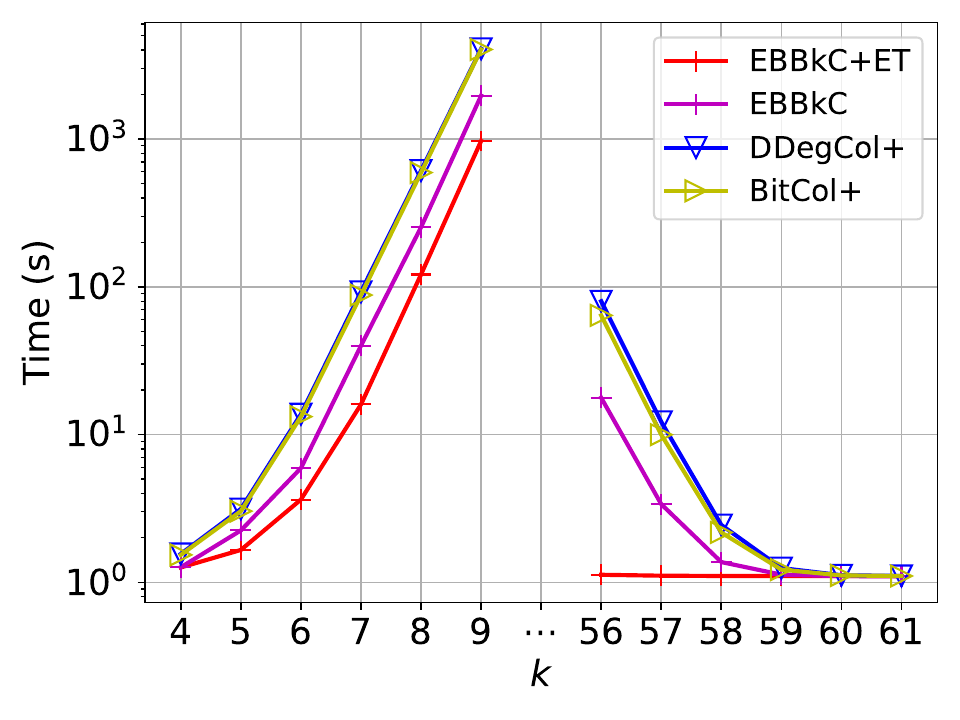}}
\subfigure[\textsf{OR}]{
    \includegraphics[width=0.48\textwidth]{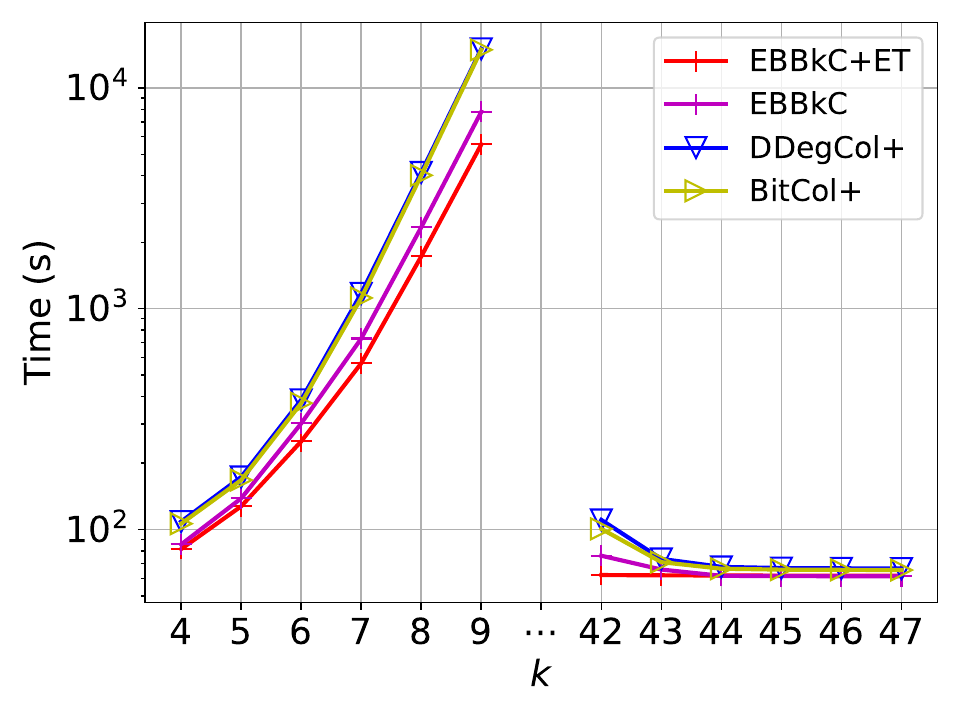}}
\caption{{Ablation studies.} }
\label{fig:variant-large}
\end{minipage}
\begin{minipage}{0.48\linewidth}
\subfigure[\textsf{ST}]{
    \includegraphics[width=0.48\textwidth]{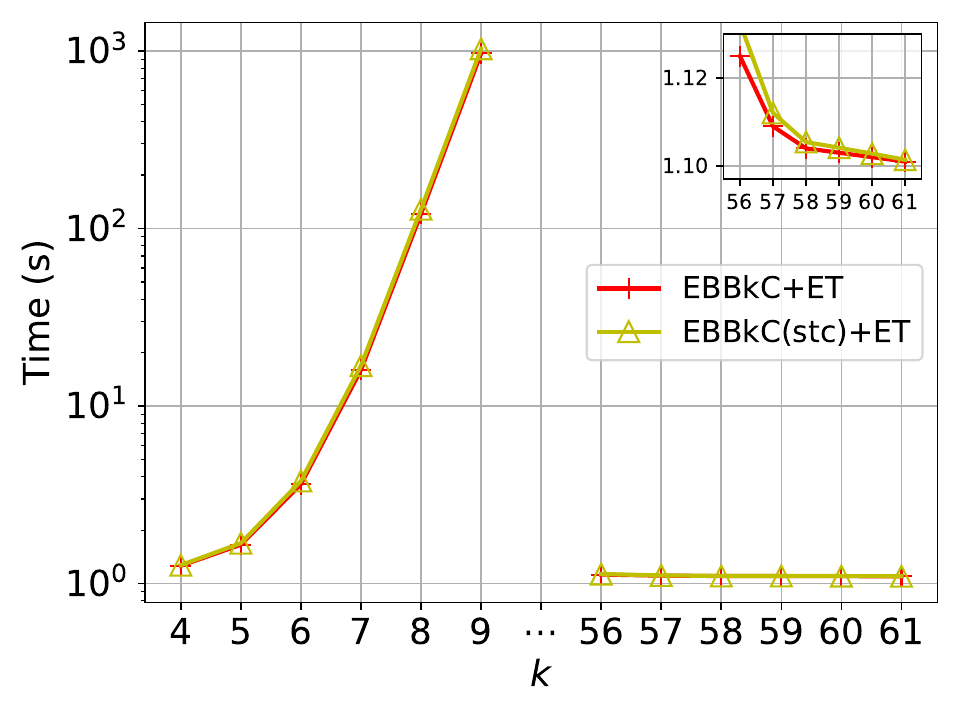}}
\subfigure[\textsf{OR}]{
    \includegraphics[width=0.48\textwidth]{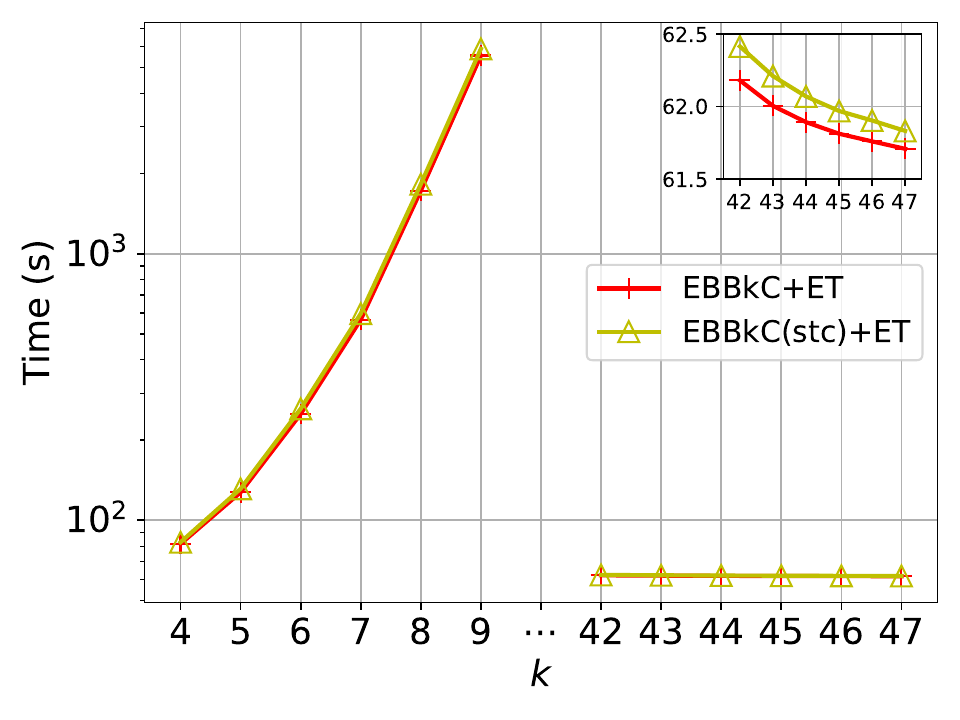}}
\caption{Effects of the color-based pruning rules (comparison between the algorithms w/ and w/o the Rule (2)).}
\label{fig:color-large}
\end{minipage}
\end{figure*}

\section{Time complexity of Algorithm~\ref{alg:k-plex}}




\begin{theorem}
\label{theo:k-plex}
Given a branch $B=(S,g,l)$ with $g$ being a $t$-plex ($t\ge 3$), \texttt{kCtPlex} lists all $k$-cliques within $B$ in at most $O(|E(g)|+ t \cdot \tbinom{|V(g)|)}{l}+k\cdot c(g,l))$ time, where $c(g,l)$ is the number of $l$-cliques in $g$.
\end{theorem}

\begin{proof}
Let $T(g,l)$ be the total time complexity. Lines 1-2 of Algorithm~\ref{alg:k-plex} can be done in $O(E(g))$ time for partitioning $V(g)$ and constructing the inverse graph $g_{inv}$. Let $T'(|C|, l')$ be the time complexity to produce all sub-branches excluding the output part (lines 5-10). We have $T'(0, \cdot) = 0$ and $T'(\cdot, 0)= 0$.
\begin{equation}
    T \le O(E(g)+k\cdot c(g,l)) + T'(|V(g)\setminus I|, l)
\end{equation}
According to Eq.~(\ref{eq:kCtPlex}), for each $1\le i\le |C|$, the size of the produced $C_i$ is at most $|C|-i$. Thus, we have the following recurrence. 
\begin{equation}
\begin{aligned}
\label{eq:TTTT}
    T'(|C|, l') &= \sum_{i=1}^{|C|} \Big( T'(|C_i|, l'-1) + O(t) \Big)  \\
    &= T'(|C|-1, l'-1) + O(t) + \sum_{j=0}^{|C|-2} \Big( T'(j, l'-1) + O(t) \Big) \\
    &= T'(|C|-1, l'-1) + T'(|C|-1, l') + O(t)
\end{aligned}
\end{equation}
We note that $O(t)$ is the time to filter out the vertices that are connected with $v_i$ in $g_{inv}$, i.e., $N(v_i, g_{inv})$. 
Let $T^*(|C|, l')=T'(|C|, l')+O(t)$ and apply it to the above equation, 
\begin{equation}
\label{eq:recurrence}
\begin{aligned}
    T^*(|C|,l') = T^*(|C|-1,l'-1) + T^*(|C|-1, l')
\end{aligned}
\end{equation}
with the initial conditions $T^*(\cdot, 0)=O(t)$ and $T^*(0, \cdot)=O(t)$. Observe that Eq.~(\ref{eq:recurrence}) is similar to an identity equation $\tbinom{n}{k}=\tbinom{n-1}{k} + \tbinom{n-1}{k-1}$. Then it is easy to verify that $T^*(|C|,l')=O\big(\tbinom{|C|}{l'}\cdot t\big)$. Therefore, $T'(|C|,l')$ can also be bounded by $O\big(\tbinom{|C|}{l'}\cdot t\big)$. Finally, since $|V(g)\setminus I|\le |V(g)|$, the time complexity of Algorithm~\ref{alg:k-plex} is at most $O(|E(g)|+ t \cdot \tbinom{|V(g)|)}{l}+k\cdot c(g,l))$. 
\end{proof}

\smallskip
\noindent\textbf{Remark.} Recall that in Eq.~(\ref{eq:kCtPlex}), we need to filter out the vertices that are connected with $v_i$ in $g_{inv}$ to create the sub-branch. There are two possible ways to implement this procedure. One way is to use a shared array to store $V(g)\setminus I$, denoted by $C$. Every time we execute line 11 and 12, the pointer in $C$ moves forward and marks those vertices in $N(v_i, g_{inv})$ as invalid, respectively. The benefit is that the procedure can be done in $O(t)$ in each round, as we show in Eq.~(\ref{eq:TTTT}), but will be hard to be parallelized. Another way is to use additional $O(C_i)$ to ensure that each sub-branch has its own copy of $C_i$, which makes the procedure easy to be parallelized. In our experiments, we implement the algorithm in the former way.

\section{Additional Experimental Results}

The experimental results on ablation studies and the effects of the color-based pruning rules (comparison between the algorithms with and without the Rule (2)) on datasets \textsf{ST} and \textsf{OR} are shown in Figure~\ref{fig:variant-large} and Figure~\ref{fig:color-large}, {\chengB respectively}.